%% file: main.tex
\documentclass[sigconf,nonacm]{acmart}

\usepackage{url}
\usepackage{latexsym}
\usepackage{enumerate}
\usepackage{booktabs}
\usepackage{array}
\usepackage{enumitem}
\usepackage{multirow}
\usepackage{graphicx}
\usepackage{framed}
\usepackage{makecell}
\usepackage{bbding}
\usepackage{amsmath, amsthm}
\usepackage{diagbox}
\newtheorem{theorem}{Theorem}
\newtheorem{lemma}[theorem]{Lemma}

\theoremstyle{definition}
\newtheorem{definition}{Definition}
\theoremstyle{remark}

\usepackage[flushleft]{threeparttable}
\usepackage[linesnumbered,boxed]{algorithm2e}

\usepackage{subcaption}

\usepackage[
lambda,
advantage,
operators,
sets,
adversary,
landau,
probability,
notions,
logic,
ff,
mm,
primitives,
events,
complexity,
asymptotics,
keys,
]{cryptocode}
\newcommand{\ignore}[1]{}

\newcommand{\dkgname}{Any-Trust DKG}
\newcommand{\name}{Scalable and Adaptively Secure Any-Trust Distributed Key Generation and All-hands Checkpointing}

\makeatother

\usepackage{enumitem}

\AtBeginDocument{%
	\providecommand\BibTeX{{%
			\normalfont B\kern-0.5em{\scshape i\kern-0.25em b}\kern-0.8em\TeX}}}

\setcopyright{acmlicensed}
\copyrightyear{2018}
\acmYear{2018}
\acmDOI{XXXXXXX.XXXXXXX}

\acmConference[Conference acronym 'XX]{Make sure to enter the correct
	conference title from your rights confirmation emai}{June 03--05,
	2018}{Woodstock, NY}
\acmISBN{978-1-4503-XXXX-X/18/06}

\author{\vspace{0.3cm} Hanwen Feng$^*$, Tiancheng Mai$^*$, and Qiang Tang$^*$}
\affiliation{\vspace{0.3cm}
  \institution{$^*$ School of Computer Science \\University of Sydney, Australia}
  \city{}
  \country
  {\tt\{hanwen.feng,tiancheng.mai,qiang.tang\}@sydney.edu.au}\vspace{0.5cm}
}

\title{\name\ }

\begin{document}
\fancyhead{}

\begin{abstract}
	The classical distributed key generation protocols (DKG) are resurging due to their widespread applications in blockchain. While efforts have been made to improve DKG communication, practical large-scale deployments are still yet to come due to various challenges, including the heavy computation and communication (particularly broadcast) overhead in their adversarial cases.
	In this paper, we propose a practical DKG for DLog-based cryptosystems, which achieves (quasi-)linear computation and communication per-node cost with the help of a common coin, even in the face of the maximal amount of Byzantine nodes. Moreover, our protocol is secure against adaptive adversaries, which can corrupt less than half of all nodes. The key to our improvements lies in delegating the most costly operations to an \textit{Any-Trust} group together with a set of techniques for adaptive security. This group is randomly sampled and consists of a small number of individuals. The population only trusts that at least one member in the group is honest, without knowing which one. Moreover, we present a generic transformer that enables us to efficiently deploy a conventional distributed protocol like our DKG, even when the participants have different \textit{weights}.
	Additionally, we introduce an extended broadcast channel based on a blockchain and data dispersal network (such as IPFS), enabling reliable broadcasting of arbitrary-size messages at the cost of constant-size blockchain storage. 
	
	Our DKG leads to a fully practical instantiation of Filecoin's checkpointing mechanism, in which {\em all} validators of a Proof-of-Stake (PoS) blockchain periodically run DKG and threshold signing to create checkpoints on Bitcoin, to enhance the security of the PoS chain. In comparison with the recent checkpointing approach of Babylon (Oakland, 2023), ours enjoys a significantly smaller cost of Bitcoin transaction fees. For $2^{12}$ validators,  our cost is merely 0.4\% of that incurred by Babylon's approach.
\end{abstract}

\maketitle

\input{data/intro}

\input{data/results.tex}

\input{data/tech}
\input{data/blocks}

\input{app/pre.tex}

\input{data/protocol}
\input{data/weight}
\input{data/broadcast}
\input{data/checkpoint}

\input{data/implement}

\input{data/related_work}

\input{data/conclusion}

\section*{Acknowledgements}
We thank Marko Vukoli\'{c} and Alejandro Ranchal-Pedrosa for the helpful discussions. This work was supported in part by Protocol Labs Research Grants under the RFP-012 on Checkpointing Filecoin onto Bitcoin.

	\bibliographystyle{ACM-Reference-Format}
	\bibliography{refer}

	\appendix
	\input{app/checkpoint-app.tex}

\end{document}

%% file: data/intro.tex

\section{Introduction}
Distributed key generation protocols (DKG) \cite{Pedersen91, GennaroJKR07} enable a set of participants to jointly generate a public key, and each of them outputs a secret key share. It is a classical topic and the basis of threshold cryptography. They were usually considered in small-scale in-house applications. There are recent resurge of interests of those protocols, mostly because of a diverse set of new blockchain applications, for example, cross-chain bridge \cite{LeeMDG23, CerulliCNPS23}, MEV protection \cite{MalkhiS22, Skale, Osmosis}, censorship resistance in asynchronous consensus \cite{MillerXCSS16, GuoL0XZ20}, checkpointing into Bitcoin \cite{AzouviV22}, and more. Those new applications raise a new fundamental challenge of deploying distributed key generation on a very large scale.

\smallskip\noindent{\bf Enhancing Proof of Stake Security via All-hands Checkpointing.} As one main motivational example of large-scale DKG, we elaborate on the checkpointing mechanism, which aims at addressing a prominent security challenge in the Proof-of-Stake (PoS) blockchain known as \textit{long-range attacks} \cite{BMS21}. At a high level, in a PoS blockchain, validators (with stakes) are in charge of proposing blocks; as time evolves, validators' secret keys could be leaked or simply sold (when all coins are spent) to the adversary, who can now easily create a fork from a historic block using the corresponding secret keys. For a newly joined node, such a fork is also considered valid. This long-range attack hints at an inherent vulnerability of ``revisionist history'' in PoS (and other resources such as space) blockchain.

One promising defense approach is to leverage proof-of-work (PoW) blockchains which are immune to such attacks, particularly to periodically let {\em the whole PoS network} (e.g., all validators) produce checkpoints and put them into Bitcoin. A recent work, Babylon \cite{TasTGKMY23}, let all validators generate checkpoints via multi-signatures and select some of them to post the checkpoints to Bitcoin. It follows that the number of Bitcoin transactions needed for each checkpoint grows linearly in the number of PoS validators (to at least put a bit vector for indicating corresponding public keys). Particularly, for a PoS chain with $2^{12}$ validators (e.g., Filecoin), the annual cost would be over 6 million USD (using the Bitcoin price retrieved on March 31st, 2024, and assuming checkpoints are created hourly). 

Instead, Filecoin proposed a blueprint for the checkpointing mechanism called Pikachu \cite{AzouviV22} via threshold Schnorr signature \cite{KomloG20, GennaroJKR07}. Specifically, all validators of a PoS chain need to run a DKG protocol for \textit{every epoch}, and the resulting public keys will serve as Bitcoin addresses. At epoch $i$, the validators from epoch $i-1$ will jointly create a Bitcoin transaction via a threshold Schnorr signing protocol, which contains the state of the PoS chain at epoch $i-1$ and transfers all coins from the address created in epoch $i-1$ to the newly created address. In doing so, the Bitcoin transactions uniquely decide the state of the PoS chain in each epoch, and they are verifiably endorsed by the majority of validators. Since this approach only needs exactly one Bitcoin transaction for each checkpoint, the Bitcoin transaction fees incurred are always a small constant regardless of the scale of the PoS blockchain network and much lower than Babylon's approach.

Despite being appealing, and recent progress on threshold Schnorr signatures \cite{KomloG20, CritesKM23, BellareCKMTZ22} could potentially be deployable, Pikachu remains in theory (or toy prototype) as all validators have to jointly run a DKG protocol for every checkpoint (as validator set evolves thus previous DKG cannot be reused). Particularly, even a moderate-scale PoS chain can have thousands of validators. Moreover, some applications need to be done in a ``timely'' manner, which makes the task more challenging. 
For example, Filecoin currently has around $2^{12}$ validators, with anticipated growth to $2^{14}$ validators in the future, while checkpoints may be supposed to be created hourly (as suggested in \cite{TasTGKMY23}). To deploy the Pikachu checkpointing into Bitcoin mechanism \cite{AzouviV22} for real-world blockchains (for example, Filecoin), we need to design a scalable DKG protocol that can be efficiently run among all validators in popular public blockchains.

\smallskip\noindent{\bf Existing DKGs are practically infeasible at a whole-chain scale.} Let us first briefly introduce the common paradigm for DKG protocols, which subsumes most DKG constructions, including \cite{Pedersen91, GennaroJKR07, KateZG10, TomescuCZAPGD20, ZhangXHSZ22}, to illustrate the astronomical communication and computation costs of existing DKGs in a large scale. 

In a nutshell, among $n$ participants where up to $t$ could be adversarial, each participant $P_i$ selects a $t$-degree polynomial $f_i$ to define $sk^{(i)}= f_i(0)$. They then {deliver} the share (as a dealer of a verifiable secret sharing (VSS) \cite{Pedersen91}) $sk_j^{(i)}= f_i(j)$ to other $P_j$ and \textit{broadcast} \footnote{Broadcast satisfies \textit{agreement}, i.e., all parties receive the same message even when the sender is malicious. Thus, it is more complicated and expensive than multicast.} a commitment, $\mathsf{com}_i$, for the polynomial $f_i(X)$. Then, each participant $P_j$ could verify if $sk_j^{(i)}$ is a valid share w.r.t. $\mathsf{com}_i$; If there are invalid shares, the participants will collectively engage in a {\em complaint} phase, where they \textit{broadcast} complaints and identify the set of qualified dealers $\mathsf{Qual} \subset [n]$, ensuring that all transmitted secret shares are valid. The final secret share for $P_i$ is $sk_i = \sum_{j\in \mathsf{Qual}} sk_i^{(j)}$, and the aggregate secret key is $sk = \sum_{j\in \mathsf{Qual}} sk^{(j)}$. 

The DKG scheme by Kate, Zaverucha, and Goldberg \cite{KateZG10} (and its recently improved version by Zhang et al. \cite{ZhangXHSZ22}, dubbed KZG hereafter) represents the state of the art following this paradigm. 
In KZG, the commitment to the polynomial $f_i$ is constant in size and enables validating a secret share with constant-sized information. In an optimistic case, when all participants are honest, KZG protocol can remain efficient even for large-scale deployment, as demonstrated in \cite{ZhangXHSZ22}. However, since there is a constant fraction of adversarial participants, its performance got dramatically worse.
Particularly, when another node complains about a dealer, the dealer shall \textit{broadcast} the corresponding secret share for public verification. Hence, in facing $O(n)$ malicious nodes, each complaining $O(n)$ nodes, there are $O(n^2)$ shares to be broadcasted, and every node shall verify these shares. Indeed, improving the adversarial case performance is the major open problem left by \cite{TomescuCZAPGD20}. 

We can readily anticipate that both communication and computation costs would skyrocket as the scale increases, as seen in scenarios like Filecoin validators. For instance, with $2^{12}$ participants in the DKG protocol, the entire network would need to transmit {\em tens of terabytes} of data, while each node would have to allocate {\em multiple hours} to verify shares (in response to complaints) to produce just one public key!\footnote{For detailed numerical estimates and comparisons, we refer to Sect.\ref{sect:estimatedformance}.}

We remark that publicly verifiable secret sharing (PVSS) can eliminate the complaint phase in DKG \cite{CascudoD17, FouqueS01,GentryHMNY21, BachoL23}, as each party could now broadcast all ``encrypted'' shares and enable the public verification immediately. However, existing PVSS-based DKGs are either even more costly than KZG \cite{FouqueS01,Groth21} or only generate group-element secrets, while mainstream threshold cryptographic protocols like threshold Schnorr signatures use field-element secrets. In addition, no existing PVSS schemes with field-element secrets are provably secure against adaptive attackers \cite{BachoL23}, while adaptive security is desired in practice\footnote{While there are two very recent PVSS-based DKG works \cite{BachoLLOP23, FengLT24} have further pushed down the asymptotical complexity of DKG, they are either only with group-element secrets or statically secure.}. 

Besides those high costs, one extra challenge may make things even worse: in PoS chains, validators usually have different \textit{weights} (proportional to the number of held stakes), while a threshold of weights is assumed to belong to honest validators. Simply viewing a validator as a participant in DKG can be leveraged by an adversary to amplify its power: the adversary can choose to corrupt many validators with small weights and eventually control the majority of DKG participants within its budget. 
A naive approach is to allocate different numbers of sub-IDs proportional to their weights. Each sub-ID is then treated as an independent participant. While this approach addresses the security concern, it can lead to an enormous number of sub-IDs. For example, we would need to allocate 674 trillion sub-IDs to 3700 Filecoin validators\footnote{\url{https://filfox.info/en/ranks/power}. It has a total mining power of around 25 EB, while the power unit is 32KB, so 674 trillion sub-IDs are needed. }.

It follows that the following question remains: 
\begin{center}
{\em Could we have a practically feasible DKG protocol that enables all-hands participation in blockchains with weighted validators, as well as being adaptively secure?}
\end{center}

%% file: data/results.tex

\subsection{Our Results}
In this article, we give an affirmative answer to this question. We proceed in two main steps:

\input{tabs/main-compare-tab.tex}

\smallskip\noindent{\bf A scalable and adaptively secure DKG.}
Our primary result is a practical DKG protocol (called \dkgname) for DLog-based cryptographic systems. We compare our scheme with state-of-the-art efficient DKG constructions\footnote{We focus on fully synchronous networks; asynchronous or partial synchronous DKGs \cite{GaoLLTXZ22, AbrahamJMMS23, DasYXMK022} are not included in the table, as
their implementations cannot simply leverage a broadcast channel and appear to be more expensive.} in Table \ref{tb:dkgcompare} and discuss more related works in Sect.\ref{ap:relatedworks}. Particularly, our DKG protocol features:
 
 \smallskip\noindent{\em -- Efficiency:} It enjoys (quasi-)linear (in $n$) per-node computation complexity\footnote{Although evaluating $O(n)$-degree polynomials at $O(n)$ points inherently causes $O(n\log n)$ computation, we only require $O(n)$ expensive group exponentiations.}, even in adversarial cases. Additionally, the size of {\em all} data to be broadcasted also only grows linearly in $n$. In contrast, previous constructions suffer from quadratic per-node computation and quadratic broadcast overhead. Specifically, as our experiments will demonstrate, for $n=2^{12}$, each node can complete all computation tasks in approximately $26$ seconds, with the data to be broadcasted totaling around $7.7$ MB in size. 

\smallskip\noindent \textit{-- Security:} Our protocol achieves optimal resilience and satisfies the security definitions achievable by classical DKGs. Specifically, in the presence of adaptive adversaries (who cannot retract messages sent by honest parties, as in many settings, such as Algorand \cite{GiladHMVZ17}), our scheme complies with the oracle-aided algebraic simulatability, which was recently introduced in  \cite{BachoL22} for capturing the adaptive security of many practical DKGs including Pedersen \cite{Pedersen91} and KZG.

We remark that our primary goal is to enable massive-scale DKG; we may use resources that are naturally available in the application settings. Compared with the classical DKG schemes, our protocol additionally leverages one common coin that is generated after all participants' public keys are determined (as the YOSO-model DKG \cite{BenhamoudaHKMR22}), which is available in most blockchains. Nonetheless, it enables a distributed randomness beacon \cite{ChoiMB23} for continuous coin generation, by applying threshold unique signatures with secret shares from the DKG \cite{CachinKS05} . 
 We introduce a set of techniques for the efficient, adaptively secure DKG (detailed in the next section).

\smallskip\noindent{\bf A generic transformer for weighted distributed protocols.} We present a generic sub-ID allocation mechanism that enables us to efficiently apply conventional distributed protocols in the weighted setting. Our sub-ID allocation mechanism deterministically decides the number of sub-IDs for each validator according to the weight distribution, such that every sub-ID will be viewed as an individual participant in the subsequent protocol.  

The trivial sub-ID allocation method that precisely preserves the portion of each validator's weight may need to issue tremendously many sub-IDs. In contrast, we have noticed and leveraged a gap between the usual assumption on the honest participant's weight ratio (assumed to be more than 2/3 due to other system components) and the honest ratio needed in threshold cryptography (usually just above 1/2). Particularly, our sub-ID allocation method is lossy-yet-qualified, guaranteeing that more than half of the sub-IDs will be issued to honest participants if they possess over 2/3 of the weights, and therefore it can be much more \textit{compact}. The number of sub-IDs issued by our method for $n$ validators is probably at most $2 n$, regardless of the weight distribution. For the real-world weight distributions of blockchain validators, the issued sub-IDs can even be fewer. For example, our method gives only 1688 sub-IDs instead of 674 trillion to the 3700 Filecoin validators. Compared with concurrent work in \cite{abs-2307-15561}, ours issues fewer sub-IDs for large validator sets like Filecoin's. More details can be found in Sect.\ref{sec:subid}.

\smallskip\noindent{\bf Implementation and Evaluation.} We implement our protocol in Java\footnote{Our code is available at \url{https://github.com/mtc2000/AnyTrustDKG}.}and deploy it on AWS EC2 instances with 16, 32, 64, 128, and 256 nodes. The results demonstrate that our protocol scales effectively and completes within a few seconds (adding some ledger waiting time, which could vary depending on the blockchain if we instantiate the broadcast channel via a distributed ledger.). Additionally, we conduct computational time tests for various values of $n$, ranging from $2^9$ to $2^{15}$. 
In comparison with the state-of-the-art DKG protocol KZG \cite{KateZG10}, our protocol's performance in both the good-case and worst-case scenarios is comparable to or even superior to KZG's performance in the good-case scenario, while KZG's cost in the adversarial case experiences a dramatic increase.
 Notably, for $n = 2^{12} \sim 2^{14}$, a node in our protocol can finish all computation tasks within around $26$ \textasciitilde $181$ seconds, even facing the maximal amount of Byzantine nodes. The total amount of data to be broadcasted is around $7.7$ \textasciitilde $30.6$ MB; If the nodes broadcast the data by posting it on the Filecoin blockchain, it takes around $5$ minutes \textasciitilde $20$ minutes. The experimental results show that our protocols effectively enable massive-scale DKG deployment.

\smallskip\noindent {\em Deployment friendliness:} It is worth noting our DKG is more friendly for large-scale deployment since our DKG makes exclusive use of multicast channels (besides broadcast channels), which can be efficiently implemented, e.g., with gossip protocols and does not require a node to know the IP addresses of all other peers. In contrast, both Pedersen DKG and KZG DKG require pair-wise {\em private} channels for their efficiency claims. While it is not infeasible for large-scale deployment, it does add extra difficulties and overheads, particularly in public blockchain settings.

\smallskip\noindent{\bf Application: better all-hands checkpointing.} We then apply our techniques to realize the checkpointing blueprint Pikachu of Filecoin \cite{AzouviV22} that requires all validators to participate. After our optimized sub-ID allocation, we execute a DKG and a threshold Schnorr signature \cite{Shoup23} among these 1688 sub-IDs to create a checkpoint. With our \dkgname, the DKG phase only incurs around 3MB of broadcast messages in total. Each node can complete all computations in just a few seconds, even when facing the maximum number of complaints. Regarding the threshold signature, we use \dkgname\ again to generate the nonce in the GJKR \cite{GennaroJKR07} signing protocol, resulting in a non-interactive threshold signing protocol (after nonce-generation), eliminating the potential single point of failures in coordinator-based protocols like FROST \cite{KomloG20}.

Compared with the existing checkpointing scheme Babylon \cite{TasTGKMY23}, in which the number of Bitcoin transactions per checkpoint grows linearly to the scale of the blockchain,  ours/Pikachu only requires exactly {\em one} Bitcoin transaction for each checkpoint. This difference is reflected in the monetary cost. As an example with Filecoin, the estimated  Bitcoin transaction fee incurred annually using Babylon would be over six million USD, while only 26,048.8 USD using ours/Pikachu.
More details can be found in Sect.\ref{sect:checkpointframework}.

%% file: tabs/main-compare-tab.tex

\begin{small}
	\begin{table*}[ht]
		\centering
		\caption{Comparison with the state-of-the-art DKGs for DLog-based Cryptography. }\label{tb:dkgcompare}
		\begin{tabular}{c|cc|c|c|cc}
			\multirow{2}{*}{\textsf{Schemes}} & \multirow{2}{*}{\textsf{Resilience}} & \multirow{2}{*}{\textsf{Adap.?}*} & \multicolumn{2}{c|}{\textsf{Comm. Cost (total)***}} & \multicolumn{2}{c}{\textsf{Comp. Cost (per node)**}}  \\
			& ~ &~& \multicolumn{1}{c}{\textsf{Good}\dag}&\textsf{Bad}\dag& \textsf{Good}\dag&\textsf{Bad}\dag \\
			\hline
			Pedersen \cite{Pedersen91}&$1/2$&\Checkmark&$O(n\mathcal{B}(n\secpar))$  + $O(n^2\secpar)$&-&$O(n^2)$&-\\
			KZG\cite{KateZG10, ZhangXHSZ22}& 1/2  &\Checkmark & $O(n\mathcal{B}(\secpar))$+ $O(n^2\secpar)$&$+O(n\mathcal{B}(n\secpar))$&$O(n\log n)$&$+O(n^2)$ \\
			\hline
			GHL \cite{GentryHL22}\ddag & $1/2$ &\XSolidBrush&	\multicolumn{2}{c|}{$O(n\mathcal{B}(n\secpar))$}&\multicolumn{2}{c}{$O(n^2)$} \\
			GJM+ \cite{GurkanJMMST21}\ddag & ${\log n}/{n}$ &\XSolidBrush&
			\multicolumn{2}{c|}{$O(n\mathcal{B}(\secpar)+ \log n\mathcal{B}(n\secpar))$}
			&\multicolumn{2}{c}{$O(n\log^2n)$}\\
			\hline
			BHK+\cite{BenhamoudaHKMR22} \S& $\approx 1/4$  &\Checkmark& $O(C\mathcal{B}(C\secpar))$ + $O(C\mathcal{M}(C^2\secpar))$&-&$O(C^3)$&-\\
			Ours (Sect.\ref{sect:dkg})\S& 1/2 &\Checkmark& $O(s\mathcal{B}(n\secpar))$ &+$O(n\mathcal{M}(s\secpar))$&$O(sn)$&-
		\end{tabular}
			{\\
				\raggedright *\textsf{Adap.?} asks if the protocol is adaptively secure, and we accept the relaxed definition from \cite{BachoL22}\\
			**\textsf{Comp.Cost} measures the number of group exponentiation operations performed by each node.\\
			***   \textsf{Comm.Cost} measures the total communication cost, where $\mathcal{B}(\ell)$ (or $\mathcal{M}(\ell)$) denotes the cost of one node broadcasting (or multicasting) $\ell$ bits to the network, and $O(\ell)$ means there are $O(\ell)$ bits sent by honest nodes over pair-wise channels.\\ 
		\dag For both communication and computation, \textsf{Good} considers the cost without complaints, \textsf{Bad} considers the \textbf{extra} cost when facing the maximal number of complaints; ``-" represents no asymptotically greater cost. \ddag GHL and GJM+ do not have a complaint phase.\\ 
		\S BHK+ and ours are committee-based approaches. For ensuring the quality of the committee with high probability (say $1-5\times10^{-9}$, as adopted by Algorand \cite{ChenM19}), BHK+ needs the committee size of { $C\approx 6000$} (estimated based on \cite{BenhamoudaG0HK020}), while ours only needs {$s=38$} (further analysis available in Table.\ref{tab:committee_sizes}). 
			\par}
	\end{table*}
\end{small}

%% file: data/tech.tex

\section{Technique Overview}\label{sect:tech}
We give a high-level overview of how we leverage various techniques to lead to our DKG protocol. Through the analysis of how DKG usually works, we observe one major reason for the inefficiency (both high communication and high computation costs) is due to the following simple facts: everyone {\em broadcast} shares to everyone, and broadcast channels are expensive!

\medskip\noindent{\bf Starting observation for efficiency: Selecting an any-trust group as VSS dealers.} 
Our starting observation is that letting all participants act as dealers of verifiable secret sharing (VSS) schemes is actually {\em unnecessary}. Recall that in the common DKG paradigm, the final secret is $sk =\sum_{j\in \mathsf{Qual}} s^{(j)}$, where $s^{(j)}$ is the secret dealt by a qualified participant ${P}_j$, and $\mathsf{Qual}\subset [n]$ is the set of all qualified dealers. However, the existence of {\bf one} honest dealer would suffice for both secrecy and robustness. Particularly for secrecy, a uniformly sampled secret $s^{(j)}$ contributed by an honest ${P}_j$ could conceal $sk$ to the adversary who may corrupt all other dealers. For robustness, even when the other dealers behave arbitrarily (e.g., go offline), one honest dealer ensures the set $\mathsf{Qual}$ is non-empty, and thus $sk$ is well-defined.

Therefore, we propose utilizing a small group of representatives, called an "any-trust" group (as introduced in the context of anonymous communication \cite{WolinskyCFJ12}), where we trust at least one member of the group is honest but do not need to know which one to trust. 
Note that such an any-trust group can be obtained by randomly sampling from the whole population (with an honest majority). Notably, the size $s$ of an any-trust group can be as small as a few tens in practice, which is in stark contrast to that of a group with an honest majority, which can be up to thousands. 

If focusing on \textit{static} adversaries who cannot corrupt parties during the protocol execution, the observation alone already leads to an efficient solution. Particularly, before the complaint phase, there are only $s$ commitments in total to broadcast and only $s$ secret shares for each party to verify. In the worst case, each party at most needs to verify $O(sn)$ shares, which is still feasible.

\medskip\noindent{\bf Challenges and techniques for adaptive security.} 
Achieving efficiency while preserving adaptive security is, however, non-trivial as the adversary does have the budget to corrupt the {\em entire} any-trust group.
Indeed, there are multiple difficulties from different layers, and we need different techniques to conquer them.

\smallskip\noindent\underline{\it Preventing the damage of corrupting the entire any-trust group.} We adopt the standard techniques from existing adaptively secure Byzantine agreement protocols \cite{GiladHMVZ17,DavidGKR18} to prevent the damage of entire any-trust group corruption. Particularly, we use the following techniques or assumptions.
 
 \smallskip\noindent{\it VRF-based sortition.} We use the verifiable random function (VRF) based sortition \cite{GiladHMVZ17} to select the any-trust group, such that only a party itself knows whether it has been selected, which prevents an adaptive adversary from targeting the group before the group members send out their first messages. 

  \smallskip\noindent{\it Memory erasure (assumption).} Once the any-trust group of dealers sends out messages, the adversary will be aware of their identities and proceed to corrupt them. We, therefore, require all these dealers to \textit{erase} all internal states related to dealing secrets (but not the long-term secret keys for signing and decryption) at the same time they send messages. Consequently, even when the adversary corrupts them, it cannot learn the secrets dealt by them.
  
  \smallskip\noindent {\it Forward-secure signatures.} However, the adversary can still violate the robustness and secrecy by sending different messages on behalf of newly corrupted dealers. In this case, an initially honest dealer may be disqualified by the network due to the disturbing messages sent by the adversary. To prevent such an attack, we apply forward-secure signatures \cite{DavidGKR18}, which ensures no further valid messages can be generated in this round after the dealer erases its secret states.

\smallskip\noindent\underline{\it Efficiently deciding the qualified dealer set with silent dealers.} Recall that in the conventional DKG schemes, including KZG, a dealer shall repudiate complaints (that he was silent) by broadcasting the corresponding shares for public verification. However, in our construction,  all dealers may be corrupted after the dealing phase. 
Also, for security, they have already erased all internal states and cannot repudiate anyway. We must enable the network to decide on the qualified set of dealers while the dealers remain silent. 

We tackle the problem by designing \textit{publicly verifiable complaints}, such that a dealer can be disqualified immediately once such a complaint against it has been presented, without the need for further repudiation from the dealer. There are two types of complaints: (1) the dealer does not send anything to the receiver. (2) the dealer sent an invalid share to the receiver. To make type (1) public verifiable, we let each dealer in the deal phase \textit{broadcast} the vector of {\em all} \textit{encrypted} shares under the receivers' public keys, such that everyone can check the existence of ciphertexts. \footnote{We remark that one can broadcast the vector of shares at a marginal cost increase compared to broadcasting one share by leveraging the effective broadcast extension trick, such as \cite{Nayak0SVX20}. 
} For type (2), we leverage {\em verifiable decryption}: if the decrypted share is invalid, the receiver can generate a NIZK proof showing the share is the correct decryption, which, together with the share itself, serves as a publicly verifiable complaint. 

\smallskip\noindent{\it Why not using PVSS?} We note that a publicly verifiable secret sharing (PVSS) scheme may look suitable for the setting with silent dealers, as the qualified set can be determined without the complaint phase. However, as we discussed before, existing PVSS schemes that produce field-element secrets are not adaptively secure. Moreover, our approach enables significantly better performance, due to the following reasons: (1) our approach incurs asymptotically lower verification cost for each node, as one verifiable complaint is sufficient to disqualify a malicious dealer, which means that a node  needs to verify correctness proofs for at most $O(s+n)$ decrypted values; In contrast, in a PVSS-based approach, a node has to verify the encryption correctness proofs for all $O(sn)$ encrypted shares. (2) Proving and verifying decryption correctness can be pratically more efficient. Our scheme can be instantiated with standard ElGamal encryption and Schnorr proof, while proving the validity of encrypted shares usually requires a special encryption scheme (e.g., Paillier \cite{LindellN18}, Lattice-based \cite{GentryHL22}) and a range proof \cite{GentryHL22}, which are considerably more expensive.  (3) In our approach a node is not required to prove or verify decryption corretness when there is no complaint, while proving and verifying encryption correteness are always mandatory in PVSS-based approaches.

\smallskip\noindent\underline{\it Simulating encrypted shares in the face of adaptive corruption.} Now an honest (selected) dealer needs to broadcast the sequence of encrypted shares $(\mathsf{Enc}(ek_1, f(1)), \mathsf{Enc}(ek_2, f(2)), \dots, \mathsf{Enc}(ek_n, f(n)))$, where $ek_i$ is the public encryption key of the party $P_i$, and $f$ is the secret polynomial such that $f(0)$ defines his secret. When an adversary corrupts $P_i$, it knows the decryption key $dk_i$ and thus the decrypted share $f(i)$. However, in the security proof, a simulator should not know $f(0)$ and all $f(i)$'s, while it needs to generate all ciphertexts to simulate an honest dealer. Under adaptive corruptions, we essentially need a non-committing public key encryption scheme \cite{BrunettaHS24}, which enables the simulator to generate valid ciphertexts without knowing the plaintexts and later open the ciphertext to an arbitrary value. However, general non-committing encryption is impossible in the standard model, unless the secret key is unresonably long \cite{Nielsen02, BrunettaHS24}. We may employ a public key encryption scheme in the random oracle model to circumvent this difficulty \footnote{Now we can see the choice that we do not prove the validity of encrypted shares is critical, as otherwise, we may not use random-oracle model PKE.}.
Particularly, let us think about the hybrid ElGamal encryption: we have $ek = g^x \in \mathbb{G}$, $dk = x \in \mathbb{Z}_p$, and the ciphertext $c$ in the form of $(g^r, \mathsf{Hash}(ek^r)\oplus m)$, where $g\in \mathbb{G}$ is the generator of the group $\mathbb{G}$ of prime order $p$, $r\in \mathbb{Z}_p$ is the fresh encryption randomness, $m$ is the plaintext, $\mathsf{Hash}$ is a hash function modeled as a random oracle, and $\oplus$ is the XOR operation on the message space (assuming binary encoded for simplicity). Then, the simulator could first generate a ciphertext as $(g^r, u)$, where $u$ is uniformly sampled from the plaintext space. Later, when the plaintext $m$ is known, the simulator programs the random oracle such that $\mathsf{Hash}(ek^r) = u \oplus m $, which opens the ciphertext to $m$.

\smallskip\noindent\underline{\it Preventing leakages due to publicly verifiable complaints.} We observe that our publicly verifiable complaints expose the decrypted results to the public, which, in some sense, provides a \textit{decryption oracle} and can potentially be leveraged by malicious nodes to break the confidentiality of the encryption scheme. We can patch this issue by employing a chose-ciphertext-attack (CCA) secure encryption scheme. However, as we already have many other requirements for the encryption scheme, we must be careful to ensure all requirements are compatible. For example, the encryption scheme must be non-committing and require programmable random oracles, which cannot coexist with standard CCA approaches like Naor-Yung \cite{NaorY90}. Meanwhile, our complaint phase needs efficient proof of decryption, which means the ciphertext must preserve some structures to enable efficient proof systems.

We use a signature of knowledge \cite{ChaseL06} to handle this issue. Specifically, for the hybrid ElGamal encryption whose ciphertexts are in the form of $(c_0 = g^r, c_1 = \mathsf{Hash}(ek^r)\oplus m)$, we require the dealer $P_i$ who produces $(c_0, c_1)$ to sign its ID $i$ using the knowledge of $r$
against $c_0 = g^r$. Then, in the security proof, the simulator could \textit{extract} $r$ from the signature of knowledge, which enables the simulator to know the encrypted share without the help of a decryption oracle. We will see other benefits of this approach when we detail the concrete encryption scheme. 

\smallskip\noindent{\bf Further optimizations to DKG}: After overcoming the difficulties of adaptive security, we turn back to optimizing the performance.

\smallskip\noindent\underline{\em Reducing communication of complaints by any-trust group again.} A straightforward complaint phase is to let all nodes directly broadcast their (verifiable) complaints to the network. In practice, it means there could be $O(n)$ broadcast again, which can incur an unpleasant overhead.
In addition, the cost cannot be reduced by our extended broadcast channel techniques either. Our broadcast technique enables one node to broadcast large-size messages, but now there are many senders.

We optimize the complaint phase via the following observation: one valid complaint is enough to disqualify a dealer, and thus, there is no need to include all complaints in the broadcast channel. We, therefore, design a complaint phase with the following three steps. First, each node disseminates the complaints using a multicast channel so that all nodes receive all complaints made by all other honest nodes. Second, we sample an any-trust group again, and let the group members deduplicate the complaints. Each group member will maintain a concise complaint list that contains at most one complaint for each dealer and all dealers complained by honest nodes. Finally, we let the group members broadcast their complaint lists, which guarantees that all malicious dealers will be disqualified. With the optimized complaint phase, there are $O(s)$ any-trust group members, each posting at most $O(s)$ complaints, where $s$ is the size of an any-trust group.

\smallskip\noindent\underline{\textit{On the choice of VSS/polynomial commitments.}} While the VSS scheme in KZG DKG \cite{KateZG10} is usually believed to be the most efficient instantiation, we do not use it in our DKG scheme due to the following considerations: (1) The polynomial commitment scheme in KZG VSS necessitates a structural common reference string, and securely establishing it in decentralized applications requires additional efforts.
(2) The communication benefits of the VSS scheme do not exist in our setting. Although its commitment size is constant, we need to broadcast all encrypted shares (and their encrypted proofs) anyway. 
(3) The generated public key is in a pairing-friendly group. We need to make an extra effort to adapt it for Schnorr signatures.
 
 Instead, we employ a VSS scheme based on a more classical polynomial commitment. Specifically, the commitment to a $t$-degree polynomial $f$ is the form of $g^{f(0)}, g^{f(1)}, \dots, g^{f(n)}$, where $n>t$ is the number shares needed to distribute. By the checking technique from Scrape \cite{CascudoD17}, a receiver could verify the $n+1$ group elements committing to a $t$-dgree polynomial at the cost of $O(n)$ group operations. Then, to verify each share $f(i)$, one just needs to perform one group exponentiation operation, such that verifying $O(n)$ shares from the dealer just costs $O(n)$ group operations, which guarantees a computationally efficient complaint phase.

 \smallskip\noindent\underline{\textit{Using multi-recipient encryption.}} For the PKE scheme, we have proposed the hybrid version of ElGamal, which is \textit{non-committing} and supports verifiable decryption. In our DKG, we use the multi-recipient variant of it \cite{BellareBKS07}, which reuses the $g^r$ component across ciphertexts under different public keys. It greatly reduces the broadcast cost, making the ratio of ciphertext size and share size close to 1. Moreover, recall that we use proof of knowledge of $r$ to prevent leakages from decryption oracles, and using multi-recipient encryption will only incur one proof of knowledge by each dealer.

\smallskip\noindent{\bf Optimization to broadcast channels: A practical extension trick.} 
 Two primary approaches for broadcast channels include using Byzantine broadcast (BB) protocols \cite{GiladHMVZ17, ChenM19, DolevR82} or utilizing existing infrastructure like blockchains. Implementing a large-scale BB protocol can be intricate and susceptible to errors; thus, using established blockchains is an attractive, simpler, and modular alternative. However, on-chain storage is generally an expensive and scarce resource. While the broadcast cost in our DKG for thousands of participants has been reduced to a few Megabytes, it can still be a considerable burden for blockchains.
 
 Therefore, we present a practical extension to a blockchain-based broadcast channel by leveraging a multicast channel and a data dispersal network (DNN) like IPFS \cite{TrautweinRTCSSG22}. Our design is simple and modular, retaining the major benefits of using blockchain, and it enables a sender to broadcast an {\it arbitrarily long} message while incurring \textit{constant} on-chain storage cost. Though it may be folklore to write digests alone into a blockchain to save bandwidth, we are unaware of any design with a formal agreement guarantee. We believe this component may be of independent interest. More details are in Sect.\ref{sec:bc}

%% file: data/blocks.tex
\section{Model and Goal}\label{sect:pre}
\smallskip\noindent{\bf Communication model. } We assume the network is synchronous, and protocols proceed by rounds. Every participant has access to multicast and broadcast channels with different guaranteed delivery time. They both achieve \textit{validity}, while broadcast channel additionally guarantees \textit{agreement}.

\noindent\underline{\it Validity.} When an honest node sends a message via this channel, all honest nodes can receive this message by the end of the round.

\noindent\underline{\it Agreement.} At the end of a broadcast round, honest receivers always receive the same message from this channel, even when the sender is Byzantine.

\smallskip\noindent{\bf Adversarial model.}  Prior to protocol execution, every node honestly generates their public key/secret key pairs and sends public keys to all other nodes. After the setup, the adversary can adaptively corrupt any node during the protocol execution and control their subsequent behaviors. Particularly, the adversary controls what messages a corrupted node will send in the same round it gets corrupted. However, messages already multicasted or broadcasted by node $i$ before $i$ become corrupted \textit{cannot be retracted}. 

\smallskip\noindent{\bf Notations and assumptions.} Throughout the paper: We use $\secpar$ to represent the security parameter. The notation $[i,n]$ represents the set $\{i,i+1,\cdots,n \}$, where $i$ and $n$ are integers with $i<n$. We might abbreviate $[1,n]$ simply as $[n]$. For a set $\{x_1, x_2, \dots, x_n\}$ and a sequence $(x_1, x_2,\dots, x_n)$, we may abbreviate them as $\{x_i\}_{i\in[n]}$ and $(x_i)_{i\in[n]}$, respectively.  A function $f(n)$ is deemed negligible in $n$, denoted by $f(n) \leq \text{negl}(n)$, if for every positive integer $c$, there exists an $n_0$ such that for all $n > n_0$, $f(n) < n^{-c}$. For a set $\mathbb{X}$, the notation $x \sample \mathbb{X}$ signifies sampling $x$ uniformly from $\mathbb{X}$. 
 We use $y \leftarrow A(x_1,x_2,\cdots)$ to represent running $A$ with inputs $x_1,x_2,\cdots$ and uniform randomness and outputting $y$.
Adversaries are assumed to be probabilistic polynomial time (PPT).

\smallskip\noindent{\bf Distributed Key Generation (DKG):} An $(n,t)$-DKG for DLog-based cryptography is an interactive protocol involving $n$ parties. At the end of execution, all honest parties possess a common public key $pk \in \mathbb{G}$ and a list of public key shares $(pk_1, \dots, pk_n)$, while each of them holds a secret share $sk_i \in \mathbb{Z}_p$. This setup allows any subset of $t+1$ honest parties to reconstruct the secret key $sk$ of $pk$. 


We follow the \textit{oracle-aided algebraic simulatability}, which was recently proposed by Bacho and Loss \cite{BachoL22} for capturing the adaptive security of many practical DKG schemes. This definition focuses on \textit{algebraic} adversaries.

\begin{definition}[Algebraic Algorithm]
    An algorithm $\mathsf{A}$ is called algebraic over a group $\mathbb{G}$ if all group element $\zeta \in \mathbb{G}$ that $\mathsf{A}$ outputs, it additionally outputs a vector $\vec{z} = \{z_0, \dots, z_m \}$ of integers in $\mathbb{Z}_p$ such that $\zeta = \prod_i g_i^{z_i}$, where $(g_1, \dots, g_m)$ is the list of group elements that $\mathsf{A}$ has received so far. 
\end{definition}

\begin{definition}\label{def:oracleaided}
    Let $\Pi$ be a protocol among $n$ parties $P_1, {P}_2, \dots, {P}_n$ where ${P}_i$ outputs a secret key share $sk_i$, a vector of public key shares $(pk_1, \dots, pk_n)$, and a public key $pk$. $\Pi$ is a secure DKG for a DL cryptosystem over a group $\mathbb{G}$ of a prime order $p$ if it satisfies the following properties.
\end{definition}
\begin{itemize}
    \item \textbf{Consistency:} $\Pi$ is $t$-consistent if although at most $t$ parties have been corrupted, the honest parties can output the same public key $pk$ and the same vector of public key shares $(pk_1, \dots, pk_n)$.
    \item \textbf{Correctness:} $\Pi$ is $t$-correct, if despite that at most $t$ parties have been corrupted, there is a $t$-degree polynomial $f(x)\in \mathbb{Z}_p[X]$, such that for every $i \in [n]$, $pk_i = g^{f(i)}$, every honest $\mathcal{P}_i$ has $sk_i = f(i)$, and the public key is $pk = g^{f(0)}$.
    \item \textbf{Oracle-aided Algebraic Simulatability:} $\Pi$ has $(t, k, T_\adv, T_{\mathsf{sim}})$-oracle-aided algebraic simulatability if for every adversary $\adv$ that runs in time at most $T_\adv$ and corrupts at most $t$ parties, there exists an algebraic simulator $\mathsf{Sim}$ that runs in time at most $T_{\mathsf{sim}}$, makes $k-1$ queries to oracle $\mathsf{DL}_g(\cdot)$ and satisfies the following properties:
    \subitem On input $\zeta = (g^{z_1}, \dots, g^{z_k})$, $\mathsf{Sim}$ simulates the role of the honest participants in an execution of $\Pi$. Upon an honest party $\mathcal{P}_i$ being corrupted, the simulator needs to return the internal state of $\mathcal{P}_i$ to the adversary.
    \subitem On input $\zeta = (g^{z_1}, \dots, g^{z_k})$, let $g_i$ denote the $i$-th query by $\mathsf{Sim}$ to $\mathsf{DL}_g(\cdot)$. Let $(\hat{a}_i, a_{i,1}, \dots, a_{i,k})$ be the corresponding algebraic coefficients of $g_i$, \textit{i.e.}, $g_i = g^{\hat{a}_i} \prod_{j=1}^{k}(g^{z_j})^{a_{i,j}}$ and set $(\hat{a}, a_{0,1}, \dots, a_{0,k})$ as the algebraic coefficients of $pk$. Then, the following matrix over $\mathbb{Z}_p$ is invertible
    \begin{equation*}
        L := \left(
        \begin{aligned}
            &a_{0,1}& a_{0,2} & \cdots & a_{0,k} \\
            &a_{1,1}& a_{1,2} & \cdots & a_{1,k} \\
            &\vdots & \vdots &  & \vdots \\
            &a_{k-1,1}& a_{k-1,2} & \cdots & a_{k-1,k} \\
        \end{aligned}
        \right).
    \end{equation*}
    Whenever $\mathsf{Sim}$ completes a simulation of an execution of $\Pi$, we call $L$ the simulatability matrix of $\mathsf{Sim}$.
    \subitem Denote by $\mathsf{view}_{\adv, y, \Pi}$ the view of $\adv$ in an execution of $\Pi$ conditioned on all honest parties outputting $pk =y$. Denote by $\mathsf{view}_{\adv, \zeta, y, \mathsf{Sim}}$ the view of $\adv$ when interacting with $\mathsf{Sim}$ on input $\zeta$, conditioned on $\mathsf{Sim}$ outputting $pk = y$. Then, for all $y$ and all $\zeta$, $\mathsf{view}_{\adv, y, \Pi}$ and $\mathsf{view}_{\adv, y, \Pi}$ are computationally indistinguishable.
\end{itemize}
Note that the adversary $\adv$ does not have to be fully algebraic. Instead, being algebraic related to $pk$ and queries $\mathsf{DL}_g(\cdot)$ would suffice, as discussed in \cite{BachoL22}.

Additionally, we consider ``key-expressibility'', as introduced in \cite{GurkanJMMST21}, against static attackers. This property is suitable for instantiating the key generation of re-keyable primitives like BLS and ElGamal. It also works with Schnorr signature as recently shown in \cite{GurkanJMMST21, Shoup23}. Formal definitions are provided in Appendix \ref{app:dkgdef}.

\begin{definition}[Key-expressability \cite{GurkanJMMST21}]\label{def:ke}
    A DKG protocol is key-expressable if for every static PPT adversary $\adv$ that corrupts up to $t$ nodes, there exists a PPT simulator $\mathsf{Sim}$, such that on input of a uniformly random element $pk'\in \mathbb{G}$, produces $\alpha \mathbb{Z}_p$, $sk_1\in\mathbb{Z}_p$, $pk_1 = g^{sk_1}\in \mathbb{G}$, and a view which is indistinguishable from $\adv$'s view from a run of the DKG protocol that ends with $pk = pk'^{\alpha} \cdot pk_1$.
\end{definition}

%% file: app/pre.tex

\section{Preliminaries}\label{app:pre}

\noindent{\bf Verifiable Random Function.}
A verifiable random function (VRF) is a pseudorandom function whose outputs can be publicly verified using the evaluator's public key. Throughout this paper, we take the DDH-based VRF scheme from \cite{GoldbergNPR16} as our instantiation. A VRF scheme consists of three algorithms: (1) $\mathsf{KeyGen}(\secparam)$ generates a verification key $vk$ and the secret evaluation key $sk$; (2) $\mathsf{Eval}(vk, sk, x)$ evaluates the function with $sk$ on the input $x$, and outputs $y$ along with a proof $\pi$. (3) $\mathsf{Verify}(vk, x,y, \pi)$ verifies if $y$ is the honest evaluation output on $x$ with the secret key of $vk$. 

A secure VRF satisfies (1)\textit{Pseudorandomness}: the function values are pseudorandom, even given the public key and the proofs; (2) \textit{Completeness}: it always holds that $\mathsf{Verify}(vk, x, \mathsf{Eval}(vk,sk,x)) =1$; and (3) \textit{Uniqueness}: it is infeasible to generate a public key $vk$, an input $x$, and two different $(y_1, \pi_1)$ and $(y_2, \pi_2)$, such that $$\mathsf{Verify}(vk, x, y_1, \pi_1) = \mathsf{Verify}(vk, x, y_2, \pi_2)=1.$$ (4) \textit{Unpredictability under malicious key generation}: if the input $x$ has enough entropy (i.e., cannot be predicted), then the correct output $y$ is indistinguishable from a uniformly random value, no matter how the VRF keys are generated. Formal definitions can be found in \cite{MicaliRV99, BadertscherGKRZ18}.

\smallskip\noindent\underline{\it VRF-based sortition.} We introduce the standard VRF-based sortition scheme below and will use it as a black box in our DKG protocol. (1) $\mathsf{Setup}(\secparam).$ Each user generates their VRF key pair $(vk, sk)$ and publishes $vk$. A public randomness $\textsf{rand}$ is sampled independent of the key generation. 
(2) $\mathsf{Sortition}(vk, sk, \textsf{rand}, \textsf{event}, \textsf{ratio})$. A user with $(vk ,sk)$ evaluates the VRF on the input of $(\textsf{rand}\|\textsf{event})$ and obtains $y$ and a proof $\pi$. It checks if $\frac{y}{\textsf{max}} \leq \textsf{ratio}$. If failed, abort. Otherwise, return $(y,\pi)$ as the credential of being selected.
(3) $\mathsf{Vrfy}(vk, \textsf{rand},\textsf{ratio}, \textsf{event}, \textsf{credential}).$ It verifies the credential by validating the VRF output $y$ and checking if $\frac{y}{\textsf{max}} \leq \textsf{ratio}$.

In above, $\textsf{ratio}$ denotes the ratio of the expected committee size to the whole group size, and the expected committee size is determined by the expected ratio of honest nodes to the committee. 

\noindent{\it Security of VRF-based sortition.} It is easy to argue when the underlying VRF satisfies the security properties defined above, and $\textsf{rand}$ is sampled independently of the key generation, the VRF-based sortition outcome is computationally indistinguishable from the outcome of a process where each user is elected with an independent probability of $\textsf{ratio}$. Our further analysis is based on this fact.

\smallskip\noindent{\bf Forward-secure digital signature.} A forward-secure signature scheme $\mathsf{FS}.\Sigma$ consists of four algorithms: (1) $\mathsf{Gen}(\secparam)\rightarrow (\mathsf{FS}.vk, \\\mathsf{FS}.sk[1])$ generates a verification key and the initial signing key;\\(2)$\mathsf{Update}(\mathsf{FS}.sk[i])\rightarrow \mathsf{FS}.sk[i+1]$ updates the signing key at round $i$ to the signing key at round $i+1$; (3) $\mathsf{Sign}(\mathsf{FS}.sk[i], m)\rightarrow \sigma$ generates a signature $\sigma$ for the message $m$; (4)  $\mathsf{Vrfy}(\mathsf{FS}.vk, i, \sigma, m)\rightarrow b$ determines if $\sigma$ is a valid signature for $m$ created by the signing key at round $i$.  

A forward-secure signature scheme guarantees the unforgeability of signatures at rounds $i< i^*$, even when the adversary has access to signing oracles at any round and corrupts the signing key at the $i^*$-th round. The formal security definitions and secure instantiations are available in \cite{ItkisR01}.

\smallskip\noindent{\bf Multi-recipient encryption.} We use the multi-recipient hybrid ElGamal encryption.
Let $g$ be a generator of $\mathbb{G}$, and let $\mathsf{Hash}$ be a hash function whose output space is the message space (which we assume is binary encoded and $\oplus$ is the XOR operation). The encryption scheme can be described as follows: (1) $\mathsf{Gen}(\secparam)$ outputs $(ek = g^x, dk = x)$, where $x\sample \mathbb{Z}_p$. (2) $\mathsf{MREnc}(ek_1, \dots, ek_n, m_1, \dots, m_n)$ outputs the ciphertexts $(c_0, c_1, \dots, c_n)$, where $c_0 = g^r$ for some uniformly sampled $r\in \mathbb{Z}_p$, $c_i = \mathsf{Hash}(ek_i^r) \oplus m_i$ for $i \in [n]$. (3) $\mathsf{Dec}(ek_i, dk_i, c_0, c_i)$ outputs $m = \mathsf{Hash}(c_0^{dk_i})\oplus c_i$.

We use the above algorithms in our DKG construction, but we directly reduce our DKG to the underlying DDH assumption without going through the security abstraction of the encryption scheme. This is because the security properties we need from the encryption are non-standard (as we discussed in the technique overview), and we would like to avoid further distractions.

\smallskip\noindent{\bf Signature of Knowledge.} A signature of knowledge (SoK) scheme is defined w.r.t. an NP relation $R$. It can be either in the common reference string model or in the random oracle model. We focus on the random-oracle-model instantiations and thus omit the algorithm for generating the common reference string. 
A user with the witness $x$ of a public statement $y$ such that $(y,x)\in R$ can sign any message $m$ via the signer algorithm $\mathsf{SoK.Sign}(y, x, m)\rightarrow \sigma$. Later, another user can verify if $m$ was signed by someone with the knowledge of the witness w.r.t. $y$ via the verifier algorithm. 

A secure SoK scheme satisfies the following properties: (1) \textit{Simulatability}: there is an efficient simulator algorithm that can produce a valid signature under any statement $y$ without using the witness $x$, and the produced signatures are indistinguishable from honestly generated signatures. (2) \textit{Extractability}: There is an efficient extractor algorithm that can extract the witness $x$ from a valid signature produced by the adversary under a statement $y$, even when the adversary has seen some simulated signatures. In the random oracle model, both the simulator and extractor are allowed to program the random oracle. Formal definitions are available in \cite{ChaseL06}.

In this paper, we use an SoK in the random oracle model for the DLog relation, \textit{i.e.}, for $y\in \mathbb{G}$ and $x\in \mathbb{Z}_p$, we have $(y,x)\in R$ iff $y = g^x$. Note that such an SoK is well-studied and can be instantiated with Schnorr signature scheme. 

\smallskip\noindent{\bf NIZK.} The proof of decryption used in our DKG is a NIZK proof system for the decryption correctness. In the random oracle model, a NIZK for an NP relation $R$ consists of a prover algorithm $\mathsf{Prove}$, which on inputs a statement $y$ and its witness $x$ outputs a proof $\pi$, and a verifier algorithm $\mathsf{Vrfy}$ which validates the proof against the statement $y$.

In this work, we require the NIZK to satisfy the following properties: (1) \textit{Completeness}: For every $(y, x)\in R$, it holds that $\mathsf{Vrfy}(y, \\ \mathsf{Prove}(y, x)) =1$. (2) \textit{Zero-knowledgeness}: There exists an efficient simulator algorithm that can produce a valid proof for any statement $y$ without knowing $x$, and the simulated proofs are indistinguishable from honestly generated proofs. (3) \textit{Simulation soundness.} For a statement $y$, if there is no witness $x$ s.t. $(y,x)\in R$, then it is infeasible for an efficient adversary to produce a valid proof for $y$, even when the adversary has seen simulated proofs. Formal definitions are available in \cite{SantisCOPS01}.

We use a NIZK for proof of decryption correctness w.r.t. the encryption scheme. It consists of a prover algorithm, $\mathsf{PKE.Prove}$, which on inputs $(c_0, c_i, ek_i,dk_i, m)$ produces a proof $\Gamma$, and a verifier algorithm $\mathsf{PKE.Vrfy}(c_0, c_i, ek_i, m, \Gamma)$ which checks whether $m$ is the correct decryption from $(c_0, c_i)$. Particularly, $\Gamma =(m, c_0^{dk_i}, \pi)$. Here $\pi$ demonstrates the discrete logarithm of $c_0^{dk_i}$ w.r.t $c_0$ is equal to that of $ek_i$ w.r.t. $g$, which is commonly known as DLEQ proof (equality of discrete logarithms) \cite{ChaumP92}.

\smallskip\noindent{\bf Scrape's polynomial commitment.} We use the polynomial commitment scheme from Scrape \cite{CascudoD17}. Particularly, let $g$ be the generator of $\mathbb{G}$ of prime order $p$, let $f$ be a $t$-degree polynomial over $\mathbb{Z}_p$, and let $n>t$ be an integer. Then, the commitment to the polynomial $f$ is 
$$
    (\mathsf{cm}_0 = g^{f(0)}, \mathsf{cm}_1 = g^{f(1)},\dots, \mathsf{cm}_n = g^{f(n)}).
$$

One can check whether these group elements commit to a $t$-degree polynomial by performing the following steps: (1) Sample an $(n-t)$ -degree polynomial $q(X)\in \mathbb{Z}_p[x]$, and compute the dual code: $ \mathsf{cm}_\tau^{\perp} = \frac{q(\tau)}{\prod_{j =0, j\ne \tau}^{n} (\tau -j )}, \forall \tau\in[0,n]$. (2) Check whether $\prod_{\tau =0}^{n} (\mathsf{cm}_{\tau})^{\mathsf{cm}_\tau^\perp} = \mathbf{1}_{\mathbb{G}}$, where $\mathbf{1}_{\mathbb{G}}$ is the identity element of $\mathbb{G}$.

It is worth noting that the first step can be reused to check different commitments. The effectiveness of the checking process is determined by the following lemma.
\begin{lemma}[\cite{CascudoD17}]
    For any $(\mathsf{cm}_0, \mathsf{cm}_1,\dots, \mathsf{cm}_n)\in \mathbb{G}^{n+1}$, if $$\prod_{\tau =0}^{n} (\mathsf{cm}_{\tau})^{\mathsf{cm}_\tau^\perp} = \mathbf{1}_{\mathbb{G}},$$ then with an overwhelming probability there exists a $t$-degree polynomial $f$, such that $cm_i = g^{f(i)}$ for $i \in [0,n]$.
\end{lemma}
After that, a share $f(i)$ can be validated by checking whether $g^{f(i)}$ equals $\mathsf{cm}_i$.

%% file: data/protocol.tex

\section{Our DKG Protocol}\label{sect:dkg}
 Following the technique overview in Sect.\ref{sect:tech}, we present our \dkgname\ in Sect.\ref{sect:dkgconstruction} based on the building blocks in Sect.\ref*{app:pre}, and analyze it in Sect.\ref{sect:dkganalysis}.

\begin{figure*}[!htb]
    \begin{pchstack}[boxed,center]
        \begin{pcvstack}
            \procedure[linenumbering]{$\textbf{Round 1 (broadcast): }\text{each } P_i \pcdo$:}{%
            \text{\color{blue} // determine whether it is elected as a dealer.}\\
                \mathsf{VRF.Sortition}(rvk_i, rsk_i, \textsf{rand}, \text{``deal"}, \textsf{ratio})\rightarrow \mathsf{CR}_i^{\mathsf{deal}}\\
                \pcif \mathsf{CR}_i^{\mathsf{deal}} = \perp, \text{\color{blue} // if not, update FS secret key and exit the round}\\
                \pcthen \mathsf{FS.Update}(\mathsf{FS}.sk_i[1])\rightarrow \mathsf{FS}.sk_i[2], \textbf{erase } \mathsf{FS}.sk_i[1], \textbf{exit } \textbf{Round 1}\\
                \text{\color{blue} // only elected users continue the followings.}\\
                \text{sample }(a_0, a_1, \dots, a_t)\sample \mathbb{Z}_p^{t+1}, \text{ define }f(X)= \sum_{\tau =0}^ta_\tau X^{\tau}\\
                \text{commit to the random polynomial $f(X)$: } (\mathsf{cm}_j = g^{f(j)})_{j\in[0,n]} \\   
                \text{encrypt shares: }\mathsf{PKE.MREnc}((ek_i)_{i\in [n]}, (f(i))_{i\in [n]})\rightarrow (c_0,\dots, c_n)\\
                \text{sign the ID $i$ using the knowledge of $r$: } \mathsf{SoK.Sign}(c_0,r, i)\rightarrow \sigma_{\mathsf{DL}}\\
                \text{denote } \mathsf{trans}_i[1]\leftarrow (\mathsf{CR}_i^{\mathsf{deal}}, (\mathsf{cm}_j)_{j\in [n]}, (c_j)_{j\in [n]}, \sigma_{\mathsf{DL}})\\
                \text{sign the transcript using FS: } \mathsf{FS.Sign}(\mathsf{FS}.sk_i[1], \mathsf{trans}_i[1])\rightarrow \sigma_i\\
                \text{Updat the secret key of FS: }\mathsf{FS.Update}(\mathsf{FS}.sk_i[1])\rightarrow \mathsf{FS}.sk_i[2]\\
                \textbf{erase } \mathsf{FS}.sk_i[1], f(X), (f(i))_{i\in[0,n]}, \text{ and encryption randomness}\\
                \textbf{broadcast } (i, \mathsf{trans}_i[1], \sigma_i[1] )
            }
            \pcvspace
            \procedure[linenumbering]{$\textbf{Round 2 (multicast): }\text{each } P_i \pcdo$:}{%
                \textbf{receive:} \{(j, \mathsf{trans}_{j}[1], \sigma_{j}[1]) \}_{j\in \mathbb{D}}, \text{ for }\mathbb{D}\subset[n]\\
                \text{set }\mathbb{D}_1, \mathbb{D}_2, \mathbb{D}_3, \mathbb{C} = \emptyset\\
                \text{\color{blue} // prepare dual code for verifying polynomial commitments}\\
                \text{sample an } (n-t)\text{-degree polynomial }q(X)\in \mathbb{Z}_p[x], \text{ compute}\\
                \t[1] \mathsf{cm}_\tau^{\perp} = \frac{q(\tau)}{\prod_{j =0, j\ne \tau}^{n} (\tau -j )}, \forall \tau\in[0,n] \\
                \pcfor j\in \mathbb{D}  \text{\color{blue} // verify each broadcast transcript as below}\\
                \t[1] \text{\color{blue} // ignore the transcript if the FS signature is invalid}\\
                \t[1] \pcif \mathsf{FS.Vrfy}(\mathsf{FS}.vk_j, 1, \sigma_j[1], \mathsf{trans}_j[1])  =0, \pcthen \textbf{continue}\\
                \t[1] \text{\color{blue} // verify if it is in a good format and the sortition credential}\\
                \t[1] \pcif \text{parse }\mathsf{trans}_j[1] = (\mathsf{CR}_j^{\mathsf{deal}}, (\mathsf{cm}^{(j)}_{\tau})_{\tau\in [n]}, (c^{(j)}_{\tau})_{\tau\in [n]}, \sigma_{\mathsf{DL}}^{(j)}) \text{ failed}\\
                \t[2] \vee \mathsf{VRF.Vrfy}(rvk_j, \textsf{rand}, \textsf{ratio},``\textsf{deal}",\mathsf{CR}_j^{\mathsf{deal}})=0\\
                \t[2] \text{\color{blue} //check if $(\mathsf{cm}^{(j)}_{\tau})_{\tau\in [n]}$ commits to a $t$-degree polynomial}\\
                \t[2] \vee  \prod_{\tau =0}^{n} (\mathsf{cm}_{\tau}^{(j)})^{\mathsf{cm}_\tau^\perp} \ne \mathbf{1}_{\mathbb{G}} \vee \mathsf{SoK.Vrfy}(c^{(j)}_0, \sigma_{\mathsf{DL}}^{(j)},j) = 0 \\
                \t[1] \pcthen \mathbb{D}_1= \mathbb{D}_1 \cup \{j\}  \text{\color{blue} // if any fails, disqualify $j$ immediately}\\
                \t[1] \pcelseif \mathsf{PKE.Dec}(dk_i, c_i^{(j)})=sk_i^{(j)}\wedge g^{sk_i^{(j)}}\ne \mathsf{cm}_i^{(j)}\\
                \t[1]  \text{\color{blue}//generate a complaint, and update the complaint list}\\
                \t[1] \pcthen \mathsf{PKE.Prove}(c^{(j)}_0, c^{(j)}_i,dk_i,sk_i^{(j)})\rightarrow \Gamma_j, \mathbb{D}_2= \mathbb{D}_2\cup \{(j, \Gamma_j)\}\\
                \t[1]  \text{\color{blue}//otherwise, update the candidate output list}\\
                \t[1] \pcelse \mathbb{D}_3= \mathbb{D}_3\cup\{j\}, \mathbb{C} =\mathbb{C}\cup \{(j,  ((\mathsf{cm}^{(j)}_{\tau})_{\tau\in [0,n]}, sk_i^{(j)}) )\}\\
                \mathbb{D}_2\rightarrow \mathsf{trans}_i[2], \mathsf{FS.Sign}(\mathsf{FS}.sk_i[1], \mathsf{trans}_i[2])\rightarrow \sigma_i[2]\\
                \pcif \mathbb{D}_2\ne \emptyset, \pcthen \textbf{multicast } (i, \mathsf{trans}_i[2], \sigma_i[2])
            }
        \end{pcvstack}
        \pchspace
        \begin{pcvstack}
            \procedure[linenumbering]{$\textbf{Round 3 (broadcast): }\text{each } P_i \pcdo:$}{%
                \textbf{receive }\{(j, \mathsf{trans}_j[2], \sigma_j[2])\}_{j\in\mathbb{R}_1}, \text{ for }\mathbb{R}_1\subset[n]\\
                 \text{\color{blue} // determine whether it is elected for broadcasting complaint list}\\
                \mathsf{VRF.Sortition}(rvk_i, rsk_i, \textsf{rand}, \text{``agree"}, \textsf{ratio})\rightarrow \mathsf{CR}_i^{\mathsf{agree}}\\
                \pcif \mathsf{CR}_i^{\mathsf{agree}} = \perp,  \text{\color{blue} // if not, update FS secret key and exit the round}\\
                \pcthen \mathsf{FS.Update}(\mathsf{FS}.sk_i[1])\rightarrow \mathsf{FS}.sk_i[2], \textbf{exit } \textbf{Round 2}\\
                \text{set } \mathsf{DisQual}, \mathsf{CompList} = \emptyset  \text{\color{blue} // start to deduplicate complaints}\\
                \pcfor j\in \mathbb{R}_1,  \pcif \mathsf{FS.Vrfy}(\mathsf{FS}.vk_j, 1, \sigma_j[2], \mathsf{trans}_j[2])=1\\
                \t \pcthen \pcfor (k, \Gamma_k)\in \mathsf{trans}_j[2]  \text{\color{blue} // check every complaint by $P_j$}\\
                \t[2] \text{\color{blue}// if $D_k$ has not been disqualified, verify the complaint}\\
                \t[2] \pcif k\notin \mathsf{DisQual}\text{, parse }\Gamma_k = (sk_j^{(k)}, \cdot)\\
                \t[2] {\color{blue}// c_0^{(k)}, c_j^{(k)} \text{ are what $P_i$ received at  line 1 of round 2}}\\
                \t[3]  \pcif \mathsf{PKE.Vrfy}(c_0^{(k)}, c_j^{(k)}, \Gamma_k)=1 \wedge g^{sk_j^{(k)}} \ne \mathsf{cm}_j^{(k)}\\
                \t[3] \text{\color{blue} // disqualify the dealer given a valid complaint}\\
                \t[3] \pcthen \mathsf{DisQual}= \mathsf{DisQual}\cup \{k\}\\
                \t[4]\mathsf{CompList} = \mathsf{CompList}\cup \{(k,\Gamma_k)\}\\
                \text{\color{blue} // stop verifying complaints from $j$ if the complaint is invalid.}\\
                \t[3]\pcelse \textbf{break}\\
                ( \mathsf{CR}_i^{\mathsf{agree}},\mathsf{CompList})\rightarrow \mathsf{trans}_i[3]\\
                \text{sign } \mathsf{FS.Sign}(\mathsf{FS}.sk_i[2], \mathsf{trans}_i[3])\rightarrow \sigma_i[3]\\
                \mathsf{FS.Update}(\mathsf{FS}.sk_i[2])\rightarrow \mathsf{FS}.sk_i[3]; 
                \textbf{erase } \mathsf{FS}.sk_i[2]\\
                 \pcif \mathsf{ComList}\ne \emptyset, \pcthen \textbf{broadcast } (i, \mathsf{trans}_i[3], \sigma_i[3])
            }
            \pcvspace
            \procedure[linenumbering]{$\textbf{At the end of Round 3: }\text{each } P_i \pcdo:$}{%
                \textbf{receive }\{(j, \mathsf{trans}_j[3], \sigma_j[3])\}_{j\in\mathbb{R}_2}, \text{ for }\mathbb{R}_2\subset[n]\\
                \text{set } \mathsf{DisQual} = \emptyset\\
                \text{\color{blue}// decide the disqualified set based on broadcast message}\\
                \pcfor j\in \mathbb{R}_2  \\
                \t \text{parse } \mathsf{trans}_j[3] = ( \mathsf{CR}_j^{\mathsf{agree}},\mathsf{CompList})\\
                \t \pcif \mathsf{FS.Vrfy}(\mathsf{FS}.vk_j, 2, \sigma_j[3], \mathsf{trans}_j[3])=1 \\
                \t \wedge \mathsf{VRF.Vrfy}(rvk_j,\textsf{rand}, \textsf{ratio},``\textsf{agree}", \mathsf{CR}_j^{\mathsf{agree}})=1\\
                \t \pcthen \pcfor (k, \Gamma_k)\in\mathsf{CompList}\\
                \t[2] \text{\color{blue}// put a newly complained dealer in the list}\\
                \t[2] \pcif k\notin \mathsf{DisQual} \wedge \mathsf{PKE.Vrfy}(c_0^{(k)}, c_j^{(k)}, \Gamma_k)=1\\
                \t[2] \wedge g^{\Gamma_k.m} \ne \mathsf{cm}_j^{(k)} \\
                \t[2] \pcthen \mathsf{DisQual}= \mathsf{DisQual}\cup \{k\}, \\
                \text{set }\mathsf{Qual} = \mathbb{D}_3\setminus \mathsf{DisQual}\\
                \textbf{output:}\\
                \t pk = \prod_{j\in \mathsf{Qual}} \mathsf{cm}_0^{(j)},  sk_i = \sum_{j\in \mathsf{Qual}} sk_i^{(j)}\\
                \t pk_\tau = \prod_{j\in \mathsf{Qual}} \mathsf{cm}_\tau^{(j)}, \text{ for every } \tau \in [n]
            }
        \end{pcvstack}
    \end{pchstack}
    \caption{The \dkgname\ construction.}\label{fig:anytrustdkg}
\end{figure*}

\subsection{The Construction}\label{sect:dkgconstruction}
The construction is based on building blocks such as $\mathsf{PKE}$ (along with the proof decryption system), the forward-secure signature $\mathsf{FS}$, the VRF-based sortition $\mathsf{VRF}$, and $\mathsf{SoK}$.

\smallskip\noindent {\bf $\textsf{Setup}$.} Given the security parameter $\secpar$, the number of participants $n$, and the corruption bound $t$ (where the adversary can corrupt up to $t$ parties), configure the system as follows:

\smallskip\noindent\underline{\textsc{Group Description: }} Based on the security parameter $\lambda$, determine the group $\mathbb{G}$ of prime-order $p$ and its generator $g$. 

\smallskip\noindent\underline{\textsc{PKI Setup:}} Every participant $P_i$ produces three key pairs: $(ek_i, dk_i)$ for PKE, $(rvk_i, rsk_i)$ for VRF, and $(\mathsf{FS}.vk_i, \mathsf{FS}.sk_i[1])$ for the forward-secure digital signature scheme.   

\smallskip\noindent\underline{\textsc{Random Coin: }} Uniformly select a string $\textsf{rand}\sample \bin^\secpar$ that is independent of all users' public keys.

The setup also determines the value of $\textsf{ratio}$ for VRF-based sortition. Given $n$ participants executing the sortition algorithm with $\textsf{ratio}$, the chosen committee will form an any-trust group. We assume the required configurations for the underlying channels are established during this setup phase.

\smallskip\noindent {\bf $\textsf{Protocol Details}$. } Post-setup, all participants collaboratively run our DKG protocol, detailed in Fig.\ref{fig:anytrustdkg}. The protocol initiates with a broadcast round, transitions to a multicast round, and concludes with a final broadcast. A brief overview of each round is as follows:

\noindent\textbf{--Round 1.} Nodes initially determine if they're selected as dealers. If not, they refresh the secret key and exit the round (lines 1-4). Elected dealers sample a $t$-degree polynomial $f$ to decide secret shares $sk_i = f(i)$, commit to $sk_0, \dots, sk_n$, encrypt shares $sk_1, \dots, sk_n$ using others' encryption keys (lines 6-8), and sign their ID using the knowledge of $r$ w.r.t. $c_0 = g^r$ in the ciphertext (line 9). Dealers then sign the commitments and ciphertexts, update their signing keys, erase secret information, and broadcast the signed commitments and ciphertexts (lines 10-14).

\noindent\textbf{--Round 2.} Nodes receive the broadcasted messages (line 1). For every message, they authenticate the signature (line 8); if failed, they move to the next message. Otherwise, they validate its format, the VRF sortition certificate (lines 10-11), and the signature of knowledge (line 13). Moreover, they ascertain if the committed values match valid coefficients of a $t$-degree polynomial (lines 3-5 and 12-13). If a transcript fails verification, the dealer is instantly disqualified (line 14). Otherwise, they check the decrypted share's validity against the commitments and, if inconsistent, generate a verifiable complaint against the dealer (lines 15-17). All complaints are multicast.

\noindent\textbf{--Round 3.}  Nodes first verify if they are selected as senders (lines 1-4). If so, they collect and verify all complaints (using ciphertexts received from line 1 of round 2), de-duplicate them, and curate a complaint list documenting all complained dealers (lines 6-17). They then sign and broadcast this complaint list (lines 18-21).

\noindent\textbf{--End of Round 3.} Nodes finalize the set of disqualified dealers based on received complaint lists (lines 1-13). Following that, they create the public key (shares) and secret key share by aggregating contributions from qualified dealers (lines 14-16).

\subsection{The Analysis}\label{sect:dkganalysis}
\smallskip\noindent{\bf Computation complexity analysis.} We primarily focus on the computationally intensive operations, particularly the group exponentiation operation $\mathsf{EXP}$, and will omit inexpensive operations like multiplication operation and hash evaluation. In Round 1, a node who is elected as a dealer needs to generate a commitment (line 6), which takes $O(n) \mathsf{EXP}$, encrypt shares under $n$ different public keys (line 7), which takes $O(n) \mathsf{EXP}$ group expo. In Round 2, to verify a transcript, each node needs to validate its commitment (line 10), which takes $O(n) \mathsf{EXP}$, and there are expected to be $s$ transcripts to verify. If there are the maximal amount of Byzantine nodes, an honest node may need to verify $O(n)$ VRF proofs and $O(n)$ signatures (lines 6, 8), which takes $O(n) \mathsf{EXP}$, and generate $O(s)$ complaints, which takes $O(s) \mathsf{EXP}$.  In Round 3, each node verifies the complaints from other nodes. Note that a node $P_i$ stops verifying complaints from $P_j$ upon finding an invalid complaint made by $P_j$, and it also stops verifying complaints against a dealer once the dealer has verifiably complained. Therefore, a node verifies at most $O(n)$ complaints, which takes $O(n) \mathsf{EXP}$. There are no group exponentiation operations in Round 4. Thus, even facing the maximal amount of Byzantine nodes, each node only needs to perform $O(sn)$ group exponentiation. 

\smallskip\noindent{\bf Communication complexity analysis.} 
In Round 1, an elected node will broadcast $O(n\secpar)$-size transcript, and there are expected to be $s$ elected node, where $\secpar$ is the computational security parameter, and the sizes of a group element, a digital signature, and a VRF credential, are counted as $O(\secpar)$. So the communication cost in Round 1 is $s\mathcal{B}(n \secpar)$, where $\mathcal{B}(\ell)$ denotes the communication cost of broadcasting $\ell$ bits by one sender. There will be no further communication if all nodes behave honestly. Otherwise, in Round 2, each node may need to multicast $O(s)$ complaints, which in total incurs the communication complexity of $n\mathcal{M}(s\secpar)$, where $\mathcal{M}(\ell)$ denotes the communication cost of multicasting $\ell$ bits by one sender. In Round 3, there are $O(s)$ elected nodes each broadcasting $O(s\secpar)$-sized complaints, which incurs the communication complexity of $s\mathcal{B}(s\secpar)$. In summary, the good case communication complexity is $s\mathcal{B}(n\secpar)$, and the adversarial-case communication complexity can be $s\mathcal{B}(n\secpar) + n\mathcal{M}(s\secpar)$.

\smallskip\noindent{\bf Security analysis.} 
The $t$-correctness and $t$-consistency, which ensure all participants at the end of the protocol obtain the same public key and correct shares, follow the facts: (1) there is at least one qualified dealer, and (2) all malicious dealers will be disqualified.
    From the security of VRF-based sortition and the forward-secure signature, at least one honest node will be selected as a dealer, and it can successfully broadcast its valid transcript even if it later becomes corrupted, which ensures (1).
    From the security of the polynomial commitment, if a dealer is not complained by any node, then all honest users receive consistent shares from the dealer. Our complaint phase guarantees that every dealer who is complained by an honest node will disqualified. Therefore, we have (2).
    
    Proving the oracle-aided algebraic simulatability is more involved. At a high level, we need to construct an efficient simulator that, on input from a sequence of group elements, produces an indistinguishable view for an adaptive adversary with the help of a DLog oracle. We follow the techniques from \cite{BachoL22} to simulate the polynomial commitments and opening shares for corrupted nodes. However, broadcasting all encrypted shares in our protocol poses additional challenges for security proof. Particularly, the simulator needs to simulate all ciphertexts without knowing the shares. It also needs to simulate the proof of decryption without using the decryption oracle of the underlying encryption scheme. As sketched in Sect.\ref{sect:tech}, we leverage the non-committing encryption and signature of knowledge to handle these challenges. Formally, 
we establish the following theorem.
\begin{theorem}\label{theo:dkgsecurity}
    The \dkgname\ satisfies $t$-consistency, $t$-correctness, and $(t, k, T_\adv, T_{\mathsf{sim}})$-oracle-aided algebraic simulatability against adaptive adversaries (cf Def.\ref{def:oracleaided}), with $n \geq 2t+1$, $k \leq n(t+1)$ and $T_{\mathsf{sim}} \leq T_\adv + \mathcal{O}(snt)$, under the DDH assumption in the ROM, and assuming the security of the underlying forward-secure signature scheme. For static adversaries, it further achieves the \textit{key-expressability}(cf. Def.\ref{def:ke}).
\end{theorem}
\begin{proof}
    Under DDH assumption in ROM, our building blocks, including the VRF, the NIZKPoK, the proof of decryption, and the multi-recipient encryption, are secure.
    
    First, we argue the $t$-consistency. Note that the public key $pk$ and the vector of public key shares are deterministically computed based on the set of qualified dealers, which are further determined by the information in the broadcast channel. As all honest users have the same view of the broadcast channel, the $t$-consistency follows easily.
    
    Then, we show the $t$-correctness. Recall that $pk = \prod_{j\in \mathsf{Qual}} \mathsf{cm}_0^{(j)}$, and $pk_i = \prod_{j\in \mathsf{Qual}} \mathsf{cm}_i^{(j)}$ for $i \in [n]$. Based on line 10 of round 2, for each $j \in \mathsf{Qual}$, with an overwhelming probability, there is a polynomial $f_j(x) \in \mathbb{Z}_p[X]$ whose degree is up to $t$, such that $\mathsf{cm}_{i}^{(j)} = g^{f_j(i)}$. Therefore, define $f(X) = \sum_{j\in \mathsf{Qual}} f_j(X)$, and then it follows that $pk = g^{f(0)}$ and $pk_i = g^{f(i)}$. Meanwhile, every honest $P_i$ should have $f(i)$. If an honest $P_i$ does not have $f(i)$, there must exist an index $j\in \mathsf{Qual}$ such that $P_i$ does not have $f_j(i)$. In this case, $P$ should follow the protocol description and multicast a verifiable complaint against the dealer $j$ to all other parties. As the verifiable complaints are posted to the broadcast channel by an any-trust group, a verifiable complaint against $j$ must be included. Then, $j$ should be disqualified, which contradicts our assumption that $j$ is in $\mathsf{Qual}$.
    
    For correctness, it remains to show that the set $\mathsf{Qual}$ is non-empty. By parameter and the security of VRF, the sampled committee contains at least one honest node with high probability. We argue this honest node will be included in $\mathsf{Qual}$. Particularly, this node shall broadcast an honestly generated transcript that contains valid shares. It is easy to see that the complaints in our system are unforgeable due to the soundness of proof of decryption. Therefore, this node cannot be disqualified because of this transcript. Moreover, although this node may be corrupted after it sends out the transcript, by the forward security of the underlying signature scheme, the adversary cannot send another message with a valid signature in this round, which means the honest node cannot be disqualified because of post-corruption.
    
    Given its length, the analysis for the oracle-aided algebraic security is presented in Lemma.\ref{lemma:oracle-aided}, and the analysis for the key expressibility is in Lemma.\ref{lemma:ke}.
\end{proof}

\begin{lemma}\label{lemma:oracle-aided}
    The \dkgname\ satisfies $(t, k, T_\adv, T_{\mathsf{sim}})$-oracle-aided algebraic simulatability.
\end{lemma}
\begin{proof}
    By definition, if $\Pi$ satisfies the oracled-aided algebraic simulatability, then, for every adversary $\adv$,  there will be an algebraic simulator $\mathsf{Sim}$ which can indistinguishably simulate the environment for $\adv$.
    We proceed with the proof by presenting the code of a universal simulator $\mathsf{Sim}$, which has access to the adversary $\adv$. 
    
    On inputs a vector of group elements $\zeta = (g^{z_1}, g^{z_2}, \dots, g^{z_k})$ for $k = n(t+1)$, $\mathsf{Sim}$ can simulate each phase of $\Pi$ for $\adv$ as follows.
    
    \smallskip\noindent{\bf \textsf{SETUP}.} $\mathsf{Sim}$ initializes the set of corrupted parties $\mathcal{C} = \emptyset$, the set of honest parties $\mathcal{H} = \{P_i \}_{i\in [n]}$, and a table $\mathsf{RO}_{\mathsf{hist}} = \emptyset$ to record the query history of the random oracle.  Then, it follows the protocol specifications to generate the public parameters and key pairs for all honest users. It answers the adversary's queries as follows.
    \begin{itemize}
        \item \textbf{Corruption queries.} When $\adv$ asks to corrupt the party $P_i$, $\mathsf{Sim}$ first checks if $|\mathcal{C}|\leq t$. If the check fails, it ignores this query; otherwise, return the secret keys of $P_i$, and update the sets $\mathcal{H} = \mathcal{H} \backslash \{P_i \}$ and $\mathcal{C} = \mathcal{C} \cup \{P_i\}$.
        \item \textbf{Random oracle queries.} When $\adv$ queries the random oracle with an input $x$, $\mathsf{Sim}$ checks if $x$ has been asked before. If there is a record of $(x, \mathsf{output}_x)$ in $\mathsf{RO}_{\mathsf{hist}}$, return $\mathsf{output}_x$; otherwise, uniformly sample $\mathsf{output}_x$, add $(x, \mathsf{output}_x)$ to $\mathsf{RO}_{\mathsf{hist}}$, and return $\mathsf{output}_x$.
    \end{itemize}
    
    \smallskip\noindent{\textbf{Round 1}.} For every honest party $P_i\in \mathcal{H}$, $\mathsf{Sim}$ runs the \textit{Self-Election} procedure using $P_i$'s VRF secret key. 
    We assume w.l.o.g. there are $s'\leq n$ honest parties being selected and denote the set by $\mathcal{H}_{\mathsf{ele}}= \{ \mathcal{D}_1, \dots, \mathcal{D}_{s'}\}$, where each party has its credential $\mathsf{CR}_{j,\mathsf{deal}}$.
    Then, $\mathsf{Sim}$ simulates the \textit{Commit to secret} procedure on behalf of each $\mathcal{D}_j \in \mathcal{H}_{\mathsf{ele}}$ as follows.
    \begin{itemize}
        \item Denote $\zeta_j = (\zeta_{j,0}, \zeta_{j,1}, \dots, \zeta_{j, t}) \\= (g^{z_{(j-1)(t+1) +1}}, g^{z_{(j-1)(t+1) +2}}, \dots, g^{z_{j(t+1)}})$. 
        \item Generate the commitments $\mathsf{cm}^{(j)}_{\tau} = \prod_{\mu \in [0,t]} \zeta_{j, \mu}^{\tau^\mu}$, for every $\tau \in [0, n]$.
        \item Generate the ciphertext $(c^{(j)}_{0}, c^{(j)}_{1}, \dots, c^{(j)}_{n})$, where $c^{(j)}_{0} = g^{r_j} $ for some $r_j \sample \mathbb{Z}_p$, and $c^{(j)}_{\tau} \sample \{0,1\}^{\lceil \log p \rceil}$ for $\tau \in [n]$.
        \item Use the simulated signer algorithm of $\mathsf{SoK}$ to sign $j$ w.r.t. $c_0^{(j)}$ and obtain a simulated signature of knowledge $\sigma_{\mathsf{DL}}^{(j)}$.
        \item Broadcast $(\mathsf{CR}^{\mathsf{deal}}_j, \mathsf{cm}^{(j)}_{0}, \dots, \mathsf{cm}^{(j)}_{n}, c^{(j)}_{0}, \dots, c^{(j)}_{n},\sigma_{\mathsf{DL}}^{(j)})$ along with its forward-secure signature on it.
    \end{itemize}
    
    $\mathsf{Sim}$ needs to answer the queries from the adversary. For the random oracle queries and corruption queries made before broadcast, $\mathsf{Sim}$ can respond as it does in the \textsf{SETUP} phase. We discuss its strategy for answering these queries that are made after the broadcast below.
    
    \begin{itemize}
        \item \textbf{Corruption queries.} When $\adv$ asks to corrupt the party $P_i$, $\mathsf{Sim}$ first checks if $|\mathcal{C}|\leq t$. If the check fails, it ignores this query; otherwise,
        $\mathsf{Sim}$ queries the oracle $\mathsf{DL}_g(\cdot)$ with $\mathsf{cm}^{(j)}_{i}$ for all $j\in \mathcal{H}_{\mathsf{ele}}$ with its representation $(1, i^1, \dots, i^t)$ over $\zeta_j$.
        $\mathsf{Sim}$ will receive $\xi_{i}^{(j)}$ from the oracle.  Then, $\mathsf{Sim}$ records all $\{x_{i}^{(j)}\}$, which are secret shares dealt by honest dealers. $\mathsf{Sim}$ can obtain the shares dealt by corrupted dealers for $P_i$ by decrypting the encrypted shares using $dk_i$. Finally, $\mathsf{Sim}$ returns the secret shares for $P_i$ and its decryption key $dk_i$ to $\adv$, and updates the sets $\mathcal{H} = \mathcal{H} \backslash \{P_i \}$ and $\mathcal{C} = \mathcal{C} \cup \{P_i\}$.
        \item \textbf{Random oracle queries.} Before answering any random oracle queries at this stage, $\mathsf{Sim}$ first calculates a matrix of group elements
        \begin{equation*}
            \gamma = \left(
            \begin{aligned}
                & \gamma_{1, 1}& \gamma_{1,2}& \dots & \gamma_{1,n}\\
                & \gamma_{2, 1}& \gamma_{2,2}& \dots & \gamma_{2,n}\\
                & \vdots& \vdots & & \vdots\\
                & \gamma_{s', 1}& \gamma_{s',2}& \dots & \gamma_{s',n}\\
            \end{aligned}
            \right),
        \end{equation*}
        where each $\gamma_{j, \tau} = pk_\tau^{r_j}$ for $j \in [s']$ and $\tau\in [n]$, $pk_\tau$ is the encryption public key of $P_\tau$, and $r_j$ is the randomness used in encryption by $\mathsf{Sim}$ when simulating $\mathcal{D}_j$. $\mathsf{Sim}$ checks if any $\gamma_{j, \tau}$ has been asked before and \textbf{aborts} in one is in the query history. Otherwise, continue.
        
        When $\adv$ queries a message $x$, $\mathsf{Sim}$ performs as follows.

        -- If $x \ne \gamma_{j, \tau}$ for any $j$ and $\tau$, proceed as what it did in the Setup phase.\\
        -- If $x = \gamma_{j, \tau}$ for some $j$ and $\tau$, checks if $P_\tau$ has been corrupted. If it has not been corrupted, then
        $\mathsf{Sim}$ first queries the oracle $\mathsf{DL}_g(\cdot)$ with $\mathsf{cm}^{(j)}_{\tau}$ and its representation $(\tau^0, \tau^1, \dots, \tau^t)$ over $\zeta_j$. $\mathsf{Sim}$ will receive $\xi_{j, \tau}$ from the oracle. If it is corrupted, then $\xi_{j, \tau}$ has been recorded by $\mathsf{Sim}$. Finally, 
        it sets $\mathsf{output}_{\gamma_{j, \tau}} : = c^{(j)}_{\tau} \oplus \xi_{j, \tau}$, records $(\gamma_{j, \tau}, \mathsf{output}_{\gamma_{j, \tau}})$ into $\mathsf{RO}_{\mathsf{hist}}$, and returns $\mathsf{output}_{\gamma_{j, \tau}}$ to $\adv$.
    \end{itemize}

    \smallskip\noindent{\bf Other rounds.} $\mathsf{Sim}$ simulates the behavior of honest parties by following the specifications of the protocol, except that whenever an honest party $P_i$ needs to decrypt an encrypted share $(c_0, c_1, \dots, c_n)$, $\mathsf{Sim}$ instead performs the following procedures for decryption.
    \begin{itemize}
        \item Use the extractor of the SoK to obtain $r$, such that $c_0 = g^r$. Then, use $r$ to ``decrypt'' the encrypted share as $sk_i = \mathsf{Hash}(ek_i^r)\oplus c_i $.
    \end{itemize}
        The queries from $\adv$ are answered in the same way as $\mathsf{Sim}$ did in Round 1. 
    
    Let $\mathsf{Qual}_C$ be the set of qualified nodes that are corrupted before the \textbf{Round 1}, and $\mathsf{Qual}= \mathsf{Qual}_C \cup \mathcal{H}_{\mathsf{ele}}$. For every $j\in \mathsf{Qual}_C$, the dealer must have distributed its secret shares to honest nodes; otherwise, it will be disqualified. As $\mathsf{Sim}$ has always controlled more than $t+1$ honest participants, it can recover the secret key $sk_j$ w.r.t. $pk_j$ for every $j\in \mathsf{Qual}_C$. Therefore, $\mathsf{Sim}$ can output the algebraic representation for the public key as:
    \begin{equation*}
        pk = \prod_{j\in \mathsf{Qual}}pk_j = g^{\sum_{j\in \mathsf{Qual}_C}sk_j}\prod_{j\in [s']}g^{z_{(j-1)(t+1)+1}}.
    \end{equation*}
    
    \smallskip Now, we argue that the simulator specified above satisfies the requirements of oracle-aided simulatability. First, it is easy to verify that the running time of $\mathsf{Sim}$ is $T_\adv + \mathcal{O}(snt)$. 
    
    Then, we show that $\mathsf{view}_{\adv, y, \Pi}$ and $\mathsf{view}_{\adv, y, \Pi}$ are identical, under the condition that $\mathsf{Sim}$ never aborts during the simulation. Specifically, from the point of $\adv$'s view, the commitment sequence outputted by an honest party $\mathcal{D}_j$ is a commitment to the polynomial $f_j(x) = \sum_{\tau = 0}^n z_{(j-1)(t+1) + \tau+1}x^{\tau}$. Note that the input group elements of $\mathsf{Sim}$ are uniformly sampled, and thus, the distribution of $f_j(x)$ is also uniform, which is identical to that in the real experiment. Moreover, in the random oracle model, the distribution of ciphertexts simulated by $\mathsf{Sim}$ is also identical to the real distribution. Notably, for every $pk_{\tau}^{r_{j}}$ that has been issued to the random oracle, which means that $\adv$ can decrypt the ciphertext $c_{j,\tau}$, it follows that
    \begin{equation*}
        c_{\tau}^{(j)} = \mathsf{Hash}(pk_{\tau}^{r_{j}})\oplus f_{j}(\tau).
    \end{equation*}
    
    Next, we argue that $\mathsf{Sim}$ only aborts with a negligible probability. When $\mathsf{Sim}$ aborts, $\adv$ must have queried the random oracle with some $x = \gamma_{j, \tau}$ before seeing the broadcast messages. However, $\gamma_{j, \tau} = pk_{\tau}^{r_j}$ is a uniformly random group element, as $r_j$ is uniformly chosen from $\mathbb{Z}_p$ and completely independent of $\adv$'s view before $g^{r_j}$ is broadcasted. Therefore, $\adv$ has negligible probability if outputting $pk_{\tau}^{r_j}$.
    
    Then, we show that $\mathsf{Sim}$ has made at most $k-1$ queries to the $\mathsf{DL}_{g}(\cdot)$ oracle. Recall that $\mathsf{Sim}$ makes a query to $\mathsf{DL}_{g}(\cdot)$ whenever $\adv$ corrupts a party or queries the random oracle with a message $x$ which is equal to some $\gamma_{j, \tau}$. We note that under the DDH assumption, $\adv$ can output $\gamma_{j, \tau} = pk_\tau^{r_j}$ only when $\adv$ has corrupted the party $P_\tau$ (and thus can compute $\gamma_{j, \tau} = (g^{r_j})^{sk_\tau}$), except a negligible probability. As $\adv$ can corrupt at most $t$ parties, $\mathsf{Sim}$ will query $\mathsf{DL}_{g}(\cdot)$ at most $ts'$ times, which is smaller than $k-1$. 
    
    Finally, we show the simulatability matrix $L$ of $\mathsf{Sim}$ is invertible. Without loss of generality, we assume that the adversary has corrupted the parties $P_1, \dots, P_t$, and $\mathsf{Sim}$ has made $s't$ queries to $\mathsf{DL}_g(\cdot)$ for simulating the queries from the adversary. For ease of analysis, we let $\mathsf{Sim}$ make some dummy queries such that the representations of all the queries are gonna form a square matrix of order $n(t+1)$. Specifically, $\mathsf{Sim}$ makes the following extra queries:
    \begin{equation*}
        g^{z_{s'(t+1)+1}}, g^{z_{s'(t+1)+2}}, \dots, g^{z_{n(t+1)}},
    \end{equation*}
    and 
    \begin{equation*}
        \prod_{\mu\in [0,t] } \zeta_{j,\mu}^{(t+1)^\mu}, \text{ for } j\in [1, s'-1].
    \end{equation*}
    The number of all queries by $\mathsf{Sim}$ is $s't+(n-s')(t+1) + s'-1 = n(t+1) -1$, which is still smaller than $k$. It is easy to verify the matrix $L$ is invertible.    
\end{proof}

\begin{lemma}\label{lemma:ke}
    The \dkgname\ satisfies the key-expressability.
\end{lemma}
\begin{proof}
    This proof is similar to the proof for Lemma.\ref{lemma:oracle-aided}, except we don't need to handle adaptive corruption queries. For any PPT adversary $\adv$, we can construct a PPT simulator $\mathsf{Sim}$ that takes as input a public key $pk'\in \mathbb{G}$ and simulates the view of $\adv$. Assume the sef of corrupted paries is $\{P_i \}_{i\in \mathsf{Corr}}$ for some $\mathsf{Corr}\subset[n]$ and $|\mathsf{Corr}|\leq t$.   After sampling the any-trust group, $\mathsf{Sim}$, on behalf of the honest node in the group, creates the following transcript: $\mathsf{cm}_0 = pk'$, $c_0=g^r$ for some $r\sample \mathbb{Z}_p$;  For $i\in \mathsf{Corr}$, $sk_i\sample \mathbb{Z}_p$, $\mathsf{cm}_i = g^{sk_i}$, and $c_i = \mathsf{Hash}(ek_i^r)\oplus sk_i$. For $i\notin \mathsf{Corr}$, $\mathsf{cm}_i$ are created by Lagrange interpolation in the exponent, while $c_i$ are randomly sampled. This transcript is indistinguishable from an honestly generated one in the view of $\adv$ and cannot be disqualified. For every other transcript with $\mathsf{cm}^{(j)} = pk^{(j)}$ which is eventually included in the qualified set, $\mathsf{Sim}$ can know the secret key $sk^{(j)}$ by reconstructing it from shares held by honest nodes. Note that the final public key is in the form of $pk'\cdot \prod pk^{(j)}$, and the simulator can express it by setting  $\alpha =1$, $sk'' = \sum sk^{(j)}$.
\end{proof}

\smallskip\noindent{\bf Committee Size.}
Recall that our construction employs a VRF-based sortition to decide the committee, in which each node can be independently elected a committee member with a probability $\texttt{ratio}$. Assume a network of $n$ nodes while at least $h$ of them remains honest. Such a sortition process will produce a committee of the expected size of $s = \texttt{ratio} \cdot n$. Then the probability $\texttt{p}$ that at least one honest node being elected can be expressed as follows:
\begin{equation}
    \texttt{p} = 1- (1-\frac{s}{n})^h.
\end{equation}

We compute the expected committee sizes necessary to ensure different values of $\texttt{p}$ in networks with varying ratios of honest parties, as depicted in Table \ref*{tab:committee_sizes}. For example, assuming over $51\%$ participants of the whole network are honest, we can set the expected committee size as $38$, which ensures the resulting committee contains at least one honest node with the probability of at least $1-5\times 10^{-9}$. These findings are applicable to networks of any size, although for networks with $n \leq 10^4$, a slightly smaller committee size may be achievable.

\begin{table}[htbp]
    \centering
    \begin{tabular}{|l|c|c|c|}
    \hline
    \diagbox{$\texttt{PR}$}{$\texttt{HR}$} & 51\% & 67\% & 80\% \\
    \hline
    $1- 5\times 10 ^{-9}$ & 38 & 29 & 24\\
    $1-  2 ^{-30}$ & 41 & 32 & 26\\
    $1- 2 ^{-40}$ & 55 & 42 & 35\\
    \hline
    \end{tabular}
    \caption{Expected committee sizes for different probability guarantees ($\texttt{PR}$) under different honest-party ratio ($\texttt{HR}$)}
    \label{tab:committee_sizes}
    \vspace{-1.cm}
\end{table}

%% file: data/weight.tex
\section{Sub-ID Allocation for the Weighted Setting}\label{sec:subid}
In this section, we present a simple yet effective sub-ID allocation mechanism that enables us to apply a conventional distributed protocol like our \dkgname\ in the weighted setting. Compared with the straightforward sub-ID allocation mechanism, ours dramatically reduces the number of required sub-IDs.

\smallskip\noindent{\bf Qualified allocation.} The traditional sub-ID allocation method ensures that the proportion of sub-IDs held by honest participants is equal to the proportion of an honest participant's weights, which we call a \textit{perfect} allocation. However, we notice a gap between the usual assumption on the honest participant's weight ratio, which is typically assumed to be more than 2/3 due to other components of the system, and the honest ratio needed in threshold cryptography, which is usually just above 1/2. Therefore, we consider a lossy-yet-qualified allocation, which guarantees that more than half of the sub-IDs will be issued to honest participants if they have more than 2/3 of the weights\footnote{While our discussion primarily centers on the gap between 2/3 and 1/2, the underlying concept can be effortlessly extended to address other thresholds or scenarios.}. Formally, we have the following definition.
\begin{definition}[Qualifed Allocation]\label{def:quallo}
    Let $W = (w_1, \dots, w_n)$ be a sequence of positive integers. Let $A$ and $B$ be any partition of the index set $[n]$ (\textit{i.e.}, $A\cup B = [n]$ and $A\cap B = \emptyset$). We say a function $\mathsf{AllocateSubID}(w_1, \dots,w_n)\rightarrow (d_1, \dots, d_n)$, where $d_i$'s are non-negative integers, is a \textbf{qualified allocation} for $W$, if for every $(A,B)$ s.t.
    \begin{equation*}
        \sum_{i\in A} w_i > 2\cdot \sum_{i\in B} w_i,\text{     it holds that }\sum_{i\in A} d_i > \sum_{i\in B} d_i.
    \end{equation*}
\end{definition}
While such a qualified allocation suffices for security, we need to find an allocation method that minimizes the number of all sub-IDs, \textit{i.e.}, $\sum_{j} d_j$ is as small as possible.

\smallskip\noindent{\bf Our method.} 
We start by observing that dividing each $w_i$ by the greatest common division (GCD) leaves the fraction for any index subset $A$ unchanged. This realization provides a straightforward allocation approach: $d_i = \frac{w_i}{\textsf{gcd}}$. However, if the GCD is small, the total sub-IDs can be vast.

A viable approach is to modify each $w_i$ to $w_i'$ so the new sequence $W' = (w_1',\dots, w_n')$ has a substantial GCD. This adjustment might increase some subsets' proportions while reducing others, potentially strengthening the adversary. Still, we determine that any increased power for the adversary remains capped if we limit the total adjustments. 

Specifically, we call an adjustment $t$-bounded for $(w_1, \dots, w_n)$, if the adjusted values $(w_1', \dots, w_n')$ satisfies $\sum_{i\in [n]}|w_i-w_i'|\leq t$. Then, if $\sum_{i\in[n]}w_i\geq 3t+1$, it ensures that  
for any partition $(A,B)$ over $[n]$ satisfying $\sum_{i\in A}w_i > 2\cdot \sum_{i\in B}w_i$, it holds that $\sum_{i\in A} w_i' > \sum_{i\in B} w_i'$, given the inequality:
\begin{equation}
    \begin{aligned}
    \sum_{i \in A} w_i' - \sum_{i \in B} w_i'
        \geq
        \sum_{i \in A} (w_i-\Delta_i) - \sum_{i \in B} (w_i +\Delta_i)\geq 1
    \end{aligned}
\end{equation}
Here, $\Delta_i = |w_i-w_i'|$. Following the adjustment, Sub-IDs, $d_i$, are derived by dividing $w_i'$ by this higher GCD.

Given our objective to minimize $\sum_{j} d_j$, the goal is to enhance the GCD. To achieve this, we consider a target $\textsf{gcd}$, defining an adjustment function $f_{\textsf{gcd}}(w_i)\rightarrow w_i'$ as:
\begin{equation}
    w_i' = \left\{
    \begin{aligned}
        & w_i  - (w_i \bmod \textsf{gcd}),  \text{ if } w_i\bmod \textsf{gcd} < \textsf{gcd}/2,\\
        &w_i + \textsf{gcd} - (w_i \bmod \textsf{gcd}), \text{ otherwise. }
    \end{aligned}
    \right.
\end{equation}

Starting with $\textsf{gcd}=1$, we increase it until $f_{\textsf{gcd}}$ is no longer a $t$-bounded adjustment for $W$. Utilizing binary search can quickly find a very large $\textsf{gcd}$. While variations in $(w_1, \dots, w_n)$ may suggest larger $\textsf{gcd}'$, our found $\textsf{gcd}$ is practically near-optimal. The allocation algorithm is detailed below.

\begin{pchstack}[boxed,center]
    \procedure{$\mathsf{AllocateSubID}(w_1,\dots,w_n)$}{%
        \textbf{binary search the largest } \textsf{gcd} \text{ from 0 to } \max_{i}w_i\\
        \t \text{s.t. } f_{\textsf{gcd}} \text{ is }t\text{-bounded for }(w_1, \dots, w_n)\\
        \textbf{output } (d_i = \frac{f_\textsf{gcd}(w_i)}{\textsf{gcd}})_{i\in[n]}
    }
\end{pchstack}

\medskip\noindent{\bf Efficiency and effectiveness.} Note that our $\mathsf{AllocateSubID}$ is only supposed to find a $t$-bounded $f_{\textsf{gcd}}$ for its input $(w_1,\dots,w_n)$. So we can efficiently check whether $\sum_{i\in [n]} |f_{\textsf{gcd}}(w_i)-w_i|\leq t$. Thus, the time-cost of $\mathsf{AllocateSubID}$ is $O(n\log n)$. Meanwhile, using binary search is effective since there is a general trend that the larger the gcd is, the larger adjustment is needed. It can give us a $t$-bounded $f_{\textsf{gcd}}$ for the input with a large gcd (not necessarily optimal).

Our sub-ID allocation is a qualified allocation as per Def.\ref{def:quallo}, since $f_{\textsf{gcd}}$ is $t$-bounded for $(w_1, \dots, w_n)$. Moreover, for a set of $n$ validators with an arbitrary power distribution, our method only issues at most $2n$ sub-IDs.
\begin{lemma}
    Given any sequence $W = (w_i)_{i\in[n]}$ with $\sum_{i\in[n]} w_i = 3t+1 $ for some integer $t$, let $(d_1, \dots, d_n)$ be the output of our $\mathsf{AllocateSubID}$. It follows that 
$
        \sum_{i\in [n]} d_i \leq \frac{4t+1}{\lfloor 2t/n \rfloor},
$
    which is around $2n$ when $n\ll t$.
\end{lemma}
\begin{proof}
    Let $\textsf{gcd} = \lfloor 2t/n \rfloor$. It is easy to see that $(w_1',\dots, w_n')$ outputted by $f_{\textsf{gcd}}(w_1,\dots, w_n)$ and $(w_1,\dots, w_n)$ are bounded by $n\cdot \lfloor 2t/n\rfloor /2 = t$. Let $d_i = \frac{w_i'}{\textsf{gcd}}$. It holds that $\sum_{i\in[n]} d_i = \frac{\sum_{i\in n}w_i'}{\lfloor 2t/n \rfloor} \leq \frac{4t+1}{\lfloor 2t/n \rfloor} \approx 2n$.
\end{proof}

\smallskip\noindent{\bf Comparison with Swiper/Dora \cite{abs-2307-15561}.} A concurrent work, Swiper/Dora \cite{abs-2307-15561}, also addresses the imparity between conventional threshold cryptography and the weighted setting. In Table \ref{tb:subid}, we compare our method and theirs for validator sets across various PoS systems. The comparison is under the same condition, \textit{i.e.}, ensuring more than 1/2 sub-IDs are allocated to honest parties with more than 2/3 weights. The result shows our method issues fewer sub-IDs to large sets of validators, such as Algorand and Fielcoin\footnote{We note a very recent version of Swiper/Dora\cite{TonkikhS24} (after our paper submission) has further reduced the number of sub-IDs.}.

\begin{table}
    \caption{Comparison with Swiper/Dora}\label{tb:subid}
    \vspace{-0.4cm}
    \begin{tabular}{c|cccc}
        Systems& \# Parties & \#Total Weights & \cite{abs-2307-15561} & Ours \\
        \hline
        Aptos\cite{Aptos} & 104 & $8.4708\times 10^{8}$ & {\color{red}27} & 34\\
        \hline 
        Tezos\cite{Tezos} & 382 & $6.7579\times 10^{8}$ & {\color{red}75} & 77\\
        \hline 
        Filecoin\cite{Filecoin} & 3700 & $2.5242\times 10^{19}$ & 1895 & {\color{red}1688} \\
        \hline 
        Algorand\cite{ChenM19} & 42920 & $9.7223\times 10^{9}$ & 373 &  {\color{red}301}
    \end{tabular}
    \vspace{-0.4cm}
\end{table}

%% file: data/broadcast.tex
\section{Practical Extended Broadcast Channels}\label{sec:bc}
In this section, we introduce a practical extension to the blockchain-based broadcast channel. Although it is folklore knowledge that, theoretically, one may throw all messages into the ledger to facilitate a broadcast, this may incur prohibitive costs in practice, as on-chain resources are generally very expensive. Instead, our extension empowers users to broadcast a message of arbitrary length while inscribing only a {\em constant-size} storage on the blockchain. Crucially, our enhanced broadcast channel retains its original simplicity and modularity. Users can conveniently interact with it using the APIs of well-established infrastructures, including both blockchains and a data dispersal network (DDN) like IPFS \cite{TrautweinRTCSSG22}.

\subsection{Building Blocks}
We formalize our building blocks. For simplicity, we model a blockchain as a public bulletin board (PBB) that allows users to post and retrieve data. 

\smallskip\noindent{\bf Public Bulletin Board.} We follow the model of PBB presented in \cite{KidronL11} and extend it to support \textit{keyword}-based retrival. A user can interact with PBB by using the following queries: 
\begin{itemize}
    \item $\mathsf{getCounter}()\rightarrow t.$ It returns the current counter value $t$.
    \item $\mathsf{post}(\textsf{kw},v)\rightarrow t$. On receiving value $v$ along with a keyword $\textsf{kw}$, it increments the counter value by 1 to $t$, stores $(t, \textsf{kw}, v)$, and responses $t$.
    \item $\mathsf{retrieve}(t_\mathsf{start}, t_\mathsf{end}, \textsf{kw})\rightarrow \{(v_i,t_i)\}$. It returns all pairs of $(v_i, t_i)$, such that $t_\mathsf{start}\leq t_i \leq t_\mathsf{end}$  and $\textsf{kw}$ is their keyword. 
\end{itemize}

We care about the storage cost of PBB. For a user posting $\ell$ bits to the PBB, we denote the cost as $\mathcal{PB}(\ell)$. We assume that a PBB satisfies the \textit{validity} and \textit{agreement}.

\smallskip\noindent\underline{\scshape Validity.} Assume an honest user posted $(v, \textsf{kw})$  to the PBB and received $t$. Then, every honest user who retrives with $(t_\mathsf{start}, t_\mathsf{end}, \textsf{kw}')$ such that $t_\mathsf{start}\leq t\leq  t_\mathsf{end}$ and $\textsf{kw}' = \textsf{kw}$  will receive a sequence of value/counter pairs containing $(v, t)$.  

\smallskip\noindent\underline{\scshape Agreement.} If an honest user retrieving with $(t_\mathsf{start}, t_\mathsf{end}, \textsf{kw})$ when $\textsf{getCounter}() \geq t_\mathsf{end}$ receives a sequence of value/counter pairs $S$, then every honest user retrieving with $(t_\mathsf{start}, t_\mathsf{end}, \textsf{kw})$ will receive the same $S$.

It is rather straightforward to use PBB as a broadcast channel by simply posting a broadcast message into the PBB. The authenticity can be established with standard digital signatures in the PKI model.

\smallskip\noindent {\bf Data Dispersal Network.} A data dispersal network (DDN) provides a platform where one can provision a data block for others who may need it. 
Compared with standard multicast, which is also for data dissemination, DDN saves communication costs when there are multiple nodes providing the same data block. Assuming there are $m$ receivers out of $n$ potential receivers, and there are $k$ data providers for a data block of $\ell$ bits. Through multicast, every sender needs to send their data to every potential receiver, incurring the communication cost of $k\cdot \mathcal{M}(\ell) = O(kn\ell)$. In contrast, through DDN, each receiver receives exactly one copy of data, incurring a total communication cost of $O(m\ell)$, which is smaller than $\mathcal{M}(\ell)$.

In principle, we can either use an erasure-code-based information dispersal protocol \cite{ReedS1960} or practical infrastures like IPFS \cite{TrautweinRTCSSG22} to instantiate a DDN. In this work, we focus on the IPFS-based instantiation as it becomes easier to implement (given IPFS already exists) and model it with the following two queries, which might be specific to the instantiation.
\begin{itemize}
    \item \textsf{register:} on receiving a node ID $\textsf{nid}$ and a block ID $\textsf{bid}$ (which is the hash value of the data), it checks whether $\textsf{bid}$ has been registered. If not, add a new entry $(\textsf{bid}, \textsf{nid})$; otherwise, it appends $\textsf{nid}$ to the existing entry with $\textsf{bid}$.
    \item \textsf{retrieve:} on receiving a block ID $\textsf{bid}$, it returns the associated datablock $v$, by orchestrating the data flow from candidate providers.
\end{itemize}

We assume as long as there is an honest data provider who has registered $\textsf{bid}$ and remains active, everyone can retrieve the data block with $\textsf{bid}$. We denote the cost of registering for $s$ data blocks as $\mathcal{R}(s)$.

\subsection{Our Extended Broadcast Channel}
\noindent{\bf A strawman and our intuition. }  A naive approach to broadcasting a sizeable data block involves posting its ID, denoted as \textsf{bid}, on the PBB  while simultaneously registering both \textsf{bid} and the sender's ID (\textsf{nid}) on the DDN. However, this methodology cannot guarantee agreement. Specifically, a malicious sender has the capability to selectively deny some retrieval requests on the DDN. Moreover, an adaptive adversary, upon observing the \textsf{bid} on the PBB, can corrupt the sender, subsequently rendering the data inaccessible on the DDN.

\begin{figure}[!htb]
    \begin{pcvstack}[boxed]
        \procedure{\textbf{Round 1: } each sender $S_j(v_j) \pcdo$:}{%
            \text{compute the block ID: } \mathsf{Hash}(v_j)\rightarrow \mathsf{bid}_j\\
            \textbf{post } \mathsf{PBB.post}(\textsf{kw},  \mathsf{bid}_j),  \textsf{kw}:= (\textsf{sid}||\textsf{send}); 
            \textbf{multicast } v_j
        }
        \pcvspace
        \procedure{\textbf{Round 2: } each receiver $P_i \pcdo$:}{%
            \mathsf{PBB}.\textsf{getCounter}()\rightarrow t'_{1} \\
            \text{\color{blue} // assume the index set of senders is $\mathbb{J}$} \\
            \mathsf{PBB}.\textsf{retrieve}(t'_0, t'_{1}, \textsf{sid}||\textsf{send})\rightarrow \{(\textsf{bid}_j,t_j)\}_{j\in \mathbb{J}}\\
            \textbf{receive }\text{multicast messages: } \{v'_j\}_{j\in \mathbb{J}}\\
            \pcfor j\in \mathbb{J}: \pcif \mathsf{Hash}(v'_j) = \mathsf{bid}_j, \pcthen \textsf{valid}_j = 1; \pcelse  \textsf{valid}_j = 0\\
            \mathsf{VRF}.\textsf{Sortition}(rvk_i,rsk_i, \textsf{rand}, \text{``check"}, \textsf{ratio}_{\mathsf{hm}})\rightarrow \textsf{CR}_i\\
            \pcif \textsf{CR}_i \ne \perp\\
            \pcthen \mathsf{PBB.post}(\textsf{kw}', \textsf{CR}_i||(\textsf{valid}_j )_{j\in \mathbb{J}}), \textsf{kw}':= (\textsf{sid}||\textsf{check})
        }
        \pcvspace
        \procedure{\textbf{Round 3: } each receiver $P_i$ (with node id $\textsf{nid}_i$) $\pcdo$:}{%
            \mathsf{PBB}.\textsf{getCounter}()\rightarrow t'_{2} \\
            \mathsf{PBB}.\textsf{retrieve}(t'_1, t'_{2}, \textsf{sid}||\textsf{check})\rightarrow \{\textsf{CR}_k||(\textsf{valid}_j^{(k)} )_{j\in \mathbb{J}}\}_{k\in \mathbb{K}'}\\
            \text{verify every } \textsf{CR}_k, \text{ and obtain the valid set }\mathbb{K}\subset\mathbb{K}'\\
            \pcfor j\in\mathbb{J}: \pcif \sum_{k\in \mathbb{K}} \textsf{valid}_{j}^{(k)}\geq \frac{|\mathbb{K}|}{2} + 1\\
            \t \pcthen \textsf{final}_j = 1; \pcelse \textsf{final}_j = 0\\
            \pcfor j\in \mathbb{J}, \pcif \textsf{final}_j= \textsf{valid}_j =1, \pcthen \mathsf{DNN.register}(\textsf{nid}_i, \textsf{bid}_j)
        }
        \pcvspace
        \procedure{\textbf{At the end of Round 3:} each receiver $P_i$ $\pcdo:$}{%
            \pcfor j\in \mathbb{J}\text{ s.t. } \textsf{valid}_j=0: \\
            \t \pcif \textsf{final}_j = 1, \pcthen \mathsf{DNN.retrieve}(\textsf{bid}_j)\rightarrow v_j; \pcelse v_j = \perp\\
            \textbf{output } (v_j)_{j\in\mathbb{J}}
        }
    \end{pcvstack}
    \caption{Our extended broadcast channel}\label{fig:bc}
\end{figure}

To address these security vulnerabilities, we suggest using DDN and PBB together in a smarter way. Recognizing the potential threat of adaptive corruption, the sender directly multicasts the data block to all receivers while posting the block ID \textsf{bid} into the PBB. Importantly, this process does not induce additional overhead compared with the DDN-based dissemination since there is only one provider, and all receivers will require the data block. For agreement, an honest majority committee is sampled, which subsequently votes to validate the accessibility of the data block against the advertised \textsf{bid}. In scenarios where the majority of the committee members vouch for the data block's availability, all receivers who successfully received the data block are then prompted to register on the DDN. This ensures that any receivers who fail to receive the data through multicast will be able to retrieve it from DDN.

\smallskip\noindent{\bf Protocol details.}
We assume the PKI setup, as well as the setup for the VRF-based sortition, such that everyone in the group gets to know others' verification keys w.r.t. a digital signature scheme and VRF. A ratio $\textsf{ratio}_{\mathsf{hm}}$ is also determined in the setup, which ensures a high probability that the sampled committee will contain an honest majority. Moreover, we assume every message has been signed by the sender. Besides that,  a session id $\textsf{sid}$ and an initial counter $t'_0$ are supposed to be known to everyone in the group and can be used to retrieve related messages from the PBB. We w.l.o.g. describe our protocols in a batch manner, \textit{i.e.}, there can be multiple senders, as this is the situation of our DKG protocol. We elucidate our design in Fig.\ref{fig:bc}.

\smallskip\noindent{\bf Complexity analysis.} Assume there are $s$ senders, and each of them broadcasts a message of $\ell$ bits to the group with $n$ nodes. The communication cost of our extended broadcast channel is 
\begin{equation*}
s\cdot \mathcal{B}(\ell)=s\mathcal{PB}(\secpar) + \mathcal{O}(sn\ell) +  c\mathcal{PB}(\secpar + s) + n \mathcal{R}(s),
\end{equation*}
where $\secpar$ denotes the security parameter (\textit{i.e.}, the size of a digest, the output length of a VRF, \textit{e.t.c.}), $s\mathcal{PB}(\secpar)$ is caused by that $s$ senders post their digests into the PBB, $\mathcal{O}(sn\ell)$ is caused by that the senders multicast their message and the receivers retrieve from a DDN, $c\mathcal{PB}(\secpar + s)$ is caused by the selected committee members vote for the broadcast status, and $n \mathcal{R}(s)$ is caused by that the honest parties register to the DDN. Now, the on-chain storage cost is\textit{ independent of $\ell$}.

\smallskip
\noindent{\bf Security analysis.} We establish the security of our extended broadcast channel in the following lemma.
\begin{lemma}
Assume the underlying PBB satisfies validity and agreement, and the DDN guarantees the data block can be retrieved when there is an honest and active provider. The protocol in Fig.\ref{fig:bc} satisfies the validity and agreement.
\end{lemma}
\begin{proof}
Our construction satisfies both the validity and agreement. Regarding validity, in our protocol, when the sender is honest, every honest receiver can receive the message $v$ from the multicast channel and retrieve the digest from the PBB. Then, in \text{round 2}, selected honest committee members would vote for this broadcast (by setting and posting $\textsf{valid} = 1$), such that the final status of this broadcast will be $1$, and honest nodes can decide on $v$.  

Regarding agreement, note that whether $v = \perp$ is determined by the votes on PBB. Therefore, if an honest receiver decides on $v = \perp$, everyone will do the same thing. The potential chance causing disagreement is that when an honest receiver decides on $v \ne \perp$, some receiver cannot successfully retrieve $v$ from the DDN. Below, we show that this case is unlikely to happen. 

Assume that the adversary is allowed to corrupt at most $t$ participants among all the $n$ participants, and the VRF-based sortition at round 2 will yield a committee $\mathcal{C}$ of $c = 2t' +1 $ participants. As the parameter is configured to guarantee the honest majority of the elected committee, it implies that, for any subgroup $A$ whose size is not greater than $t$, the following probability is very small:
\begin{equation*}
    \Pr[|Z|\geq t'+1: Z = A \cap \mathcal{C}].
\end{equation*}
Now, we consider the group $B$ of nodes that are, before the election, either corrupted nodes or honest nodes that have received $v$. In the case that there are $t'+1$ votes endorsing the availability of $v$, it holds that $|B\cap \mathcal{C}| \geq t'+1$, which implies the probability $\Pr[|B|\leq t]$ is small. Therefore, the adversary cannot corrupt all nodes in $B$ even after knowing the committee $\mathcal{C}$. It follows that there is always at least one honest node that has received $v$ and provisioned it to the DDN, such that everyone can retrieve the data from the DDN and can agree on the value $v$.
\end{proof}




%% file: data/checkpoint.tex
\section{Application to All-hands Checkpointing into Bitcoin}\label{sect:checkpointframework}
In this section, we delineate how our DKG yields the first realization of the checkpointing blueprint Pikachu of Filecoin \cite{AzouviV22} that involves all validators in the whole blockchain network, e.g., Filecoin, that has 3700 of them, with various mining power. 

\subsection{Realizing the Bitcoin Checkpointing Pikachu with \dkgname}
We review the checkpointing blueprint Pikachu in Appendix.\ref{app:checkapp}. At a high level, all validators need to run a DKG for Schnorr signature every epoch, and the resulting public keys will be used as Bitcoin addresses \footnote{Bitcoin has supported Schnorr signature since its \textsc{Taproot} update.}. 
A Bitcoin transaction that embeds the digest of the PoS chain at epoch $i-1$ and transfers assets from the address at epoch $i-1$ to epoch $i$ will serve as a checkpoint for epoch $i-1$. All validators jointly run the threshold Schnorr signing protocol to create such checkpointing transactions.

Pikachu only gave a proof-of-concept prototype with 21 participants due to the inefficiency of their underlying DKG scheme. Meanwhile, as they instantiated the threshold Schnorr signature with FROST \cite{KomloG20}, which relies on a coordinator, there may be a single point of failure. In the following, we demonstrate how our \dkgname\ can realize the blueprint efficiently and securely. 

\smallskip\noindent{\bf Sub-ID allocation.} At each epoch $i$, the current validators of the blockchain locally run the deterministic sub-ID allocation algorithm on a publicly agreed power distribution, and then they obtain the same sub-ID allocation outcome. A validator with $m$ sub-IDs will participate in further protocols as $m$ individuals.

Our optimized sub-ID allocation algorithm in Sect.\ref{sec:subid} issues fewer sub-IDs to validators than the straightforward approach. We consider a snapshot of Filecoin's validator distribution\footnote{\url{https://filfox.info/en/ranks/power}}, which has 3700 validators with a total mining power of around 25 EB, while the power unit is 32KB. The standard method may issue around 674 trillion sub-IDs. In contrast, our method identifies that 13 PB can be a good GCD, and only 1688 sub-IDs need to be issued, significantly reducing the scale of the problem. 

\smallskip\noindent{\bf Apply \dkgname.} The validators with 1688 sub-IDs will act like 1688 participants to execute the DKG protocol to generate a public key for Schnorr signature and share the secret keys. 
We set $\textsf{ratio}_{\mathsf{at}}=38/1688$, guaranteeing the committee has at least one good node with a probability of $1-5\times10^{-9}$. Then, the validators can run our \dkgname\, which incurs only around 3 MB of data that needs to be broadcasted. It takes each node a few seconds to finish computation, even facing the maximum number of complaints.

\smallskip\noindent{\bf Checkpointing with non-interactive threshold Schnorr signature.} At epoch $i$, the validators of epoch $i-1$ use their shared keys to sign the checkpointing Bitcoin transaction. Note that no matter how many nodes try to post the signed transaction to the Bitcoin, there will be only one transaction appearing on the chain. To sign this transaction, we adopt the GJKR protocol\cite{GennaroJKR07}, which does not require a coordinator and is thus free of single-point failures. The GJKR protocol involves a DKG as its subroutine for generating the nonce and follows a non-interactive phase where every signer can locally compute its signature share (or called a partial signature). GJKR was believed to be inadequate for large-scale deployment due to its DKG subroutine, which, however, is no longer a bottleneck with our any-trust DKG. Since our DKG is key-expressable (cf. Def.\ref{def:ke} and \cite{GurkanJMMST21}), the static security of the resulting scheme directly follows the recent result in \cite{Shoup23}. Note that despite recent advancements \cite{CritesKM23}, achieving adaptively secure and robust threshold Schnorr without using a coordinator remains a significant open problem. We leave it as future work to analyze the adaptive security of this scheme, namely GJKR with an oracle-aided algebraic simulatable DKG.

\begin{small}
    \begin{table}[]
        \centering
        \caption{Checkpointing cost per annum. in USD.} \label{tab:checkpointcost}
        \vspace{-0.3cm}
        \begin{tabular}{c|c|c|c}
            \#\textsf{Parties}&     $2^7$(Cosmos) & $2^{10}$(Polkadot) & $2^{12}$(Filecoin) \\
            \hline
            \textsf{Babylon} & 1510826.9 & 2266245.6 & 6043306.5\\
            \hline
            Ours & \multicolumn{3}{c}{26048.8}\\
        \end{tabular}
        {\\
            *Based on the Bitcoin price on Mar. 31, 2024: 0.000708 USD per Satoshi. 
            \par}
    \vspace{-0.5cm}
    \end{table}
\end{small}
\subsection{Comparison with Babylon Checkpointing}
\noindent{\bf Overview of Babylon.} Babylon \cite{TasTGKMY23} is a recently proposed checkpointing scheme that does not use DKG and threshold signature. Instead, it employs the following approach: (1) All validators sign the digest of the PoS block to be checkpointed. (2) One honest validator collects and aggregates enough signatures (using the BLS aggregatable signature scheme \cite{BonehLS01}) and publishes a Bitcoin transaction with the $\mathsf{OP\_RETURN}$ code. This transaction contains the epoch number, the digest, the aggregated signature, and a bit vector that indicates the public keys involved.

\noindent{\bf Comparison of Bitcoin Transaction Fees.} It's important to note that for $n$ validators, at least $n$ bits are needed to encode the public key list. A Bitcoin transaction allows 80 bytes with $\mathsf{OP\_RETURN}$, which means the number of Bitcoin transactions per checkpoint grows linearly with the number of validators. Particularly, the epoch number, the block digest, and the aggregated signature together take 88 bytes; the bit-vector requires $n$ bits. Therefore, the number of Bitcoin transactions for a Babylon checkpoint can be calculated as $\#\textsf{Bitcoin Tx}_{\mathsf{Babylon}} = 1 + \lceil\frac{n+64}{640}\rceil$. 

Moreover, since it assumes an honest validator to create the checkpointing transaction, it might have a single point of failure. This issue can be resolved by sampling a committee that includes at least one honest validator for creating Bitcoin transactions. For the more secure version of Babylon, the number of Bitcoin transactions per checkpoint would increase by a factor of the any-trust committee size $\kappa$, i.e., $\#\textsf{Bitcoin Tx}_{\mathsf{secure-Babylon}}= \kappa + \kappa\cdot \lceil\frac{n+64}{640}\rceil$. For $\kappa=29$ (see Table \ref{tab:committee_sizes}) and $n = 2^{12}$, we have 
    $\#\textsf{Bitcoin Tx}_{\mathsf{Babylon}} =8$, while $\#\textsf{Bitcoin Tx}_{\mathsf{secure-Babylon}} =232$.

In comparison, our approach (Pikachu) only requires 1 Bitcoin transaction for each checkpoint, since the transaction is uniquely created via threshold signing, and Bitcoin will only accept one transaction no matter how many validators try to publish it. It is naturally free of single points of failure.

We compare the Bitcoin transaction fees for checkpointing per annum in Table \ref{tab:checkpointcost}, where the cost of Babylon is for its secure version.
%
Following \cite{TasTGKMY23}, we consider the checkpoint transactions to be created hourly. 
We assume, without loss of generality, each Bitcoin transaction has 300 bytes, the transaction fee is 14 Satoshi per byte (such that the transaction can be confirmed within six blocks as per \footnote{\url{https://btc.network/estimate}}), and the price of a Satoshi is 0.000708 USD \footnote{updated on Mar. 31, 2024, from \url{https://coincodex.com/crypto/satoshi-sats/}}. We evaluate the cost for PoS chains with different numbers of validators:  $2^7$ validators for small-scale PoS chains (like the ones in Cosmos \cite{Cosmos}), $2^{10}$validators for moderate-scale chains (like Polkadot\cite{Polkadot}), and $2^{12}$ validators for relatively large-scale chains (like Filecoin \cite{Filecoin}).

%% file: data/implement.tex

\section{Implementation and Evaluation}
We implemented our proposed DKG and present the experimental results in this section.

\smallskip\noindent{\bf Implementation.} We implemented our protocol in Java 8, comprising approximately 1500 lines of code. To facilitate Elliptic Curve operations and communication, we utilized the open-source Java library \texttt{mpc4j}\footnote{\url{https://github.com/alibaba-edu/mpc4j}} and the \texttt{Bouncy Castle} library\footnote{\url{https://www.bouncycastle.org/}}. Given our protocol's primary application in creating checkpoints on Bitcoin, we opted for the \textsf{secp256k1} curve and \textsf{SHA-256} for cryptographic operations. Our implementation includes components such as VRF and multi-recipient encryption but does not employ the broadcast extension trick in Appendix \ref{sec:bc}. It is essential to note that this implementation serves as a proof-of-concept, demonstrating the practicality of our protocol for large-scale deployment, even under the presence of the maximal number of Byzantine nodes.  We do not implement forward-secure signatures; however, their cost is marginal and independent of the scale.
Whenever possible, we set the expected size $s$ of an any-trust group to 38, which ensures that the committee qualifies with a probability of $1 - 5 \times 10^{-9}$, as in Table \ref{tab:committee_sizes}. For small-scale tests like $n = 16$ and $32$, we set $s = n/2 +1$.


\begin{figure}[h]
    \centering
    \begin{subfigure}[t]{\linewidth}
        \centering
        \includegraphics[width=1\linewidth]{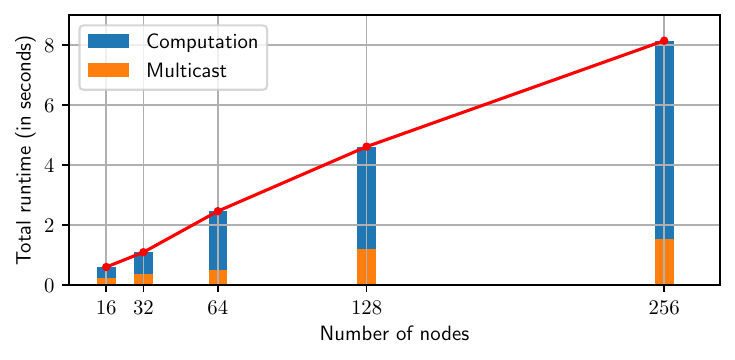}
        \vspace{-0.7cm}
        \caption{Worst-case Adjusted Running Time of bad instances.}
        \label{subfig:e2e-running-time}
    \end{subfigure}

    \begin{subfigure}[t]{\linewidth}
        \centering
        \includegraphics[width=1\linewidth]{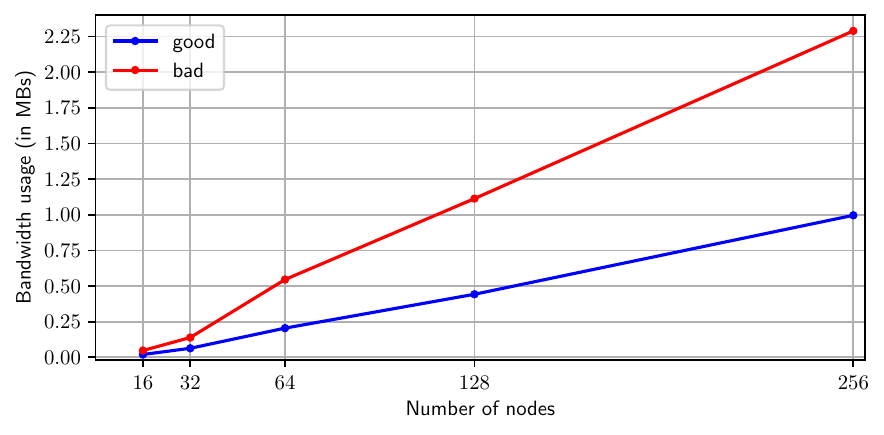} \vspace{-0.7cm}
        \caption{Worst-case bandwidth usage, the amount of data transfers inbound to and outbound from a node during the ATDKG protocol.}
        \label{subfig:e2e-bandwidth}
    \end{subfigure}
    \vspace{-0.4cm}
    \caption{End-to-end Test Results}
    \label{fig:e2e-result}
    \vspace{-0.6cm}
\end{figure}
\subsection{End-to-End Implementation}

\smallskip\noindent{\bf Evaluation Setup.} We evaluate our \dkgname\ implementation with a varying number of nodes: 16, 32, 64, 128, and 256. Each node is encapsulated within an individual Amazon Web Services (AWS) \textbf{t3a.medium} EC2 virtual machine (VM). Each VM has 2 vCPUs and 4 GiB RAM and runs in \textsf{Amazon Linux 2023 AMI 2023.4.20240416.0 x86\_64 HVM kernel-6.1}. All nodes are placed in the same AWS region and are connected pair-wise; for example, every two nodes are directly connected.
Since the network delay within the same AWS region is almost negligible, we simulate a more realistic delay by employing the Linux command \texttt{tc} (traffic control) to introduce an artificial delay of 100 ms for all traffic.

\smallskip\noindent{\bf Implementation Remarks.} 
We set up an additional node to simulate a blockchain, which serves as the broadcast channel in our implementation. The blockchain node is directly connected to all other nodes. In \text{Round 1} and \text{Round 3} of our DKG, whenever a node needs to broadcast a message, it sends the message directly to the blockchain node. The blockchain node then relays all received messages in the round to every node in the network. As our protocol assumes network synchrony and proceeds round by round, we need to specify the time window for each round. 

In practice, the time window setting for Round 1 and 3 can vary depending on the blockchain. For simplicity, we artificially configure the time window to be 30 seconds: the blockchain node receives messages in the first 20 seconds and then relays the messages. All nodes stop receiving current-round messages at 30 seconds and move to the next round. Given such a configuration, a 60-second broadcast running time is inherent to our experiments, and our experiments are more concerned about the running time incurred by Round 2 and the computation cost in Round 1, Round 3 and at the end of Round 3.


We evaluate the performance of our DKG in both the good-case and bad-case scenarios. In the good cases, all nodes are honest. In the bad cases, we set all nodes whose node ID is smaller than $n/2$ to be corrupted. A corrupted node, if elected as a dealer in Round 1, will broadcast malformed ciphertexts to all nodes, causing $n$ complaints against it in Round 2. Given a fixed number of nodes and good or bad cases, each experiment configuration is repeated eight times.

In reality, the timeout parameter for the \textbf{receive\{\}} at the start of Round 3 should be calibrated based on the communication cost in bad cases. In our implementation, the calibration is implicit: the timeout parameter is set to a sufficiently large value, while the actual communication cost in bad cases is measured independently.

\smallskip\noindent{\bf Adjusted Running Time.} The total running time of the entire \dkgname\ protocol can be defined by the time difference between the moment when the communication network is established and when a node finishes computing the shared public key and its secret share.
However, this measurement will always incorporate the 60-second broadcast time, which may vary in different settings.
Hence, an adjusted running time is measured by subtracting the 60-second broadcast cost from the running time of the whole \dkgname\ protocol. Observe that the adjusted running time consists only two components: the communication cost incurred by the \textbf{multicast} in Round 2, and all computation costs throughout this protocol.

We take the \textbf{maximum} adjusted running time across all nodes, and all repeats, to represent the end-to-end running time of our protocol.
The adjusted running time of bad cases are shown in Fig.\ref{subfig:e2e-running-time}. In additional, a breakdown by the \textbf{multicast} communication cost and the computation cost is shown within the stacked bar chart.
Note that the good cases should always perform better than the bad cases, hence, we only represent the bad cases to demonstrate the worst-cast scenario.

Our DKG protocol only requires a few seconds to finish the multicast round and all computation tasks, in additional to the omitted 60-second broadcast cost.

\smallskip\noindent{\bf Bandwidth Usage.} We record the inbound and outbound bandwidth of each node in Megabytes ($10^6$ bytes) and demonstrate the \textbf{maximum} bandwidth usage of all nodes, and all repeats in Fig.\ref{subfig:e2e-bandwidth}. The key observation is that the bandwidth grows linearly depending on the size of the group.

At first glance at the results, some non-linearity may be noticed. However, this is mainly caused by (a) a lower sortition ratio for $n=16$ and $n=32$ and (b) the randomness in the sortition results in the \dkgname\ protocol.
Specifically, the protocol has no deterministic control over the actual number of parties being elected as dealers, meaning fluctuations will be observed in bandwidth usage, as the actual number of dealer may vary. 


\subsection{Performance Analysis on Large Scale}\label{sect:estimatedformance}
While our end-to-end implementation demonstrates that our protocol remains practical when $n = 2^8$, we further tested the computation time of our protocol and estimated the communication cost on larger scales ranging from $n = 2^9$ to $n = 2^{15}$, this range covers the sizes of most PoS chain validators. 

\smallskip\noindent{\bf Broadacst cost.} 
We calculate the total number of bits to be sent via the broadcast channel. We compare our protocol and KZG in terms of it, ranging from $n = 2^9$ to $n = 2^{15}$, considering both the good case and the bad case with the maximal number of complaints. Note that in KZG, a share along with the proof for validating has the size of $224$ Bytes; in the bad case, there are $n^2/2$ shares (with their proofs) to be broadcasted for public verification.

As shown in Figure \ref{fig:bccost}, for our protocol, the costs in the good case and the worst case are very close and grow steadily. For $n = 2^9$, the cost is around $1.05$ MB, while for $n = 2^{15}$, the cost is approximately $61.1$ MB. In contrast, while the good-case KZG has very low broadcast costs, its worst-case costs grow quadratically and would require over $120$ GB when $n = 2^{15}$. 

\begin{figure}
    \centering
    \includegraphics[width=.9\linewidth]{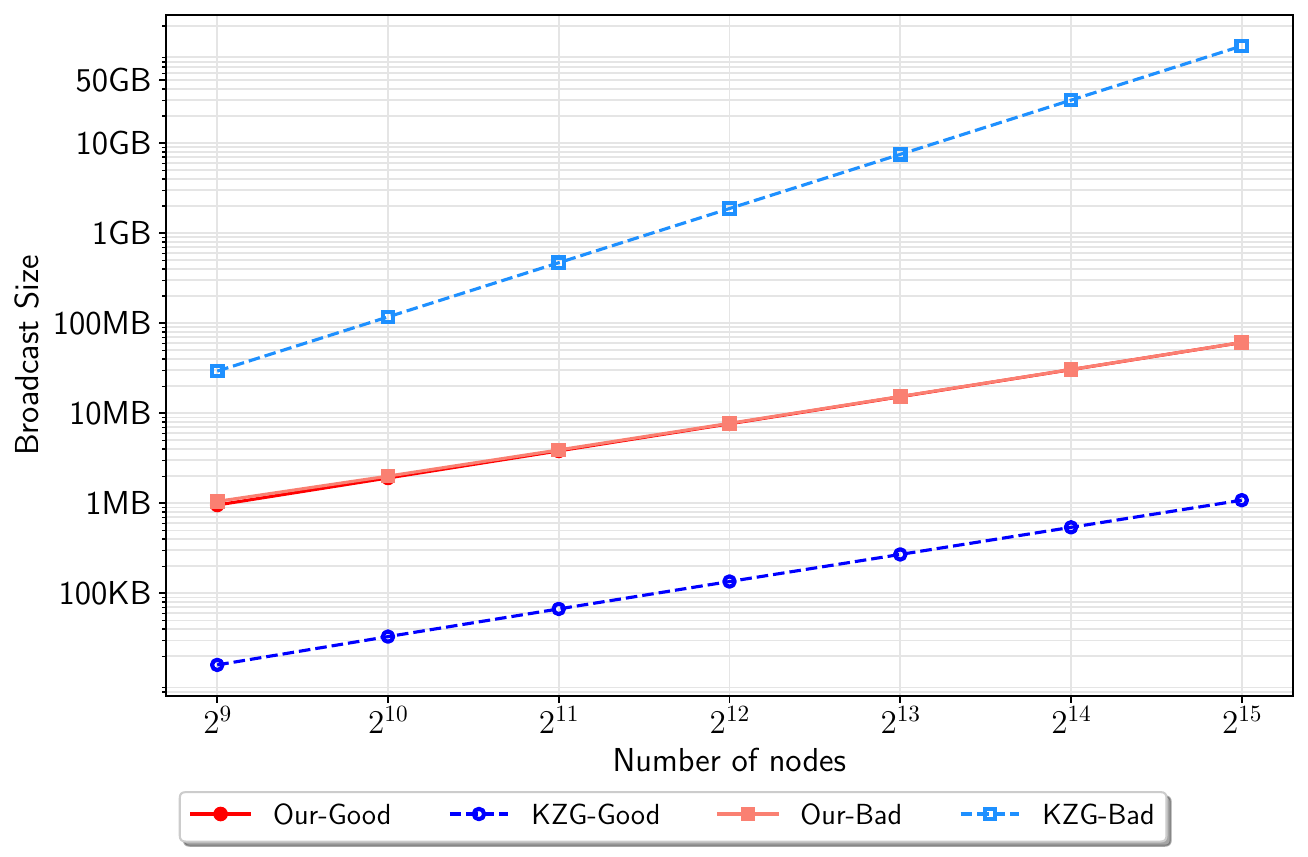} \vspace{-0.4cm}
    \caption{Broadcast Channel Overhead}
    \label{fig:bccost}
    \vspace{-0.4cm}
\end{figure}

\begin{figure}[h]
    \centering
    \includegraphics[width=.9\linewidth]{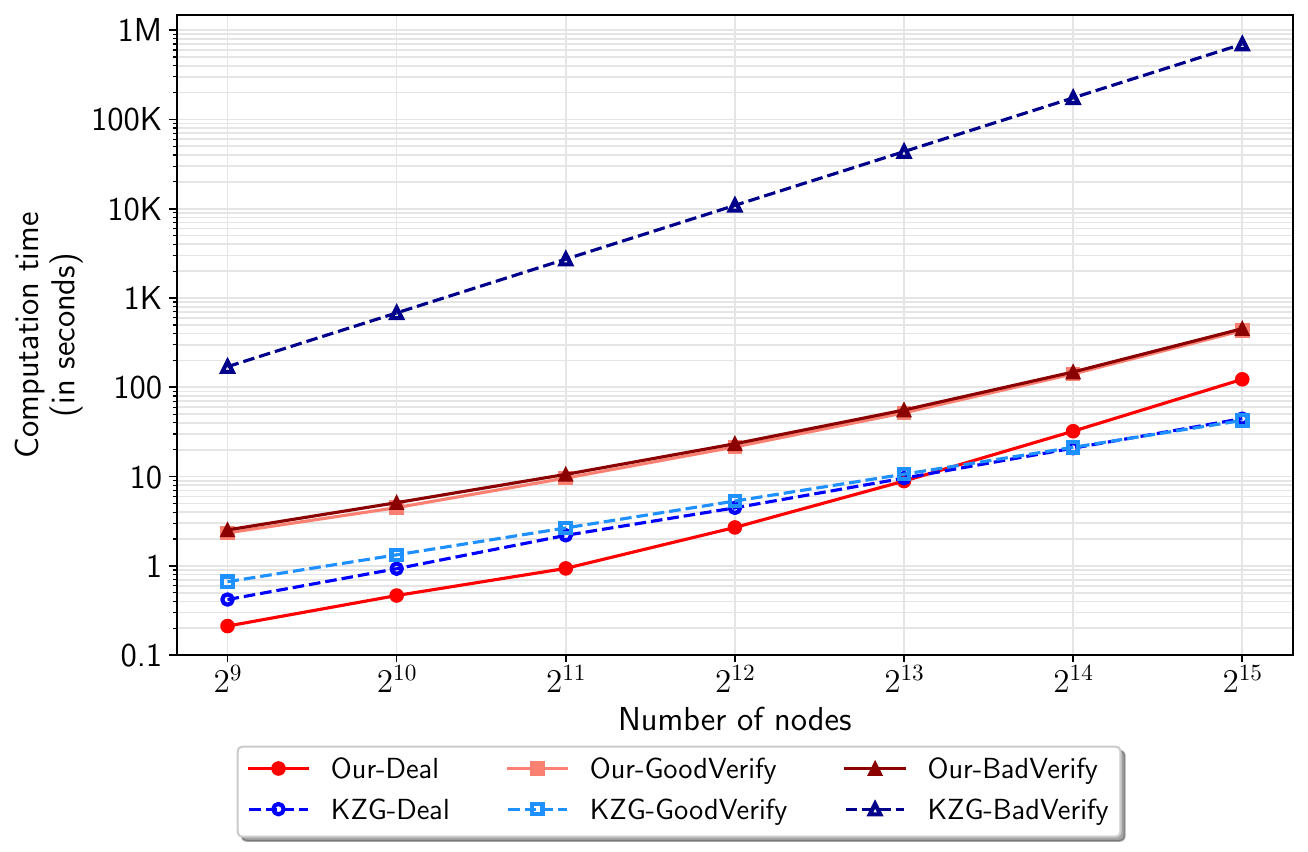} \vspace{-0.4cm}
    \caption{Computation Overhead}
    \label{fig:computationcost}
    \vspace{-0.4cm}
\end{figure}

\smallskip\noindent{\bf Computation time.} We conducted tests to measure the computation time for generating a secret-sharing transcript (Deal) and reaching an agreement on a qualified set (Verify) in both good case and bad case on AWS c5a.large (AMD EYPC 7002 CPU with 2 cores and 4 GB RAM). We compared our results with KZG, utilizing the reported findings from \cite{ZhangXHSZ22} for the good case while estimating the worst-case scenario by assuming that $n^2/2$ shares need to be verified (each share verification takes 1.3 ms). As illustrated in Fig.\ref{fig:computationcost}, in the good case, our protocol's performance is comparable to or even better than KZG, although their programming language (C++) and environment (AWS c5a.24xlarge, AMD EYPC 7002 CPU with 96 cores, and 187 GB RAM) are supposed to be superior to ours. However, in the worst-case scenario, our protocol remains efficient while KZG becomes infeasible. 

Note that our computation time grows faster than KZG's, which we believe is due to the use of a na{\"i}ve implementation of multi-point polynomial evaluation. The complexity of our current implementation is \(O(n^2)\) for evaluating an \(O(n)\)-degree polynomial at \(O(n)\) points. In contrast, the implementation in \cite{ZhangXHSZ22} employs an optimized algorithm whose complexity is \(O(n\log^2 n)\). However, it is important to highlight that our DKG protocol can benefit from the \(O(n\log^2 n)\) polynomial evaluation algorithm as well, and our implementation can be enhanced if a Java implementation for the algorithm becomes available.

%% file: data/related_work.tex

\section{Related Works}\label{ap:relatedworks}
\noindent{\bf VSS-based DKG.} Distributed Key Generation (DKG) has been a prominent area of research for several decades. Pedersen's seminal work \cite{Pedersen91} established the foundation in this field by introducing an efficient protocol for Dlog-based cryptosystems. This protocol builds upon Feldman's Verifiable Secret Sharing (VSS) \cite{Feldman87}. Within this scheme, each participant collaboratively runs $n$ instances of Feldman's VSS, taking on the role of the dealer in one of these instances.

In the VSS framework established by Feldman, the dealer is required to broadcast a commitment to a polynomial while distributing the shares privately among all participants. Given that the commitment's size is proportional to $O(n\secpar)$, the resultant communication overhead becomes $O(n\mathcal{B}(n\secpar))$. Additionally, Pedersen's DKG involves a complaint phase where participants broadcast any grievances against dishonest dealers. If a participant were to lodge multiple complaints concurrently, the communication overhead of this phase is likewise $O(n\mathcal{B}(n\secpar))$. It is vital to highlight that during this phase, each participant may validate up to $O(n^2)$ shares. In Feldman's VSS, the computational effort to validate a single share is equivalent to $O(n)$ group operations. This implies a per-node computational burden before the complaint phase of $O(n^2)$, which can potentially amplify to $O(n^3)$ during the complaint process.

A majority of DKG architectures conform to the joint-VSS model. In essence, any innovative VSS protocol can be adapted into a new DKG protocol. Furthermore, given that VSS can be constructed using polynomial commitments, any polynomial commitment scheme can be evolved into both a VSS and, consequently, a DKG. A significant advancement in this field was made by Kate et al. \cite{KateZG10}, who proposed the first polynomial commitment (abbreviated as KZG) with a commitment size of $O(\secpar)$. This innovation ensures that prior to the complaint phase, the communication overhead can be reduced to $O(n\mathcal{B}(\secpar))$. A notable feature of the KZG polynomial commitment is its efficiency in verifying shares; the computational cost for verifying a single share is a mere $O(1)$. This denotes that the computational overhead for each node, in terms of verification before the complaint phase, is simply $O(n)$ in group operations, but this can rise to $O(n^2)$ during the complaint process. Historically, the computational load for producing a polynomial commitment was believed to be $O(n^2)$ \cite{TomescuCZAPGD20}. However, a recent exploration by Zhang et al. \cite{ZhangXHSZ22} revealed that the computational overhead for generating a KZG commitment can be streamlined to $O(n\log n)$. It's noteworthy that although KZG requires a CRS setup, there have been other efforts \cite{ZhangXHSZ22, YurekLFKM22} that prioritize efficient polynomial commitments without relying on a trusted setup, but these don't match KZG's efficiency.

 \smallskip\noindent{\bf PVSS-based DKG.} \text{Fouque and Stern} \cite{FouqueS01} offered a solution that sidestepped the necessity for a complaint phase by incorporating publicly verifiable secret sharing (PVSS). In the event that a PVSS transcript consists of $O(n)$ ciphertexts, the communication overhead will naturally be $O(n\textsf{BB}_{n}(n\secpar))$ should every participant choose to broadcast this transcript. Historically, the validation of a PVSS transcript required an overhead of $O(n^2)$, suggesting that the per-node computational overhead in DKG might ascend to $O(n^3)$. However, this obstacle was surmounted by Cascudo and David with their Scrape protocol \cite{CascudoD17}, which introduced a PVSS methodology that caps the verification duration at $O(n)$. It's worth highlighting that Scrape's strategy is versatile and can be applied to improve many VSS-based DKG schemes, including that of Pedersen's, ensuring that computational overhead during the complaint phase is kept at $O(n^2)$ and doesn't spike to $O(n^3)$. A few recent works focus on improving the concrete performance of PVSS schemes, including the lattice-based PVSS \cite{GentryHL22}, Groth's PVSS \cite{Groth21}, and PVSS using class groups \cite{KateMMST23}.

 \smallskip\noindent{\bf Aggregatable-PVSS-based DKG.} Aggregatable PVSS schemes \cite{GurkanJMMST21} are PVSS schemes whose transcripts can be concisely merged into one.  
 There are a few designs that leverage customized communication protocols rather than simply leveraging Byzantine broadcast protocols (or broadcast channels), enjoying asymptotically better complexity. Notably, \text{Gurkan et al.} \cite{GurkanJMMST21} leveraged an aggregatable PVSS combined with gossip protocols to craft a publicly verifiable DKG. Their communication overhead is streamlined to $n\mathcal{B}(\secpar)+\log n\cdot\mathcal{B}(n\secpar)$ as opposed to $n\mathcal{B}(n\secpar)$, with their per-node communication overhead being $O(n\log^2n)$. It's pertinent to note, however, that their model can only accommodate $O(\log n)$ Byzantine nodes. Very recently, Feng et al.\cite{FengLT24} and Bacho et al. \cite{BachoLLOP23} leverage specially designed communication protocols together with aggregatable PVSS schemes and present DKG schemes with sub-quadratic per-party computation/communication cost while enjoying optimal resilience. 
 
 Note that existing aggregatable PVSS schemes all produce secrets in an Elliptic curve group, thus incompatible with many threshold cryptographic protocols. Feng et al. \cite{FengLT24} also give a variant of DKG using conventional PVSS schemes but with slightly higher (still sub-quadratic) per-party complexity.

 \smallskip\noindent{\bf DKG in the YOSO model.} A common strategy to enhance scalability is selecting a committee and executing the threshold cryptographic systems within this smaller subset. However, this approach is fraught with challenges. Once aware of the committee's composition, an adaptive adversary can compromise the entire group, thereby undermining security. Furthermore, given that each member of the committee is required to contribute multiple times during both key generation and subsequent threshold operations, methods like silent committee sampling (e.g., using a verifiable random function \cite{ChenM19}) and assuming memory erasure fail to provide protection against adaptive attackers. Recent advances in the YOSO (You-Only-Speak-Once) MPC realm \cite{GentryHKMNRY21, BenhamoudaHKMR22} hint at potential solutions to deter adaptive adversaries targeting the committee. Benhamouda et al. \cite{BenhamoudaHKMR22} presents a DKG in the YOSO model. However, the YOSO techniques come with their own set of challenges. Notably, existing YOSO techniques (if without using resource-intensive tools like fully homomorphic encryption \cite{GentryHMNY21} ) need to sample a huge committee, say with a few or tens of thousands of nodes, which is already as large as the network scale we are interested in, let alone the extra overhead incurred by using YOSO techniques. Furthermore, as successive committees remain anonymous, inter-committee communication is heavily dependent on a broadcast channel.

\smallskip\noindent{\bf On the security of DKG.} Beyond endeavors aimed at bolstering the efficiency of DKG, various research initiatives have tackled this challenge from different perspectives. Gennaro et al. \cite{GennaroJKR07} pinpointed vulnerabilities in Pedersen's DKG where the secret key distribution could be manipulated by adversaries. They addressed this flaw by achieving complete secrecy, albeit with a higher computational overhead. Gurkan et al. \cite{GurkanJMMST21} conceptualized a milder form of secrecy, coined as ``key-expressability", which assumes that adversaries can influence key distribution but within predetermined constraints. They postulated that a key-expressable DKG suffices for many applications, with multiple DKG architectures, including Pedersen's \cite{Pedersen91}, Fouque-Stern's \cite{FouqueS01}, and our own, fitting this criteria. Another remarkable contribution by Canetti et al. \cite{CanettiGJKR99} introduced a DKG protocol with adaptive security, a departure from our model and numerous others that ensure security only against static adversaries. Bacho and Loss's recent work \cite{BachoL22} put forth an oracle-aided adaptive definition and ascertained that several protocols, including \cite{Pedersen91, FouqueS01}, conform to this definition in the algebraic group model. Our model also complies with this adaptive security definition.

\smallskip\noindent{\bf Asynchronous DKG.} Lastly, some recent research efforts \cite{GaoLLTXZ22, AbrahamJMMS23, DasYXMK022} have pivoted towards DKG within asynchronous networks. These designs adopt the joint-VSS blueprint and depend on an asynchronous broadcast protocol, referred to as ``reliable broadcast" \cite{bracha1987asynchronous}, to guarantee verifiability, yet they encounter the cubic computational challenge. Notably, Das et al. \cite{DasYXMK022} showcased the inaugural asynchronous DKG with a communication overhead of $O(n^3\secpar)$ for field-element secrets, whereas Abraham et al. \cite{AbrahamJMMS23} furnished an adaptively secure asynchronous DKG with identical complexity.

%% file: data/conclusion.tex
\section{Conclusion and Future Direction}\label{sect:conclusion}
To enable large-scale threshold cryptographic applications, such as Filecoin's checkpointing, we present an adaptively secure DKG protocol with (quasi)linear per-node communication and computation costs. The key idea is to use a common coin, well-established in blockchain settings, to sample a small any-trust committee for secret contributions, which suffices for DKG security. The main challenge is countering adaptive attackers, which we address through carefully applied techniques. We also introduce a method to deploy our DKG in a weighted setting without compromising efficiency.

Our approach assumes a synchronous network, leading to a future challenge: designing a scalable, adaptively secure DKG for asynchronous settings that better reflect real Internet conditions. Another broader question is how to adapt DKG and threshold protocols to other cryptographic systems, like lattice-based ones, while maintaining our performance and security metrics.

%% file: app/checkpoint-app.tex

\section{Supplementary Materials for Bitcoin Checkpointing}\label{app:checkapp}

\subsection{The Blueprint of Pikachu}
\noindent{\bf Long-range attacks against PoS blockchain.} Unlike proof-of-work chains, block creation in PoS systems is both costless (in terms of physical resources like energy) and timeless (unconstrained by time limits), which enables adversaries to easily fork a chain. Existing PoS chains prevent malicious forking by punishing misbehavior validators. However, an attacker can choose to present the fork chain after all its stakes have been withdrawn, thus free of being slashed, 
What is worse, a late coming client may not be able to decide the canonical chain among the forks.

\smallskip\noindent{\bf Securing PoS with Bitcoin checkpointing.} A few works \cite{AzouviV22, TasTGKMY23} have shown that long-range attacks can be effectively mitigated by creating checkpoints of the PoS chain on a PoW chain, such that a late coming client can distinguish the canonical chain among forks. Pikachu illustrates a threshold signature-based checkpointing mechanism.
At a high level, the lifetime of the PoS system is divided into multiple epochs, and checkpoints are supposed to be created per epoch. At every epoch $i$, a configuration $C_i=\{(\mathcal{V}_{i,j},w_{i,j})\}_{j\in[n_i]}$ for some integer $n_i$, which is the set of all validators $\{\mathcal{V}_{i,j}\}_{j\in[n_i]}$ with their weights $\{w_{i,j}\}_{j\in[n_i]}$, is associated with a public key $Q_i$ (w.r.t. Schnorr signature scheme) which can serve as a Bitcoin address, while the secret key of $Q_i$ is secretly shared among $C_i$. At epoch $i+1$, validators in $C_i$ will jointly create a Bitcoin transaction which transfers all assets on $Q_i$ to $Q_{i+1}$, the address belongs to the current configuration $C_{i+1}$; This transaction is the checkpoint. We elucidate their design with the following three algorithms/protocols\footnote{Slightly different from their original description where the PoS digest is embedded into the Bitcoin address, we choose to put it in $\mathsf{OP\_RETURN}$ for simplicity.}.

\begin{itemize}
    \item $\mathsf{AllocateSubID}(C)\rightarrow \{d_j\}_{j\in[n]}$.The sub-identity allocation algorithm takes input as a configuration $$C =\{(\mathcal{V}_{j},w_{j})\}_{j\in[n]}$$ and determines the number of sub-identities $d_j$ for each $\mathcal{V}_j$ according to their weight $w_j$.
    \item $\mathsf{DKG}(\{(\mathcal{V}_{j}, d_j)\}_{j\in[n]})$.
    The validators in $C$ run a DKG protocol, while each sub-identity is viewed as an independent participant. Therefore, each validator $\mathcal{V}_j$ obtains $d_j$ pairs of $(pk_{j,z},sk_{j,z})_{z\in[d_j]}$, and all validators obtain the same public key $Q=pk$ and the list of public key shares $\vec{pk}= (pk_{j,z})_{j\in[n], z\in[d_j]}$.
    \item $\mathsf{CreateCKP}(C_i,\textsf{ckp},\textsf{PreAdd},Q_{i+1})\rightarrow \textsf{TX}.$ At the epoch $i+1$, assume that validators in $C_{i+1}$ have generated the public key $Q_{i+1}$, the digest of PoS block to be checkpointed is $\textsf{ckp}$, and the address of the last checkpointing Bitcoin transaction is $\textsf{PreAdd}$. Then, the validators in $C_i$ invoke a Threshold Schnorr protocol to sign a Bitcoin transaction $\textsf{TX}$ with the following information.
    \begin{equation*}
        \{\textsf{Input}: \textsf{PreAdd}; \textsf{Output}: Q_{i+1}; \textsf{OP\_Return}: \textsf{ckp}\}.
    \end{equation*}
    Once the transaction has been properly signed, every validator should disseminate it to the Bitcoin network.
\end{itemize}

With checkpoints on Bitcoin, it is rather straightforward for a late-coming user to decide which fork is the canonical chain when the user is provided with a block tree of finalized PoS blocks. Specifically, the user first synchronizes with the Bitcoin blockchain. Then, it finds the initial checkpoint transaction and builds a chain of transactions following the initial transaction. Next, it obtains the digest $\textsf{ckp}$ from the latest checkpoint transaction and decides the fork with the block whose digest is $\textsf{ckp}$ as the canonical chain. Moreover, while other approaches like key-evolving forward-secure signatures \cite{ChenM19, DavidGKR18} may also mitigate long-range attacks, the checkpointing mechanism enjoys the unique advantage of ensuring malicious validators are always slashable. We defer a detailed discussion to Sect.\ref{ap:baby}.

\subsection{Security of Checkpointing}\label{ap:baby}
This paradigm has been thoroughly analyzed in \cite{AzouviV22}. It considers an efficient adversary $\adv$, which at each epoch $i$ can corrupt all validators in previous configurations $\{C_{j}\}_{j<i-L}$ and a fraction of validators up to $f$ in ``recent" configurations $\{C_j\}_{i-L<j\leq i}$, for some parameter $L$ such that the checkpoint transaction for epoch $i_0$ will be confirmed in Bitcoin by epoch $i_0+L$.
Such an adversary can mount long-range attacks by using the previous secret keys to forge another validate-looking chain (called a long-range attack chain). However, since the Bitcoin blockchain has recorded transactions that transferred all assets from previous addresses $\{Q_j\}_{j<i-L}$, $\adv$ cannot create valid checkpoints using secret keys of $\{Q_j\}_{j<i-L}$. Therefore, a bootstrapping client can decide the canonical chain with Bitcoin checkpoints.
We summarize their results in the following theorem.
\begin{theorem}[\cite{AzouviV22}]
    Assume both the Bitcoin blockchain and the PoS chain satisfy consistency, chain growth, and chain quality (as defined in \cite{GarayKL15}). Assume the Threshold Schnorr signature satisfies unforgeability and robustness under the DKG protocol against $\adv$ corrupting up to $t$ sub-identities,  and $\mathsf{AllocateSubID}$ allocates at most $t$ sub-IDs to $\adv$. Then, the checkpointing mechanism satisfies the following properties.
    \begin{itemize}
        \item Safety. $\adv$ cannot produce any valid checkpointing transactions for long-range attack chains.
        \item Liveness. $\adv$ cannot stop the checkpoints from happening.
    \end{itemize}
\end{theorem}

\smallskip\noindent{\bf On Slashable Safety.} Babylon claims the slashable safety. Specifically, for a PoS system with $3t+1$ units of stake, validators with at least $t$ units should become slashable in the view of all honest validators whenever there is a safety violation. Many PoS systems offer slashable safety against \textit{short-range} attacks by locking validators' stakes for a period and slashing one's stake once proof of security violation is presented. However, long-range attackers can evade being slashed by publishing the attack chain after withdrawing their stakes from the canonical chain. 

It has been proved in \cite{TasTGKMY23} that slashable safety against long-range attacks is impossible without external trust. With this result, \cite{TasTGKMY23} also shows that other approaches for mitigating long-range attacks, such as key-evolving signatures \cite{BadertscherGKRZ18, ChenM19} cannot provide slashable safety. Nonetheless, leveraging the Bitcoin blockchain as an external trust can certainly bypass this impossibility. Assuming that checkpoints for the canonical PoS chain have been properly posted on the Bitcoin blockchain,
the attacker cannot present an attack chain that diverges from the canonical chain before the latest checkpoint. In this case, the attacker must not have withdrawn its stakes and thus is slashable. 

In light of the above, both ours/Pikachu and Babylon can guarantee slashable safety once the checkpoints have been properly created. Now, we turn to examine the case in which checkpoints may not be generated correctly. The adversary has the following options: (1) not make a checkpoint; (2) make a checkpoint for an ill-formed block; (3) make more than one checkpoint for different well-formed blocks at the same height and hide the block whose checkpoint appears earlier;
(4) make a checkpoint for a well-formed block but exclude some valid transactions (for censorship). Babylon introduces an \textit{emergency break} to prevent from (2) and (3). The client can notice these attacks happening and then no longer process this chain. In case the adversary refuses to participate in the checkpoint creation, Babylon considered the punishment of inactivity, which enables the removal of the inactive validators. Regarding censorship resistance (4), Babylon proposed a roll-up technique that is orthogonal to the checkpointing mechanism. 

In our system, as all checkpoints are in the chain of transactions, the adversary cannot mount the attack of (3). For (2) and (4), we can follow the exact same approach as Babylon does. For the attack of (1), it may be hard to identify who makes the DKG/threshold signing fail. Instead, we require a checkpoint to be made by a certain height of the Bitcoin blockchain, and then the client can switch to \textit{emergency} break when it does not find a valid checkpoint by the designated position. In summary, our checkpointing mechanism provides slashable safety as long as honest clients do not switch to emergency breaks.

%% file: main.bbl

\begin{thebibliography}{70}


\ifx \showCODEN    \undefined \def \showCODEN     #1{\unskip}     \fi
\ifx \showDOI      \undefined \def \showDOI       #1{#1}\fi
\ifx \showISBNx    \undefined \def \showISBNx     #1{\unskip}     \fi
\ifx \showISBNxiii \undefined \def \showISBNxiii  #1{\unskip}     \fi
\ifx \showISSN     \undefined \def \showISSN      #1{\unskip}     \fi
\ifx \showLCCN     \undefined \def \showLCCN      #1{\unskip}     \fi
\ifx \shownote     \undefined \def \shownote      #1{#1}          \fi
\ifx \showarticletitle \undefined \def \showarticletitle #1{#1}   \fi
\ifx \showURL      \undefined \def \showURL       {\relax}        \fi
\providecommand\bibfield[2]{#2}
\providecommand\bibinfo[2]{#2}
\providecommand\natexlab[1]{#1}
\providecommand\showeprint[2][]{arXiv:#2}

\bibitem[Abraham et~al\mbox{.}(2023)]%
        {AbrahamJMMS23}
\bibfield{author}{\bibinfo{person}{Ittai Abraham}, \bibinfo{person}{Philipp
  Jovanovic}, \bibinfo{person}{Mary Maller}, \bibinfo{person}{Sarah
  Meiklejohn}, {and} \bibinfo{person}{Gilad Stern}.}
  \bibinfo{year}{2023}\natexlab{}.
\newblock \showarticletitle{Bingo: Adaptivity and Asynchrony in Verifiable
  Secret Sharing and Distributed Key Generation}. In
  \bibinfo{booktitle}{\emph{{CRYPTO} {(1)}}} \emph{(\bibinfo{series}{LNCS},
  Vol.~\bibinfo{volume}{14081})}. \bibinfo{publisher}{Springer},
  \bibinfo{pages}{39--70}.
\newblock


\bibitem[Azouvi and Vukolic(2022)]%
        {AzouviV22}
\bibfield{author}{\bibinfo{person}{Sarah Azouvi} {and} \bibinfo{person}{Marko
  Vukolic}.} \bibinfo{year}{2022}\natexlab{}.
\newblock \showarticletitle{Pikachu: Securing PoS Blockchains from Long-Range
  Attacks by Checkpointing into Bitcoin PoW using Taproot}. In
  \bibinfo{booktitle}{\emph{ConsensusDay@CCS}}. \bibinfo{publisher}{{ACM}},
  \bibinfo{pages}{53--65}.
\newblock


\bibitem[Bacho et~al\mbox{.}(2023)]%
        {BachoLLOP23}
\bibfield{author}{\bibinfo{person}{Renas Bacho}, \bibinfo{person}{Christoph
  Lenzen}, \bibinfo{person}{Julian Loss}, \bibinfo{person}{Simon
  Ochsenreither}, {and} \bibinfo{person}{Dimitrios Papachristoudis}.}
  \bibinfo{year}{2023}\natexlab{}.
\newblock \showarticletitle{GRandLine: Adaptively Secure {DKG} and Randomness
  Beacon with (Almost) Quadratic Communication Complexity}.
\newblock \bibinfo{journal}{\emph{{IACR} Cryptol. ePrint Arch.}}
  (\bibinfo{year}{2023}), \bibinfo{pages}{1887}.
\newblock
\urldef\tempurl%
\url{https://eprint.iacr.org/2023/1887}
\showURL{%
\tempurl}


\bibitem[Bacho and Loss(2022)]%
        {BachoL22}
\bibfield{author}{\bibinfo{person}{Renas Bacho} {and} \bibinfo{person}{Julian
  Loss}.} \bibinfo{year}{2022}\natexlab{}.
\newblock \showarticletitle{On the Adaptive Security of the Threshold {BLS}
  Signature Scheme}. In \bibinfo{booktitle}{\emph{{CCS}}}.
  \bibinfo{publisher}{{ACM}}, \bibinfo{pages}{193--207}.
\newblock


\bibitem[Bacho and Loss(2023)]%
        {BachoL23}
\bibfield{author}{\bibinfo{person}{Renas Bacho} {and} \bibinfo{person}{Julian
  Loss}.} \bibinfo{year}{2023}\natexlab{}.
\newblock \showarticletitle{Adaptively Secure (Aggregatable) {PVSS} and
  Application to Distributed Randomness Beacons}. In
  \bibinfo{booktitle}{\emph{{CCS}}}. \bibinfo{publisher}{{ACM}},
  \bibinfo{pages}{1791--1804}.
\newblock


\bibitem[Badertscher et~al\mbox{.}(2018)]%
        {BadertscherGKRZ18}
\bibfield{author}{\bibinfo{person}{Christian Badertscher},
  \bibinfo{person}{Peter Gazi}, \bibinfo{person}{Aggelos Kiayias},
  \bibinfo{person}{Alexander Russell}, {and} \bibinfo{person}{Vassilis Zikas}.}
  \bibinfo{year}{2018}\natexlab{}.
\newblock \showarticletitle{Ouroboros Genesis: Composable Proof-of-Stake
  Blockchains with Dynamic Availability}. In \bibinfo{booktitle}{\emph{{CCS}}}.
  \bibinfo{publisher}{{ACM}}, \bibinfo{pages}{913--930}.
\newblock


\bibitem[Bellare et~al\mbox{.}(2007)]%
        {BellareBKS07}
\bibfield{author}{\bibinfo{person}{Mihir Bellare}, \bibinfo{person}{Alexandra
  Boldyreva}, \bibinfo{person}{Kaoru Kurosawa}, {and} \bibinfo{person}{Jessica
  Staddon}.} \bibinfo{year}{2007}\natexlab{}.
\newblock \showarticletitle{Multirecipient Encryption Schemes: How to Save on
  Bandwidth and Computation Without Sacrificing Security}.
\newblock \bibinfo{journal}{\emph{{IEEE} Trans. Inf. Theory}}
  \bibinfo{volume}{53}, \bibinfo{number}{11} (\bibinfo{year}{2007}),
  \bibinfo{pages}{3927--3943}.
\newblock


\bibitem[Bellare et~al\mbox{.}(2022)]%
        {BellareCKMTZ22}
\bibfield{author}{\bibinfo{person}{Mihir Bellare},
  \bibinfo{person}{Elizabeth~C. Crites}, \bibinfo{person}{Chelsea Komlo},
  \bibinfo{person}{Mary Maller}, \bibinfo{person}{Stefano Tessaro}, {and}
  \bibinfo{person}{Chenzhi Zhu}.} \bibinfo{year}{2022}\natexlab{}.
\newblock \showarticletitle{Better than Advertised Security for Non-interactive
  Threshold Signatures}. In \bibinfo{booktitle}{\emph{{CRYPTO} {(4)}}}
  \emph{(\bibinfo{series}{LNCS}, Vol.~\bibinfo{volume}{13510})}.
  \bibinfo{publisher}{Springer}, \bibinfo{pages}{517--550}.
\newblock


\bibitem[Benhamouda et~al\mbox{.}(2020)]%
        {BenhamoudaG0HK020}
\bibfield{author}{\bibinfo{person}{Fabrice Benhamouda}, \bibinfo{person}{Craig
  Gentry}, \bibinfo{person}{Sergey Gorbunov}, \bibinfo{person}{Shai Halevi},
  \bibinfo{person}{Hugo Krawczyk}, \bibinfo{person}{Chengyu Lin},
  \bibinfo{person}{Tal Rabin}, {and} \bibinfo{person}{Leonid Reyzin}.}
  \bibinfo{year}{2020}\natexlab{}.
\newblock \showarticletitle{Can a Public Blockchain Keep a Secret?}. In
  \bibinfo{booktitle}{\emph{{TCC} {(1)}}} \emph{(\bibinfo{series}{LNCS},
  Vol.~\bibinfo{volume}{12550})}. \bibinfo{publisher}{Springer},
  \bibinfo{pages}{260--290}.
\newblock


\bibitem[Benhamouda et~al\mbox{.}(2022)]%
        {BenhamoudaHKMR22}
\bibfield{author}{\bibinfo{person}{Fabrice Benhamouda}, \bibinfo{person}{Shai
  Halevi}, \bibinfo{person}{Hugo Krawczyk}, \bibinfo{person}{Alex Miao}, {and}
  \bibinfo{person}{Tal Rabin}.} \bibinfo{year}{2022}\natexlab{}.
\newblock \showarticletitle{Threshold Cryptography as a Service (in the
  Multiserver and {YOSO} Models)}. In \bibinfo{booktitle}{\emph{{CCS}}}.
  \bibinfo{publisher}{{ACM}}, \bibinfo{pages}{323--336}.
\newblock


\bibitem[Blockchain({[n.\,d.]})]%
        {Aptos}
\bibfield{author}{\bibinfo{person}{Aptos Blockchain}.}
  \bibinfo{year}{[n.\,d.]}\natexlab{}.
\newblock \bibinfo{title}{Aptosan}.
\newblock \bibinfo{howpublished}{\url{https://aptoscan.com}}.
\newblock


\bibitem[Boneh et~al\mbox{.}(2001)]%
        {BonehLS01}
\bibfield{author}{\bibinfo{person}{Dan Boneh}, \bibinfo{person}{Ben Lynn},
  {and} \bibinfo{person}{Hovav Shacham}.} \bibinfo{year}{2001}\natexlab{}.
\newblock \showarticletitle{Short Signatures from the Weil Pairing}. In
  \bibinfo{booktitle}{\emph{{ASIACRYPT}}} \emph{(\bibinfo{series}{LNCS},
  Vol.~\bibinfo{volume}{2248})}. \bibinfo{publisher}{Springer},
  \bibinfo{pages}{514--532}.
\newblock


\bibitem[Bracha(1987)]%
        {bracha1987asynchronous}
\bibfield{author}{\bibinfo{person}{Gabriel Bracha}.}
  \bibinfo{year}{1987}\natexlab{}.
\newblock \showarticletitle{Asynchronous Byzantine agreement protocols}.
\newblock \bibinfo{journal}{\emph{Information and Computation}}
  \bibinfo{volume}{75}, \bibinfo{number}{2} (\bibinfo{year}{1987}),
  \bibinfo{pages}{130--143}.
\newblock


\bibitem[Brunetta et~al\mbox{.}(2024)]%
        {BrunettaHS24}
\bibfield{author}{\bibinfo{person}{Carlo Brunetta}, \bibinfo{person}{Hans
  Heum}, {and} \bibinfo{person}{Martijn Stam}.}
  \bibinfo{year}{2024}\natexlab{}.
\newblock \showarticletitle{SoK: Public Key Encryption with Openings}. In
  \bibinfo{booktitle}{\emph{{PKC} {(4)}}} \emph{(\bibinfo{series}{LNCS},
  Vol.~\bibinfo{volume}{14604})}. \bibinfo{publisher}{Springer},
  \bibinfo{pages}{35--68}.
\newblock


\bibitem[Cachin et~al\mbox{.}(2005)]%
        {CachinKS05}
\bibfield{author}{\bibinfo{person}{Christian Cachin}, \bibinfo{person}{Klaus
  Kursawe}, {and} \bibinfo{person}{Victor Shoup}.}
  \bibinfo{year}{2005}\natexlab{}.
\newblock \showarticletitle{Random Oracles in Constantinople: Practical
  Asynchronous Byzantine Agreement Using Cryptography}.
\newblock \bibinfo{journal}{\emph{J. Cryptol.}} \bibinfo{volume}{18},
  \bibinfo{number}{3} (\bibinfo{year}{2005}), \bibinfo{pages}{219--246}.
\newblock


\bibitem[Canetti et~al\mbox{.}(1999)]%
        {CanettiGJKR99}
\bibfield{author}{\bibinfo{person}{Ran Canetti}, \bibinfo{person}{Rosario
  Gennaro}, \bibinfo{person}{Stanislaw Jarecki}, \bibinfo{person}{Hugo
  Krawczyk}, {and} \bibinfo{person}{Tal Rabin}.}
  \bibinfo{year}{1999}\natexlab{}.
\newblock \showarticletitle{Adaptive Security for Threshold Cryptosystems}. In
  \bibinfo{booktitle}{\emph{{CRYPTO}}} \emph{(\bibinfo{series}{LNCS},
  Vol.~\bibinfo{volume}{1666})}. \bibinfo{publisher}{Springer},
  \bibinfo{pages}{98--115}.
\newblock


\bibitem[Cascudo and David(2017)]%
        {CascudoD17}
\bibfield{author}{\bibinfo{person}{Ignacio Cascudo} {and}
  \bibinfo{person}{Bernardo David}.} \bibinfo{year}{2017}\natexlab{}.
\newblock \showarticletitle{{SCRAPE:} Scalable Randomness Attested by Public
  Entities}. In \bibinfo{booktitle}{\emph{{ACNS}}}
  \emph{(\bibinfo{series}{LNCS}, Vol.~\bibinfo{volume}{10355})}.
  \bibinfo{publisher}{Springer}, \bibinfo{pages}{537--556}.
\newblock


\bibitem[Cerulli et~al\mbox{.}(2023)]%
        {CerulliCNPS23}
\bibfield{author}{\bibinfo{person}{Andrea Cerulli}, \bibinfo{person}{Aisling
  Connolly}, \bibinfo{person}{Gregory Neven}, \bibinfo{person}{Franz{-}Stefan
  Preiss}, {and} \bibinfo{person}{Victor Shoup}.}
  \bibinfo{year}{2023}\natexlab{}.
\newblock \bibinfo{title}{vetKeys: How a Blockchain Can Keep Many Secrets}.
\newblock \bibinfo{howpublished}{Cryptology ePrint Archive, Paper 2023/616}.
\newblock
\newblock
\shownote{\url{https://eprint.iacr.org/2023/616}}.


\bibitem[Chase and Lysyanskaya(2006)]%
        {ChaseL06}
\bibfield{author}{\bibinfo{person}{Melissa Chase} {and} \bibinfo{person}{Anna
  Lysyanskaya}.} \bibinfo{year}{2006}\natexlab{}.
\newblock \showarticletitle{On Signatures of Knowledge}. In
  \bibinfo{booktitle}{\emph{{CRYPTO}}} \emph{(\bibinfo{series}{LNCS},
  Vol.~\bibinfo{volume}{4117})}. \bibinfo{publisher}{Springer},
  \bibinfo{pages}{78--96}.
\newblock


\bibitem[Chaum and Pedersen(1992)]%
        {ChaumP92}
\bibfield{author}{\bibinfo{person}{David Chaum} {and}
  \bibinfo{person}{Torben~P. Pedersen}.} \bibinfo{year}{1992}\natexlab{}.
\newblock \showarticletitle{Wallet Databases with Observers}. In
  \bibinfo{booktitle}{\emph{{CRYPTO}}} \emph{(\bibinfo{series}{LNCS},
  Vol.~\bibinfo{volume}{740})}. \bibinfo{publisher}{Springer},
  \bibinfo{pages}{89--105}.
\newblock


\bibitem[Chen and Micali(2019)]%
        {ChenM19}
\bibfield{author}{\bibinfo{person}{Jing Chen} {and} \bibinfo{person}{Silvio
  Micali}.} \bibinfo{year}{2019}\natexlab{}.
\newblock \showarticletitle{Algorand: {A} secure and efficient distributed
  ledger}.
\newblock \bibinfo{journal}{\emph{Theor. Comput. Sci.}}  \bibinfo{volume}{777}
  (\bibinfo{year}{2019}), \bibinfo{pages}{155--183}.
\newblock


\bibitem[Choi et~al\mbox{.}(2023)]%
        {ChoiMB23}
\bibfield{author}{\bibinfo{person}{Kevin Choi}, \bibinfo{person}{Aathira
  Manoj}, {and} \bibinfo{person}{Joseph Bonneau}.}
  \bibinfo{year}{2023}\natexlab{}.
\newblock \showarticletitle{SoK: Distributed Randomness Beacons}. In
  \bibinfo{booktitle}{\emph{{SP}}}. \bibinfo{publisher}{{IEEE}},
  \bibinfo{pages}{75--92}.
\newblock


\bibitem[Cosmos({[n.\,d.]})]%
        {Cosmos}
\bibfield{author}{\bibinfo{person}{Cosmos}.}
  \bibinfo{year}{[n.\,d.]}\natexlab{}.
\newblock
\newblock
\newblock
\shownote{https://cosmos.network}.


\bibitem[Crites et~al\mbox{.}(2023)]%
        {CritesKM23}
\bibfield{author}{\bibinfo{person}{Elizabeth~C. Crites},
  \bibinfo{person}{Chelsea Komlo}, {and} \bibinfo{person}{Mary Maller}.}
  \bibinfo{year}{2023}\natexlab{}.
\newblock \showarticletitle{Fully Adaptive Schnorr Threshold Signatures}. In
  \bibinfo{booktitle}{\emph{{CRYPTO} {(1)}}} \emph{(\bibinfo{series}{LNCS},
  Vol.~\bibinfo{volume}{14081})}. \bibinfo{publisher}{Springer},
  \bibinfo{pages}{678--709}.
\newblock


\bibitem[Das et~al\mbox{.}(2022)]%
        {DasYXMK022}
\bibfield{author}{\bibinfo{person}{Sourav Das}, \bibinfo{person}{Thomas Yurek},
  \bibinfo{person}{Zhuolun Xiang}, \bibinfo{person}{Andrew Miller},
  \bibinfo{person}{Lefteris Kokoris{-}Kogias}, {and} \bibinfo{person}{Ling
  Ren}.} \bibinfo{year}{2022}\natexlab{}.
\newblock \showarticletitle{Practical Asynchronous Distributed Key Generation}.
  In \bibinfo{booktitle}{\emph{{SP}}}. \bibinfo{publisher}{{IEEE}},
  \bibinfo{pages}{2518--2534}.
\newblock


\bibitem[David et~al\mbox{.}(2018)]%
        {DavidGKR18}
\bibfield{author}{\bibinfo{person}{Bernardo David}, \bibinfo{person}{Peter
  Gazi}, \bibinfo{person}{Aggelos Kiayias}, {and} \bibinfo{person}{Alexander
  Russell}.} \bibinfo{year}{2018}\natexlab{}.
\newblock \showarticletitle{Ouroboros Praos: An Adaptively-Secure,
  Semi-synchronous Proof-of-Stake Blockchain}. In
  \bibinfo{booktitle}{\emph{{EUROCRYPT} {(2)}}} \emph{(\bibinfo{series}{LNCS},
  Vol.~\bibinfo{volume}{10821})}. \bibinfo{publisher}{Springer},
  \bibinfo{pages}{66--98}.
\newblock


\bibitem[de~Souza and Tonkikh(2023)]%
        {abs-2307-15561}
\bibfield{author}{\bibinfo{person}{Luciano~Freitas de Souza} {and}
  \bibinfo{person}{Andrei Tonkikh}.} \bibinfo{year}{2023}\natexlab{}.
\newblock \showarticletitle{Swiper and Dora: efficient solutions to weighted
  distributed problems}.
\newblock \bibinfo{journal}{\emph{CoRR}}  \bibinfo{volume}{abs/2307.15561}
  (\bibinfo{year}{2023}).
\newblock


\bibitem[Dolev and Reischuk(1982)]%
        {DolevR82}
\bibfield{author}{\bibinfo{person}{Danny Dolev} {and}
  \bibinfo{person}{R{\"{u}}diger Reischuk}.} \bibinfo{year}{1982}\natexlab{}.
\newblock \showarticletitle{Bounds on Information Exchange for Byzantine
  Agreement}. In \bibinfo{booktitle}{\emph{{PODC}}}.
  \bibinfo{publisher}{{ACM}}, \bibinfo{pages}{132--140}.
\newblock


\bibitem[Feldman(1987)]%
        {Feldman87}
\bibfield{author}{\bibinfo{person}{Paul Feldman}.}
  \bibinfo{year}{1987}\natexlab{}.
\newblock \showarticletitle{A Practical Scheme for Non-interactive Verifiable
  Secret Sharing}. In \bibinfo{booktitle}{\emph{{FOCS}}}.
  \bibinfo{publisher}{{IEEE} Computer Society}, \bibinfo{pages}{427--437}.
\newblock


\bibitem[Feng et~al\mbox{.}(2023)]%
        {FengLT24}
\bibfield{author}{\bibinfo{person}{Hanwen Feng}, \bibinfo{person}{Zhenliang
  Lu}, {and} \bibinfo{person}{Qiang Tang}.} \bibinfo{year}{2023}\natexlab{}.
\newblock \bibinfo{title}{Breaking the Cubic Barrier: Distributed Key and
  Randomness Generation through Deterministic Sharding}.
\newblock \bibinfo{howpublished}{Cryptology ePrint Archive, Paper 2024/168}.
\newblock
\newblock
\shownote{\url{https://eprint.iacr.org/2024/168}}.


\bibitem[Filecoin({[n.\,d.]})]%
        {Filecoin}
\bibfield{author}{\bibinfo{person}{Filecoin}.}
  \bibinfo{year}{[n.\,d.]}\natexlab{}.
\newblock \bibinfo{howpublished}{\url{https://filecoin.io/}}.
\newblock


\bibitem[Fouque and Stern(2001)]%
        {FouqueS01}
\bibfield{author}{\bibinfo{person}{Pierre{-}Alain Fouque} {and}
  \bibinfo{person}{Jacques Stern}.} \bibinfo{year}{2001}\natexlab{}.
\newblock \showarticletitle{One Round Threshold Discrete-Log Key Generation
  without Private Channels}. In \bibinfo{booktitle}{\emph{Public Key
  Cryptography}} \emph{(\bibinfo{series}{LNCS}, Vol.~\bibinfo{volume}{1992})}.
  \bibinfo{publisher}{Springer}, \bibinfo{pages}{300--316}.
\newblock


\bibitem[Gao et~al\mbox{.}(2022)]%
        {GaoLLTXZ22}
\bibfield{author}{\bibinfo{person}{Yingzi Gao}, \bibinfo{person}{Yuan Lu},
  \bibinfo{person}{Zhenliang Lu}, \bibinfo{person}{Qiang Tang},
  \bibinfo{person}{Jing Xu}, {and} \bibinfo{person}{Zhenfeng Zhang}.}
  \bibinfo{year}{2022}\natexlab{}.
\newblock \showarticletitle{Efficient Asynchronous Byzantine Agreement without
  Private Setups}. In \bibinfo{booktitle}{\emph{{ICDCS}}}.
  \bibinfo{publisher}{{IEEE}}, \bibinfo{pages}{246--257}.
\newblock


\bibitem[Garay et~al\mbox{.}(2015)]%
        {GarayKL15}
\bibfield{author}{\bibinfo{person}{Juan~A. Garay}, \bibinfo{person}{Aggelos
  Kiayias}, {and} \bibinfo{person}{Nikos Leonardos}.}
  \bibinfo{year}{2015}\natexlab{}.
\newblock \showarticletitle{The Bitcoin Backbone Protocol: Analysis and
  Applications}. In \bibinfo{booktitle}{\emph{{EUROCRYPT} {(2)}}}
  \emph{(\bibinfo{series}{LNCS}, Vol.~\bibinfo{volume}{9057})}.
  \bibinfo{publisher}{Springer}, \bibinfo{pages}{281--310}.
\newblock


\bibitem[Gennaro et~al\mbox{.}(2007)]%
        {GennaroJKR07}
\bibfield{author}{\bibinfo{person}{Rosario Gennaro}, \bibinfo{person}{Stanislaw
  Jarecki}, \bibinfo{person}{Hugo Krawczyk}, {and} \bibinfo{person}{Tal
  Rabin}.} \bibinfo{year}{2007}\natexlab{}.
\newblock \showarticletitle{Secure Distributed Key Generation for Discrete-Log
  Based Cryptosystems}.
\newblock \bibinfo{journal}{\emph{J. Cryptol.}} \bibinfo{volume}{20},
  \bibinfo{number}{1} (\bibinfo{year}{2007}), \bibinfo{pages}{51--83}.
\newblock


\bibitem[Gentry et~al\mbox{.}(2021a)]%
        {GentryHKMNRY21}
\bibfield{author}{\bibinfo{person}{Craig Gentry}, \bibinfo{person}{Shai
  Halevi}, \bibinfo{person}{Hugo Krawczyk}, \bibinfo{person}{Bernardo Magri},
  \bibinfo{person}{Jesper~Buus Nielsen}, \bibinfo{person}{Tal Rabin}, {and}
  \bibinfo{person}{Sophia Yakoubov}.} \bibinfo{year}{2021}\natexlab{a}.
\newblock \showarticletitle{{YOSO:} You Only Speak Once - Secure {MPC} with
  Stateless Ephemeral Roles}. In \bibinfo{booktitle}{\emph{{CRYPTO} {(2)}}}
  \emph{(\bibinfo{series}{LNCS}, Vol.~\bibinfo{volume}{12826})}.
  \bibinfo{publisher}{Springer}, \bibinfo{pages}{64--93}.
\newblock


\bibitem[Gentry et~al\mbox{.}(2022)]%
        {GentryHL22}
\bibfield{author}{\bibinfo{person}{Craig Gentry}, \bibinfo{person}{Shai
  Halevi}, {and} \bibinfo{person}{Vadim Lyubashevsky}.}
  \bibinfo{year}{2022}\natexlab{}.
\newblock \showarticletitle{Practical Non-interactive Publicly Verifiable
  Secret Sharing with Thousands of Parties}. In
  \bibinfo{booktitle}{\emph{{EUROCRYPT} {(1)}}} \emph{(\bibinfo{series}{LNCS},
  Vol.~\bibinfo{volume}{13275})}. \bibinfo{publisher}{Springer},
  \bibinfo{pages}{458--487}.
\newblock


\bibitem[Gentry et~al\mbox{.}(2021b)]%
        {GentryHMNY21}
\bibfield{author}{\bibinfo{person}{Craig Gentry}, \bibinfo{person}{Shai
  Halevi}, \bibinfo{person}{Bernardo Magri}, \bibinfo{person}{Jesper~Buus
  Nielsen}, {and} \bibinfo{person}{Sophia Yakoubov}.}
  \bibinfo{year}{2021}\natexlab{b}.
\newblock \showarticletitle{Random-Index {PIR} and Applications}. In
  \bibinfo{booktitle}{\emph{{TCC} {(3)}}} \emph{(\bibinfo{series}{LNCS},
  Vol.~\bibinfo{volume}{13044})}. \bibinfo{publisher}{Springer},
  \bibinfo{pages}{32--61}.
\newblock


\bibitem[Gilad et~al\mbox{.}(2017)]%
        {GiladHMVZ17}
\bibfield{author}{\bibinfo{person}{Yossi Gilad}, \bibinfo{person}{Rotem Hemo},
  \bibinfo{person}{Silvio Micali}, \bibinfo{person}{Georgios Vlachos}, {and}
  \bibinfo{person}{Nickolai Zeldovich}.} \bibinfo{year}{2017}\natexlab{}.
\newblock \showarticletitle{Algorand: Scaling Byzantine Agreements for
  Cryptocurrencies}. In \bibinfo{booktitle}{\emph{{SOSP}}}.
  \bibinfo{publisher}{{ACM}}, \bibinfo{pages}{51--68}.
\newblock


\bibitem[Goldberg et~al\mbox{.}(2016)]%
        {GoldbergNPR16}
\bibfield{author}{\bibinfo{person}{Sharon Goldberg}, \bibinfo{person}{Moni
  Naor}, \bibinfo{person}{Dimitrios Papadopoulos}, {and}
  \bibinfo{person}{Leonid Reyzin}.} \bibinfo{year}{2016}\natexlab{}.
\newblock \showarticletitle{{NSEC5} from Elliptic Curves: Provably Preventing
  {DNSSEC} Zone Enumeration with Shorter Responses}.
\newblock \bibinfo{journal}{\emph{{IACR} Cryptol. ePrint Arch.}}
  (\bibinfo{year}{2016}), \bibinfo{pages}{83}.
\newblock


\bibitem[Groth(2023)]%
        {Groth21}
\bibfield{author}{\bibinfo{person}{Jens Groth}.}
  \bibinfo{year}{2023}\natexlab{}.
\newblock \bibinfo{title}{Non-interactive distributed key generation and key
  resharing}.
\newblock \bibinfo{howpublished}{Cryptology ePrint Archive, Paper 2021/339}.
\newblock
\newblock
\shownote{\url{https://eprint.iacr.org/2021/339}}.


\bibitem[Guo et~al\mbox{.}(2020)]%
        {GuoL0XZ20}
\bibfield{author}{\bibinfo{person}{Bingyong Guo}, \bibinfo{person}{Zhenliang
  Lu}, \bibinfo{person}{Qiang Tang}, \bibinfo{person}{Jing Xu}, {and}
  \bibinfo{person}{Zhenfeng Zhang}.} \bibinfo{year}{2020}\natexlab{}.
\newblock \showarticletitle{Dumbo: Faster Asynchronous {BFT} Protocols}. In
  \bibinfo{booktitle}{\emph{{CCS}}}. \bibinfo{publisher}{{ACM}},
  \bibinfo{pages}{803--818}.
\newblock


\bibitem[Gurkan et~al\mbox{.}(2021)]%
        {GurkanJMMST21}
\bibfield{author}{\bibinfo{person}{Kobi Gurkan}, \bibinfo{person}{Philipp
  Jovanovic}, \bibinfo{person}{Mary Maller}, \bibinfo{person}{Sarah
  Meiklejohn}, \bibinfo{person}{Gilad Stern}, {and} \bibinfo{person}{Alin
  Tomescu}.} \bibinfo{year}{2021}\natexlab{}.
\newblock \showarticletitle{Aggregatable Distributed Key Generation}. In
  \bibinfo{booktitle}{\emph{{EUROCRYPT} {(1)}}} \emph{(\bibinfo{series}{LNCS},
  Vol.~\bibinfo{volume}{12696})}. \bibinfo{publisher}{Springer},
  \bibinfo{pages}{147--176}.
\newblock


\bibitem[Itkis and Reyzin(2001)]%
        {ItkisR01}
\bibfield{author}{\bibinfo{person}{Gene Itkis} {and} \bibinfo{person}{Leonid
  Reyzin}.} \bibinfo{year}{2001}\natexlab{}.
\newblock \showarticletitle{Forward-Secure Signatures with Optimal Signing and
  Verifying}. In \bibinfo{booktitle}{\emph{{CRYPTO}}}
  \emph{(\bibinfo{series}{LNCS}, Vol.~\bibinfo{volume}{2139})}.
  \bibinfo{publisher}{Springer}, \bibinfo{pages}{332--354}.
\newblock


\bibitem[Kate et~al\mbox{.}(2023)]%
        {KateMMST23}
\bibfield{author}{\bibinfo{person}{Aniket Kate}, \bibinfo{person}{Easwar~Vivek
  Mangipudi}, \bibinfo{person}{Pratyay Mukherjee}, \bibinfo{person}{Hamza
  Saleem}, {and} \bibinfo{person}{Sri Aravinda~Krishnan Thyagarajan}.}
  \bibinfo{year}{2023}\natexlab{}.
\newblock \bibinfo{title}{Non-interactive {VSS} using Class Groups and
  Application to {DKG}}.
\newblock \bibinfo{howpublished}{Cryptology ePrint Archive, Paper 2023/451}.
\newblock
\newblock
\shownote{\url{https://eprint.iacr.org/2023/451}}.


\bibitem[Kate et~al\mbox{.}(2010)]%
        {KateZG10}
\bibfield{author}{\bibinfo{person}{Aniket Kate}, \bibinfo{person}{Gregory~M.
  Zaverucha}, {and} \bibinfo{person}{Ian Goldberg}.}
  \bibinfo{year}{2010}\natexlab{}.
\newblock \showarticletitle{Constant-Size Commitments to Polynomials and Their
  Applications}. In \bibinfo{booktitle}{\emph{{ASIACRYPT}}}
  \emph{(\bibinfo{series}{LNCS}, Vol.~\bibinfo{volume}{6477})}.
  \bibinfo{publisher}{Springer}, \bibinfo{pages}{177--194}.
\newblock


\bibitem[Kidron and Lindell(2011)]%
        {KidronL11}
\bibfield{author}{\bibinfo{person}{Dafna Kidron} {and} \bibinfo{person}{Yehuda
  Lindell}.} \bibinfo{year}{2011}\natexlab{}.
\newblock \showarticletitle{Impossibility Results for Universal Composability
  in Public-Key Models and with Fixed Inputs}.
\newblock \bibinfo{journal}{\emph{J. Cryptol.}} \bibinfo{volume}{24},
  \bibinfo{number}{3} (\bibinfo{year}{2011}), \bibinfo{pages}{517--544}.
\newblock


\bibitem[Komlo and Goldberg(2020)]%
        {KomloG20}
\bibfield{author}{\bibinfo{person}{Chelsea Komlo} {and} \bibinfo{person}{Ian
  Goldberg}.} \bibinfo{year}{2020}\natexlab{}.
\newblock \showarticletitle{{FROST:} Flexible Round-Optimized Schnorr Threshold
  Signatures}. In \bibinfo{booktitle}{\emph{{SAC}}}
  \emph{(\bibinfo{series}{LNCS}, Vol.~\bibinfo{volume}{12804})}.
  \bibinfo{publisher}{Springer}, \bibinfo{pages}{34--65}.
\newblock


\bibitem[Lee et~al\mbox{.}(2023)]%
        {LeeMDG23}
\bibfield{author}{\bibinfo{person}{Sung{-}Shine Lee}, \bibinfo{person}{Alexandr
  Murashkin}, \bibinfo{person}{Martin Derka}, {and} \bibinfo{person}{Jan
  Gorzny}.} \bibinfo{year}{2023}\natexlab{}.
\newblock \showarticletitle{SoK: Not Quite Water Under the Bridge: Review of
  Cross-Chain Bridge Hacks}. In \bibinfo{booktitle}{\emph{{ICBC}}}.
  \bibinfo{publisher}{{IEEE}}, \bibinfo{pages}{1--14}.
\newblock


\bibitem[Malkhi and Szalachowski(2022)]%
        {MalkhiS22}
\bibfield{author}{\bibinfo{person}{Dahlia Malkhi} {and} \bibinfo{person}{Pawel
  Szalachowski}.} \bibinfo{year}{2022}\natexlab{}.
\newblock \showarticletitle{Maximal Extractable Value {(MEV)} Protection on a
  {DAG}}. In \bibinfo{booktitle}{\emph{Tokenomics}}
  \emph{(\bibinfo{series}{OASIcs}, Vol.~\bibinfo{volume}{110})}.
  \bibinfo{publisher}{Schloss Dagstuhl - Leibniz-Zentrum f{\"{u}}r Informatik},
  \bibinfo{pages}{6:1--6:17}.
\newblock


\bibitem[Micali et~al\mbox{.}(1999)]%
        {MicaliRV99}
\bibfield{author}{\bibinfo{person}{Silvio Micali}, \bibinfo{person}{Michael~O.
  Rabin}, {and} \bibinfo{person}{Salil~P. Vadhan}.}
  \bibinfo{year}{1999}\natexlab{}.
\newblock \showarticletitle{Verifiable Random Functions}. In
  \bibinfo{booktitle}{\emph{{FOCS}}}. \bibinfo{publisher}{{IEEE} Computer
  Society}, \bibinfo{pages}{120--130}.
\newblock


\bibitem[Miller et~al\mbox{.}(2016)]%
        {MillerXCSS16}
\bibfield{author}{\bibinfo{person}{Andrew Miller}, \bibinfo{person}{Yu Xia},
  \bibinfo{person}{Kyle Croman}, \bibinfo{person}{Elaine Shi}, {and}
  \bibinfo{person}{Dawn Song}.} \bibinfo{year}{2016}\natexlab{}.
\newblock \showarticletitle{The Honey Badger of {BFT} Protocols}. In
  \bibinfo{booktitle}{\emph{{CCS}}}. \bibinfo{publisher}{{ACM}},
  \bibinfo{pages}{31--42}.
\newblock


\bibitem[Naor and Yung(1990)]%
        {NaorY90}
\bibfield{author}{\bibinfo{person}{Moni Naor} {and} \bibinfo{person}{Moti
  Yung}.} \bibinfo{year}{1990}\natexlab{}.
\newblock \showarticletitle{Public-key Cryptosystems Provably Secure against
  Chosen Ciphertext Attacks}. In \bibinfo{booktitle}{\emph{{STOC}}}.
  \bibinfo{publisher}{{ACM}}, \bibinfo{pages}{427--437}.
\newblock


\bibitem[Nayak et~al\mbox{.}(2020)]%
        {Nayak0SVX20}
\bibfield{author}{\bibinfo{person}{Kartik Nayak}, \bibinfo{person}{Ling Ren},
  \bibinfo{person}{Elaine Shi}, \bibinfo{person}{Nitin~H. Vaidya}, {and}
  \bibinfo{person}{Zhuolun Xiang}.} \bibinfo{year}{2020}\natexlab{}.
\newblock \showarticletitle{Improved Extension Protocols for Byzantine
  Broadcast and Agreement}. In \bibinfo{booktitle}{\emph{{DISC}}}
  \emph{(\bibinfo{series}{LIPIcs}, Vol.~\bibinfo{volume}{179})}.
  \bibinfo{publisher}{Schloss Dagstuhl - Leibniz-Zentrum f{\"{u}}r Informatik},
  \bibinfo{pages}{28:1--28:17}.
\newblock


\bibitem[Nielsen(2002)]%
        {Nielsen02}
\bibfield{author}{\bibinfo{person}{Jesper~Buus Nielsen}.}
  \bibinfo{year}{2002}\natexlab{}.
\newblock \showarticletitle{Separating Random Oracle Proofs from Complexity
  Theoretic Proofs: The Non-committing Encryption Case}. In
  \bibinfo{booktitle}{\emph{{CRYPTO}}} \emph{(\bibinfo{series}{LNCS},
  Vol.~\bibinfo{volume}{2442})}. \bibinfo{publisher}{Springer},
  \bibinfo{pages}{111--126}.
\newblock


\bibitem[Pedersen(1991)]%
        {Pedersen91}
\bibfield{author}{\bibinfo{person}{Torben~P. Pedersen}.}
  \bibinfo{year}{1991}\natexlab{}.
\newblock \showarticletitle{Non-Interactive and Information-Theoretic Secure
  Verifiable Secret Sharing}. In \bibinfo{booktitle}{\emph{{CRYPTO}}}
  \emph{(\bibinfo{series}{LNCS}, Vol.~\bibinfo{volume}{576})}.
  \bibinfo{publisher}{Springer}, \bibinfo{pages}{129--140}.
\newblock


\bibitem[Polkadot({[n.\,d.]})]%
        {Polkadot}
\bibfield{author}{\bibinfo{person}{Polkadot}.}
  \bibinfo{year}{[n.\,d.]}\natexlab{}.
\newblock \bibinfo{howpublished}{\url{https://www.polkadot.network/}}.
\newblock


\bibitem[Reed and Solomon(1960)]%
        {ReedS1960}
\bibfield{author}{\bibinfo{person}{Irving~S Reed} {and}
  \bibinfo{person}{Gustave Solomon}.} \bibinfo{year}{1960}\natexlab{}.
\newblock \showarticletitle{Polynomial codes over certain finite fields}.
\newblock \bibinfo{journal}{\emph{Journal of the society for industrial and
  applied mathematics}} \bibinfo{volume}{8}, \bibinfo{number}{2}
  (\bibinfo{year}{1960}), \bibinfo{pages}{300--304}.
\newblock


\bibitem[Ruby.Exchange(2021)]%
        {Skale}
\bibfield{author}{\bibinfo{person}{Ruby.Exchange}.}
  \bibinfo{year}{2021}\natexlab{}.
\newblock \bibinfo{title}{How SKALE Solves The Front-Running Problem}.
\newblock
  \bibinfo{howpublished}{\url{https://blog.ruby.exchange/how-skale-solves-the-front-running-problem/?ref=blog.pantherprotocol.io}}.
\newblock


\bibitem[Santis et~al\mbox{.}(2001)]%
        {SantisCOPS01}
\bibfield{author}{\bibinfo{person}{Alfredo~De Santis},
  \bibinfo{person}{Giovanni~Di Crescenzo}, \bibinfo{person}{Rafail Ostrovsky},
  \bibinfo{person}{Giuseppe Persiano}, {and} \bibinfo{person}{Amit Sahai}.}
  \bibinfo{year}{2001}\natexlab{}.
\newblock \showarticletitle{Robust Non-interactive Zero Knowledge}. In
  \bibinfo{booktitle}{\emph{{CRYPTO}}} \emph{(\bibinfo{series}{LNCS},
  Vol.~\bibinfo{volume}{2139})}. \bibinfo{publisher}{Springer},
  \bibinfo{pages}{566--598}.
\newblock


\bibitem[Shoup(2023)]%
        {Shoup23}
\bibfield{author}{\bibinfo{person}{Victor Shoup}.}
  \bibinfo{year}{2023}\natexlab{}.
\newblock \bibinfo{title}{The many faces of Schnorr}.
\newblock \bibinfo{howpublished}{Cryptology ePrint Archive, Paper 2023/1019}.
\newblock
\newblock
\shownote{\url{https://eprint.iacr.org/2023/1019}}.


\bibitem[Steinhoff et~al\mbox{.}(2021)]%
        {BMS21}
\bibfield{author}{\bibinfo{person}{Selma Steinhoff}, \bibinfo{person}{Chrysoula
  Stathakopoulou}, \bibinfo{person}{Matej Pavlovic}, {and}
  \bibinfo{person}{Marko Vukolic}.} \bibinfo{year}{2021}\natexlab{}.
\newblock \showarticletitle{{BMS:} Secure Decentralized Reconfiguration for
  Blockchain and {BFT} Systems}.
\newblock \bibinfo{journal}{\emph{CoRR}}  \bibinfo{volume}{abs/2109.03913}
  (\bibinfo{year}{2021}).
\newblock


\bibitem[Tas et~al\mbox{.}(2023)]%
        {TasTGKMY23}
\bibfield{author}{\bibinfo{person}{Ertem~Nusret Tas}, \bibinfo{person}{David
  Tse}, \bibinfo{person}{Fangyu Gai}, \bibinfo{person}{Sreeram Kannan},
  \bibinfo{person}{Mohammad~Ali Maddah{-}Ali}, {and} \bibinfo{person}{Fisher
  Yu}.} \bibinfo{year}{2023}\natexlab{}.
\newblock \showarticletitle{Bitcoin-Enhanced Proof-of-Stake Security:
  Possibilities and Impossibilities}. In \bibinfo{booktitle}{\emph{{SP}}}.
  \bibinfo{publisher}{{IEEE}}, \bibinfo{pages}{126--145}.
\newblock


\bibitem[Tezos({[n.\,d.]})]%
        {Tezos}
\bibfield{author}{\bibinfo{person}{Tezos}.}
  \bibinfo{year}{[n.\,d.]}\natexlab{}.
\newblock \bibinfo{howpublished}{\url{https://tezos.com}}.
\newblock


\bibitem[Tomescu et~al\mbox{.}(2020)]%
        {TomescuCZAPGD20}
\bibfield{author}{\bibinfo{person}{Alin Tomescu}, \bibinfo{person}{Robert
  Chen}, \bibinfo{person}{Yiming Zheng}, \bibinfo{person}{Ittai Abraham},
  \bibinfo{person}{Benny Pinkas}, \bibinfo{person}{Guy Golan{-}Gueta}, {and}
  \bibinfo{person}{Srinivas Devadas}.} \bibinfo{year}{2020}\natexlab{}.
\newblock \showarticletitle{Towards Scalable Threshold Cryptosystems}. In
  \bibinfo{booktitle}{\emph{{IEEE} Symposium on Security and Privacy}}.
  \bibinfo{publisher}{{IEEE}}, \bibinfo{pages}{877--893}.
\newblock


\bibitem[Total-blockchain(2022)]%
        {Osmosis}
\bibfield{author}{\bibinfo{person}{Total-blockchain}.}
  \bibinfo{year}{2022}\natexlab{}.
\newblock \bibinfo{title}{Osmosis will soon be frontrunning MEV free}.
\newblock
  \bibinfo{howpublished}{\url{https://medium.com/@totalblockchainemail/osmosis-will-soon-be-frontrunning-mev-free-b7da89f04ce9}}.
\newblock


\bibitem[Trautwein et~al\mbox{.}(2022)]%
        {TrautweinRTCSSG22}
\bibfield{author}{\bibinfo{person}{Dennis Trautwein}, \bibinfo{person}{Aravindh
  Raman}, \bibinfo{person}{Gareth Tyson}, \bibinfo{person}{Ignacio Castro},
  \bibinfo{person}{Will Scott}, \bibinfo{person}{Moritz Schubotz},
  \bibinfo{person}{Bela Gipp}, {and} \bibinfo{person}{Yiannis Psaras}.}
  \bibinfo{year}{2022}\natexlab{}.
\newblock \showarticletitle{Design and evaluation of {IPFS:} a storage layer
  for the decentralized web}. In \bibinfo{booktitle}{\emph{{SIGCOMM}}}.
  \bibinfo{publisher}{{ACM}}, \bibinfo{pages}{739--752}.
\newblock


\bibitem[Wolinsky et~al\mbox{.}(2012)]%
        {WolinskyCFJ12}
\bibfield{author}{\bibinfo{person}{David~Isaac Wolinsky},
  \bibinfo{person}{Henry Corrigan-Gibbs}, \bibinfo{person}{Bryan Ford}, {and}
  \bibinfo{person}{Aaron Johnson}.} \bibinfo{year}{2012}\natexlab{}.
\newblock \showarticletitle{Scalable anonymous group communication in the
  anytrust model}. In \bibinfo{booktitle}{\emph{European Workshop on System
  Security (EuroSec)}}, Vol.~\bibinfo{volume}{4}.
\newblock


\bibitem[Yurek et~al\mbox{.}(2022)]%
        {YurekLFKM22}
\bibfield{author}{\bibinfo{person}{Thomas Yurek}, \bibinfo{person}{Licheng
  Luo}, \bibinfo{person}{Jaiden Fairoze}, \bibinfo{person}{Aniket Kate}, {and}
  \bibinfo{person}{Andrew Miller}.} \bibinfo{year}{2022}\natexlab{}.
\newblock \showarticletitle{hbACSS: How to Robustly Share Many Secrets}. In
  \bibinfo{booktitle}{\emph{{NDSS}}}. \bibinfo{publisher}{The Internet
  Society}.
\newblock


\bibitem[Zhang et~al\mbox{.}(2022)]%
        {ZhangXHSZ22}
\bibfield{author}{\bibinfo{person}{Jiaheng Zhang}, \bibinfo{person}{Tiancheng
  Xie}, \bibinfo{person}{Thang Hoang}, \bibinfo{person}{Elaine Shi}, {and}
  \bibinfo{person}{Yupeng Zhang}.} \bibinfo{year}{2022}\natexlab{}.
\newblock \showarticletitle{Polynomial Commitment with a One-to-Many Prover and
  Applications}. In \bibinfo{booktitle}{\emph{{USENIX} Security Symposium}}.
  \bibinfo{publisher}{{USENIX} Association}, \bibinfo{pages}{2965--2982}.
\newblock


\end{thebibliography}
